\documentclass[12pt,reqno] {amsart}

 \usepackage{amsmath, wasysym}
 \usepackage{bbm, dsfont}
 \usepackage{amssymb}
\usepackage{graphicx}%
\usepackage{hyperref}
\usepackage{cleveref}
\usepackage{comment}
 \usepackage{mathrsfs}
 \usepackage{float}
 \usepackage{tikz, tikz-cd}
  \usetikzlibrary{arrows.meta, knots, calc}
 \usepackage{bm}
 \usepackage{mathtools} 
\usepackage{extarrows}

 \usepackage{enumerate}
  \usepackage{mathdots}
\usetikzlibrary{decorations.markings}
\usetikzlibrary{decorations.pathreplacing}

\theoremstyle{plain}
\newtheorem{theorem}{Theorem}[section]

\newtheorem{corollary}[theorem]{Corollary}
\newtheorem{lemma}[theorem]{Lemma}
\newtheorem{proposition}[theorem]{Proposition}
\newtheorem*{ac}{Acknowledgement}

\newtheorem{remark}[theorem]{Remark}

\newtheorem{definition}[theorem]{Definition}

\newtheorem{notation}{Notation}[section]

\newtheorem{example}[theorem]{Example}

\newcommand{\cM}{\mathcal{M}}
\newcommand{\cN}{\mathcal{N}}
\newcommand{\cL}{\mathcal{L}}
\newcommand{\cJ}{\mathcal{J}}

\newcommand{\cR}{\mathcal{R}}

\newcommand{\cH}{\mathcal{H}}
\newcommand{\sI}{\mathscr{I}}
\newcommand{\bC}{\mathbb{C}}
\newcommand{\bZ}{\mathbb{Z}}
\newcommand{\bE}{\mathbb{E}}
\newcommand{\bN}{\mathbb{N}}
\newcommand{\bR}{\mathbb{R}}

\newcommand{\Ker}{\mathrm{Ker}}
\newcommand{\K}{\mathbf{K}}
\newcommand{\CS}{\mathfrak{CS}}
\newcommand{\Div}{\mathrm{Div}}
\newcommand{\Ran}{\mathrm{Ran }}
\newcommand{\Dom}{\mathrm{Dom }}
\newcommand{\Tr}{\mathrm{Tr}}
\newcommand{\grad}{\mathrm{grad}}
\newcommand{\fix}{\mathscr{M}}
\newcommand{\nfix}{\mathscr{N}}

\newcommand{\ThinLine}{\draw[line width=1pt]}

\newcommand{\FillGrey}{\path[fill=gray!40!white]}
\newcommand{\FillWhite}{\path[fill=white]}

\begin{document}

\title{Bimodule Quantum Markov Semigroups}

\author{Jinsong Wu}
\address{Jinsong Wu, Beijing Institute of Mathematical Sciences and Applications, Beijing, 101408, China}
\email{wjs@bimsa.cn}

\author{Zishuo Zhao}
\address{Zishuo Zhao, Tsinghua University, Beijing, 100084, China}
\email{zzs21@mails.tsinghua.edu.cn}


\date{}

\begin{abstract}
We present a systematic investigation of bimodule quantum Markov semigroups within the framework of quantum Fourier analysis. 
We introduce the concepts of bimodule detailed balance conditions and bimodule KMS symmetry, which not only generalize the classical notions of detailed balance but also expose interesting structures of quantum channels. 
We demonstrate that the evolution of densities governed by the bimodule quantum Markov semigroup with bimodule detailed balance is the gradient flow for the relative entropy with respect to the hidden density. 
Consequently, we obtain a modified logarithmic Sobolev inequality and a Talagrand inequality with respect to a hidden density from higher dimensional structure. 
Furthermore, we establish a Poincar\'{e} inequality for irreducible inclusions and relative ergodic bimodule quantum semigroups.
\end{abstract}

\maketitle

\section{Introduction} 

Classical Markov semigroups, such as those governing heat flow, play a fundamental role in harmonic analysis. Several important inequalities, including Young's inequality, the entropy power inequality, and the logarithmic Sobolev inequality, can be derived using the heat flow method.

In quantum mechanics, states of a system is represented by a Hilbert space, while observables correspond to self-adjoint operators. 
The system of observables is described by von Neumann algebras. 
A Markov semigroup acting on von Neumann algebras is commonly referred to as a quantum Markov semigroup.
The quantum Markov semigroup \cite{Lin76,Fri78, FriVer82} is a powerful tool in quantum statistical mechanics for modeling open quantum systems. 
It also plays a crucial role in noncommutative analysis, noncommutative probability and noncommutative geometry. 

In quantum statistical mechanics, an open system interacts with a heat flow in thermal equilibrium, which mathematically corresponds to an equilibrium state.
Due to the noncommutativity of the setting, the symmetries of the heat flow relative to the equilibrium state are more intricate than in the classical case.
Two fundamental examples of such symmetries are Gelfand-Naimark-Segal (GNS) symmetry \cite{Wirth2022a, wirth2022b} and Kubo-Martin-Schwinger (KMS) symmetry \cite{FagReb15,FagUma08, FagUma10, KFGV77, VerWir23}. 

Modern subfactor theory was pioneered by Vaughan Jones \cite{Jon83, Jon85}.
His work on subfactors revealed a wealth of unexpected symmetries beyond classical group symmetry, now known as quantum symmetries.
The axiomatic characterization of subfactors includes Ocneanu's paragroups \cite{Ocn88}, Popa's $\lambda$-lattices\cite{Popa95}, and Jones' planar algebras \cite{Jones2021}. 
Among these, Jones' planar algebras provide a topological framework for studying quantum symmetries, consisting of a sequence of $n$-box spaces.
Quantum Markov semigroups, when encoded by quantum symmetries, exhibit highly intricate and fascinating structural properties.

In this paper, we investigate quantum Markov semigroups on a finite inclusion $\cN\subset\cM$ of finite von Neumann algebras that preserve $\cN$, referred to as bimodule quantum Markov semigroups.
Assuming that the inclusion is a $\lambda$-extension\cite{PimPop86}, we leverage the computational advantages of Jones' planar algebras. 
We explore equilibrium, GNS symmetry, and KMS symmetry in the bimodule setting.
Inspired by the quantum Fourier analysis developed by Jaffe, Jiang, Liu, Ren, and Wu \cite{JJLRW20}, we introduce the notions of bimodule GNS symmetry, and bimodule KMS symmetry within the framework of quantum Fourier analysis (See Theorems \ref{thm:equivbalance}, \ref{thm:kms}).
These bimodule symmetries significantly broaden the study of symmetric quantum channels and quantum semigroups.

The bimodule quantum Markov semigroups can be fully characterized by $\mathfrak{F}$ positive elements, as introduced by Huang, Jaffe, Liu, and Wu\cite{HJLW23}, in the 2-box space, together with self-adjoint elements in the 1-box space (see Proposition \ref{prop:generatorformula} and Theorem \ref{thm:generatorform}). 
The derivation associated with a bimodule quantum Markov semigroup resides in the 3-box space.
Utilizing this characterization of the derivations, we establish the Poincaré inequality for irreducible subfactors (Theorem \ref{thm:poincare1}).
For bimodule GNS symmetric and relatively ergodic Markov semigroups, we analyze the limit of the semigroup and derive the equation of the gradient flow. 
Additionally, we introduce the concept of hidden density, obtained by projecting elements from the 2-box space into the 1-box space.
This additional structure in the bimodule semigroup framework allows us to establish both the bimodule logarithmic Sobolev inequality (Theorem \ref{thm:entropydecay}) under the intertwining property introduced in\cite{CarMaa17}, and the bimodule Talagrand inequality (Theorem \ref{thm:talagrand}). 
We shall investigate the bimodule KMS symmetric semigroups in the future.


The rest of the paper is organized as follows. 
In Section \ref{sec:: Preliminaries}, we recall the $\lambda$-extension of finite von Neumann algebras, providing examples of finite inclusions and demonstrating that our work includes matrix cases. 
Section \ref{sec:: Bimodule quantum channels} focuses on bimodule quantum channels and their Fourier multipliers. 
Section \ref{sec:: Bimodule Equilibrium and Bimodule GNS Symmetry} introduces the notions of bimodule GNS symmetry, and bimodule KMS symmetry, showing how they naturally extend from GNS symmetry and KMS symmetry. 
In Section \ref{sec:: Bimodule Quantum Markov Semigroups}, we study bimodule quantum Markov semigroups within the framework of quantum Fourier analysis and introduce the derivation in 3-box spaces. 
We also prove the Poincaré inequality for irreducible subfactors. 
In Section \ref{sec:: Gradient Flow}, we analyze the gradient flow of bimodule GNS symmetric and relatively ergodic Markov semigroups, 
introducing the concept of hidden density within the semigroup. 
Using this, we derive the logarithmic Sobolev inequality and Talagrand inequality with respect to the hidden density.

\begin{ac}
J. Wu was supported by grants from Beijing Institute of Mathematical Sciences and Applications.
J.~W. was supported by NSFC (Grant no. 12371124) and partially supported by NSFC (Grant no. 12031004). 
Z. Zhao was supported by Beijing Natural Science Foundation (Grant No. Z221100002722017). 
\end{ac}

\section{Preliminaries}\label{sec:: Preliminaries}

Let $\cN\subset \cM$ be an unital inclusion of finite von Neumann algebras and $\tau$ be a normal faithful tracial state on $\cM$. 
Denote by $L^2(\cM, \tau)$ the GNS Hilbert space of $\tau$, with cyclic separating vector $\Omega$ and modular conjugation $J$ given by $Jx\Omega = x^*\Omega$ for all $x\in\cM$. 
Let $e_1$ be the Jones projection from $L^2(\cM, \tau)$ onto $L^2(\cN, \tau)$ and $\bE_{\cN}$ the $\tau$-preserving conditional expectation of $\cM$ onto $\cN$.  
We have that $e_1 x e_1 =\bE_{\cN}(x)e_1$ for all $x\in\cM$. 
The basic construction $\cM_1=\langle \cM, e_1\rangle$ is the von Neumann algebra generated by $\cM$ and $e_1$. 
The inclusion $\cN\subset\cM$ is called finite if $\cM_1$ is a finite von Neumann algebra, and irreducible if $\cN'\cap \cM=\mathbb{C}$, where $\cN'$ is the commutant of $\cN$ on $L^2(\cM, \tau)$.  
We have $\cM_1 = J\cN'J$. 
We focus on the finite inclusions of finite von Neumann algebras in this paper. 

Suppose $\tau_1$ is a faithful normal trace on $\cM_1$ extending $\tau$, and let $\mathbb{E}_{\cM}$ be the trace-preserving conditional expectation onto $\cM$. 
The pair $(\cM_1, \tau_1)$ is called a $\lambda$-extension of $\cN\subset\cM$ if $\tau_1|_{\cM}=\tau$ and $\bE_{\cM}(e_1)=\lambda$ for some positive constant $\lambda$. 
The index of the extension is defined as $[\cM:\cN] = \lambda^{-1}$. 
We denote by $\Omega_1$ the cyclic and separating vector in $L^2(\cM_1, \tau_1)$, and by $e_2$ the Jones projection of $L^2(\cM_1, \tau_1)$ onto $L^2(\cM,\tau_1)$. 
The $\lambda$-extension is called extremal if $\tau_1(x) = \tau_1(Jx^*J)$ for all $x\in\cN'\cap \cM$. 
In this paper, $\lambda$-extensions are always assumed to be extremal. 
We denote $\mathbb{E}_{\cN'}$ to be the $\tau_1$-preserving conditional expectation from $\cM_1$ onto $\cN'\cap \cM_1$. 

Let $\cM_2 = \langle \cM_1,e_2\rangle$ be the basic construction of the inclusion $\cM\subset \cM_1$, with a normal faithful trace $\tau_2$ extending $\tau_1$. 
We assume $\cM_1\subset \cM_2$ is a $\lambda$-extension of $\cM\subset\cM_1$, i.e. $\mathbb{E}_{\cM_1}(e_2) = \lambda$, where $\mathbb{E}_{\cM_1}$ is the $\tau_2$-preserving conditional expectation onto $\cM_1$. 
Denote by $\mathbb{E}_{\cM'}$ the $\tau_2$-preserving conditional expectation from $\cM_2$ onto $\cM'\cap \cM_2$.

In the following, we assume that $\cN \subset \cM\subset \cM_1$ are $\lambda$-extensions.
Then there exists a finite set $\{\eta_j\}_{j=1}^m$ of operators in $\cM$ called Pimsner-Popa basis for $\cN\subset \cM$, which satisfies $\displaystyle x=\sum_{j=1}^m  \bE_{\mathcal{N}}(x\eta^*_j)\eta_j$, for all $x\in\cM$\cite{PimPop86}. 
In terms of the Jones projection $e_1$, this condition is expressed as $\displaystyle \sum_{j=1}^m\eta_j^*e_1\eta_j =1$. 
This implies that any operator in $\cM_1$ is a finite sum of operators of the form $ae_1b$ with $a,b\in \cM$. 
As a consequence, for any $y\in \cM_1$, there is a unique $x\in\cM$ such that $ye_1 = xe_1$.
This indicates that $x=\lambda^{-1}\bE_{\cM}(ye_1)$.
Moreover $\displaystyle \sum_{j=1}^m\eta_j^*\eta_j =\lambda^{-1}$. 
We shall assume that the basis is orthogonal, that is $\bE_{\cN}(\eta_k\eta_j^*)=0$ for $k\neq j$. 
The conditional expectation $\bE^{\cN'}_{\cM'}(=\bE_{\cM'})$ from $\cN' = J\cM_1 J$ onto $\cM' = J\cM J$ can be written as 
\begin{align*}
    \bE^{\cN'}_{\cM'} (x) = \lambda\sum^m_{j=1} \eta^*_j x \eta_j,\quad x\in \cM_2.
\end{align*}
Note that this implies that $\mathbb{E}_{\cM'}(e_1) = \lambda$. 
We also have $\mathbb{E}_{\cM'}(yx) = \mathbb{E}_{\cM'}(xy)$ for all $y\in \cM$ and $x\in \cM_2$. 

The Pimsner-Popa inequality \cite{PimPop86} for the inclusion states that $\bE_{\cN}(x)\geq \lambda_{\cN\subset\cM} x$ for any $0\leq x\in \cM$, where $\lambda_{\cN\subset\cM}$ is the Pimsner-Popa constant. 

The basic construction from a $\lambda$-extension is assumed to be iterated to produce the Jones tower 
\begin{align*}
\cN\subset\cM\subset\cM_1\subset\cM_2\subset\cdots.
\end{align*}
The sequence of higher relative commutants consists the standard invariant of the inclusion initial. 
The standard invariants are axiomatized by planar algebras in \cite{Jones2021}.

\subsection{Fourier Transform} 

The Fourier transform $\mathfrak{F}: \cN'\cap \cM_1 \to \cM'\cap \cM_2$ and the inverse Fourier transform $\mathfrak{F}^{-1}$ are defined as 
\begin{align}\label{eqn:: Fourier transform}
    \mathfrak{F}(x)=&\lambda^{-3/2}\bE_{\cM'}(xe_2e_1), \quad x\in \cN'\cap \cM_1.\\
    \mathfrak{F}^{-1}(x) =& \lambda^{-3/2}\bE_{\cM_1}(xe_1e_2), \quad x\in \cM'\cap \cM_2. 
\end{align}
We check that for any $x\in \cN'\cap \cM_1$:
\begin{align*}
    \mathfrak{F}^{-1}(\mathfrak{F}(x))
    &= \mathfrak{F}^{-1}(\lambda^{-3/2}\lambda\sum_{j=1} ^m\eta^*_jxe_2e_1\eta_j)\\
    &= \mathfrak{F}^{-1}(\lambda^{-1/2}\sum_{j=1} ^m \eta^*_jxe_2e_1\eta_j)\\
    &= \lambda^{-3/2}\lambda^{-1/2}\sum_{j=1} ^m \mathbb{E}_{\cM_1}(\eta^*_jxe_2e_1\eta_je_1e_2)\\
    &= \lambda^{-2}\sum_{j=1} ^m\mathbb{E}_{\cM_1}(\eta^*_j\mathbb{E}_{\cN}(\eta_j)xe_2e_1e_2)\\
    &= \lambda^{-2}x\mathbb{E}_{\cM_1}(\lambda e_2) = x.
\end{align*}
It is then readily checked that $\mathfrak{F}$ satisfies the Plancherel identity: $\tau_{2}(\mathfrak{F}(x)^*\mathfrak{F}(x)) = \tau_{1}(x^*x)$, for all $x\in \cN'\cap \cM_1$.  
\begin{align*}
    \lambda^{-3}\tau_2(\bE_{\cM'}(e_1e_2x^*)\bE_{\cM'}(xe_2e_1)) &= \lambda^{-2}\sum_{j=1}^m \tau_2(e_1e_2 x^*\eta^*_jxe_2e_1\eta_j)\\
    &= \lambda^{-2}\sum_{j=1}^m \tau_2(e_2 x^*\eta^*_jxe_2\bE_{\cN}(\eta_j)e_1)\\
    &= \lambda^{-2}\sum_{j=1}^m \tau_2(e_2 x^*\eta^*_j\bE_{\cN}(\eta_j)xe_2e_1)\\
    &= \lambda^{-2}\tau_2(x^*xe_2e_1e_2) = \tau_1(x^*x). 
\end{align*}
For irreducible inclusions, we have the Hausdorff-Young inequality, which states $\|\mathfrak{F}(x)\|_\infty\leq \lambda^{-1/2}\|x\|_1$, where $\|x\|_1=\tau_1(|x|)$. 

In the planar algebra of the inclusion, the Fourier transform from $\cN'\cap \cM_1$ to $\cM'\cap \cM_2$ is represented by the $90$-degree clockwise rotation: 
\begin{align*}
\mathfrak{F}(x):=\raisebox{-0.9cm}{
\begin{tikzpicture}[scale=1.5]
\path [fill=gray!40] (-0.3, -0.4) rectangle (0.8, 0.9);
\path [fill=white] (0.35, -0.4)--(0.35, 0.5) .. controls +(0, 0.3) and +(0, 0.3) .. (0.65, 0.5)--(0.65, -0.4);
\path[fill=white] (0.15, 0.9) -- (0.15, 0) .. controls +(0, -0.3) and +(0, -0.3) .. (-0.15, 0)--(-0.15, 0.9);
\draw [blue, fill=white] (0,0) rectangle (0.5, 0.5);
\node at (0.25, 0.25) {$x$};
\draw (0.35, 0)--(0.35, -0.4) (0.15, 0.5)--(0.15, 0.9);
\draw (0.35, 0.5) .. controls +(0, 0.3) and +(0, 0.3) .. (0.65, 0.5)--(0.65, -0.4);
\draw (0.15, 0) .. controls +(0, -0.3) and +(0, -0.3) .. (-0.15, 0)--(-0.15, 0.9);
\end{tikzpicture}}
\;.
\end{align*}
There is a Fourier transform from $\cM'\cap\cM_2$ to $\cN'\cap\cM_1$. 
We shall denote it and its inverse also by $\mathfrak{F}$ and $\mathfrak{F}^{-1}$ respectively: 
\begin{align}\label{eqn:: Fourier transform on the dual}
    \mathfrak{F}(y) &= \lambda^{-3/2}\bE_{\cM_1}(e_2e_1y), \quad y\in \cM'\cap \cM_2. \\
    \mathfrak{F}^{-1}(y) &= \lambda^{-3/2}\bE_{\cM'}(e_1e_2y), \quad y\in \cN'\cap \cM_1. 
\end{align}
The same planar tangle with opposite shading represents the Fourier transform from $\cM'\cap \cM_2$ to $\cN'\cap \cM_1$. 

Composing the two types of Fourier transform results in a 180-degree rotation on $\cN'\cap\cM_1$, which is related to the modular conjugation $J$ as follows: 
\begin{align}
    \mathfrak{F}^2(x) = Jx^*J,\quad x\in\cN'\cap \cM_1. 
\end{align}
To see this, we first have: 
\begin{align*}
    \mathfrak{F}^2(x) &= \lambda^{-3}\bE_{\cM_1}(e_2e_1\bE_{\cM'}(xe_2e_1))\\
    &= \lambda^{-2}\sum_{j=1}^m \bE_{\cM_1}(e_2e_1 \eta^*_jxe_2e_1\eta_j)\\
    &= \lambda^{-2}\sum_{j=1}^m \bE_{\cM_1}(e_2e_1 \eta^*_jxe_2)e_1\eta_j
\end{align*}
Next notice that the action of the operator $J(ae_1b)^*J\in J\cM_1J$ restricted to $\cM\Omega$ is given by $J(ae_1b)^*Jz\Omega = \mathbb{E}_{\cM}(e_1zae_1b)\Omega$, where $z\in\cM$. 
Therefore we have 
\begin{align*}
    \lambda^{-2}\bE_{\cM_1}(e_2e_1 \eta^*_jxe_2)\Omega = \lambda^{-1}\bE_{\cM_1}(\mathbb{E}_{\cM}(e_1 \eta^*_jx)e_2)\Omega = Jx^*J\eta^*_j\Omega,
\end{align*}
consequently 
\begin{align*}
    \lambda^{-2}\sum_{j=1}^m \bE_{\cM_1}(e_2e_1 \eta^*_jxe_2)e_1\eta_j = Jx^*J. 
\end{align*}
We call $Jx^*J$ the modular conjugation of $x$, and denote it by $\overline{x}$. 
Since $\overline{\overline{x}} = x$, we see that $\mathfrak{F}^4 = id$. 
There is a similar relation between the modular conjugation on and the square of Fourier transform on $\cM'\cap \cM_2$. 

\begin{lemma}\label{lem:fourier}
    For any $a\in \cN'\cap \cM_1$ and $x, y\in \cM$, we have 
    \begin{align}
        \mathfrak{F}(a)xe_1y \Omega_1=\lambda^{1/2} x a y\Omega_1.
    \end{align}
    Moreover, we have $\mathfrak{F}(a)e_1e_2=\lambda^{1/2} ae_2$.
\end{lemma}
\begin{proof}
We have that 
\begin{align*}
 \mathfrak{F}(a)xe_1y \Omega_1
 =& \lambda^{-3/2}\bE_{\cM'}(ae_2e_1) xe_1 y\Omega_1\\
  =& \sum^m_{j=1} \lambda^{-1/2}x \eta^*_j ae_2e_1\eta_j e_1 y\Omega_1\\
  =& \sum^m_{j=1} \lambda^{-1/2}x \eta^*_j\mathbb{E}_{\cN}(\eta_j) ae_2e_1 y\Omega_1\\
    =& \lambda^{-1/2}x ae_2e_1 e_2 y\Omega_1\\
    =& \lambda^{1/2}x a y\Omega_1.
\end{align*}
Let $x=y=1$.
We obtain that $\mathfrak{F}(a)e_1\Omega_1=\lambda^{1/2} a\Omega_1$.
This implies that $\mathfrak{F}(a)e_1e_2\Omega_1=\lambda^{1/2} ae_2\Omega_1$.
Hence $\mathfrak{F}(a)e_1e_2=\lambda^{1/2} ae_2$.
This completes the proof of the lemma. 
\end{proof}

The rotation $\Theta: \cM_1\to \cM_1$ is defined as
\begin{align*}
    \Theta(x)=\lambda^{-3}\bE_{\cM_1}(e_2e_1\bE_{\cM'}(xe_2e_1)), \quad x\in \cM_1.
\end{align*}
Note that $\Theta|_{\cN'\cap \cM_1}=\mathfrak{F}^2$.

We define the convolution between $x,y\in \cM'\cap \cM_2$ as 
\begin{align*}
    x*y =& \mathfrak{F}^{-1}(\mathfrak{F}(y)\mathfrak{F}(x))\\
    =& \lambda^{-9/2}\mathbb{E}_{\cM'}(e_1e_2\mathbb{E}_{\cM_1}(e_2e_1y)\mathbb{E}_{\cM_1}(e_2e_1x)). 
\end{align*}
The convolution admits a simple graphical representation through the planar algebra: 
\begin{align*}
        x*y = \vcenter{\hbox{\begin{tikzpicture}
            \begin{scope}[scale=0.6]
        \FillWhite (-1.3,1.5) -- (-1.3, -0.5) -- (1.3, -0.5) -- (1.3, 1.5) -- (-1.3,1.5);
        \ThinLine (1.3, -0.5) -- (1.3, 1.5);
        \ThinLine (-1.3,-0.5) -- (-1.3,1.5);
        \ThinLine (0.6,1) .. controls +(0,0.4) and +(0,0.4) .. (-0.6,1);
        \ThinLine (0.6,0) .. controls +(0,-0.4) and +(0,-0.4) .. (-0.6,0);
        \FillGrey (0.6,1) .. controls +(0,0.4) and +(0,0.4) .. (-0.6,1)
        -- (-0.6,0) -- (-0.6,0) .. controls +(0,-0.4) and +(0,-0.4) .. (0.6,0)
        -- (0.6,1);
        \FillGrey (1.3,-0.5) -- (1.9,-0.5) -- (1.9,1.5) -- (1.3,1.5);
        \FillGrey (-1.9,-0.5) -- (-1.3,-0.5) -- (-1.3,1.5) -- (-1.9,1.5) -- (-1.9,-0.5);
        \draw[blue, fill=white] (0.3,0) rectangle (1.6,1);
        \node at (0.95,0.5) {$y$};
        \draw[blue, fill=white] (-1.6,0) rectangle (-0.3,1);
        \node at (-0.95,0.5) {$x$};
    \end{scope}
    \end{tikzpicture}}}.
    \end{align*}
    The Schur product theorem states $x*y\geq 0$ provided that $x,y\geq0$. 
    If the inclusion is irreducible, then we have Young's inequality: $\Vert x*y\Vert_r\leq \lambda^{-1/2}\Vert x\Vert_p \Vert y\Vert_q$, for $r^{-1} +1= p^{-1} + q^{-1}$, $p,q,r \geq 1$. 

The shift $\gamma_{1, +}: \cM_1'\cap \cM_3 \to \cN'\cap \cM_1$ is an isomorphism given by
\begin{align*}
\gamma_{1, +}(x)e_3=\lambda^{-2}e_3e_2e_1 x e_1e_2 e_3, \quad x\in \cM_1'\cap \cM_3.
\end{align*}
The inverse $\gamma_{1, +}^{-1}: \cN'\cap \cM_1\to \cM_1'\cap \cM_3$ is given by
\begin{align*}
\gamma_{1, +}^{-1}(x)e_1=\lambda^{-2}e_1e_2 e_3 x e_3e_2e_1, \quad x\in \cN'\cap \cM_1.
\end{align*}
    
\subsection{The Inclusion $\bC\subset \bC^n$} 

Let $\cN=\bC$, $\cM=\bC^n$.
Let $\{E_k\}^n_{k=1}$ be $n$ distinct minimal projections in $\mathbb{C}^n$, we define the normalized trace $\displaystyle \tau(E_k) = \frac{1}{n}$ for $k=1, \ldots, n$. 
The unit of $\cN$ is identified with that of $\cM$, which is $\displaystyle \sum^n_{k=1} E_k$. 


We can identify $L^2(\cM)$ with $\mathbb{C}^n$ under the correspondence $E_k\Omega\mapsto \begin{bmatrix}
    \underbrace{0  \cdots 0}_{k-1} &1 & 0\cdots   0
\end{bmatrix}^{\mathsf{T}}$, where $\mathsf{T}$ stands for the transpose. 
The left regular representation of $\cM$ on $L^2(\cM)$ is given as 
\begin{align*}
    E_k\mapsto \begin{bmatrix}
        0 & & & &\\
        & \ddots & & &\\
        & & 1 &\\
        & & & \ddots &\\
        & & & & 0
    \end{bmatrix}.
\end{align*}
Define $E_{j,k}$ as $E_{j,k}E_l\Omega = \delta_{k,l}E_j\Omega$ for $1\leq j,k\leq n$. 
The set $\{E_{j,k}\}^n_{j,k=1}$ forms a system of matrix units of $\mathcal{B}(L^2(\cM)) = M_n(\mathbb{C})$, in which we have $E_j = E_{j,j}$. 
The Jones projection is $\displaystyle e_1 = \frac{1}{n}\sum^n_{j,k=1}E_{j,k}$, whose image is spanned by the vector $\begin{bmatrix}
    1 & \cdots & 1
\end{bmatrix}$. 
Note that we have 
\begin{align*}
    \frac{1}{n}E_{j,k} = E_{j,j} e_1E_{k,k}. 
\end{align*}
The basic construction $\cM_1$ is the algebra generated by $\cM$ and $e_1$, i.e. the algebra generated by the algebra of diagonal matrices and $e_1$. 
We have $\cM_1 = M_n(\mathbb{C})$. 
Hence $\cN'\cap \cM = \mathbb{C}^n$, and $\cN'\cap \cM_1 = M_{n}(\mathbb{C})$. 
The modular conjugation $J$ on $\cN'\cap \cM_1$ satisfies $JX^*J = X^{\mathsf{T}}$, $X\in M_n(\mathbb{C})$. 
Note that $\cM_1=J\cN'J$. 
We see that $\cM_1=M_n(\bC)$ directly by the fact that $\cN=\bC$.

The trace $\tau_1$ is the unique normalized trace on $M_n(\mathbb{C})$. 
The conditional expectation $\mathbb{E}_{\cM}$ on $\cM_1$ is 
\begin{align*}
\mathbb{E}_{\cM}(X) = \sum^n_{k=1}E_{k,k} XE_{k,k}.
\end{align*}
We have $\displaystyle \mathbb{E}_{\cM}(e_1) = \frac{1}{n}$, therefore $\cN\subset\cM$ is a $\displaystyle \frac{1}{n}$-extension. 
We have a natural choice of Pimsner-Popa basis of $\cN\subset \cM$ given by $\{\sqrt{n}E_{k, k}\}^n_{k=1}$ subject to the condition 
\begin{align*}
\sum^n_{k=1}\sqrt{n}E_{k,k}e_1\sqrt{n}E_{k,k} = 1.
\end{align*}

The GNS Hilbert space $L^2(\cM_1)$ has a basis $\{E_{j,k}\Omega_1\}^n_{j,k=1}$, with the left action of $\cM_1$ as $E_{t,j}E_{k,l}\Omega_1 = \delta_{j,k}E_{t,l}\Omega_1$. 
Define operators $E_{(j,k),(p,q)}$ as $E_{(j,k),(p,q)}E_{p',q'}\Omega_1 = \delta_{(p,q),(p',q')}E_{j,k}\Omega_1$. 
The element $E_{j,k}\in \cM_1$ embedded in $\cM_2$ is $\displaystyle \sum_{l=1}^n E_{(j, l),(k,l)}$. 
The Jones projection on $L^2(\cM_1)$ is $\displaystyle e_2= \sum^n_{k=1}E_{(k,k),(k,k)}$. 
Then $\cM_2$ is the algebra generated by $\cM_1$ and $e_2$. 
By identifying $L^2(\cM_1)$ with $\mathbb{C}^n\otimes\mathbb{C}^n$ under the unitary transformation 
\begin{align*}
E_{j,k}\Omega_1\mapsto \sqrt{n}E_j\Omega\otimes E_k\Omega,
\end{align*}
the left/right action of $\cM$ are identified with the action on the first/second tensor factor. 

The left action of $\cM_1$ on $\mathbb{C}^n\otimes\mathbb{C}^n$ is given by $E_{j,k}\otimes I$. 
The action of $e_2$ on $\mathbb{C}^n\otimes \mathbb{C}^n$ is given by the projection $\displaystyle \sum^n_{k=1}E_{k,k}\otimes E_{k,k}$. 
From this we deduce $\cM_2 \cong M_n(\mathbb{C})\otimes \mathbb{C}^n$. 
The trace $\tau_2$ on $\cM_2$ is given as $\displaystyle \tau_2(E_{(j,l),(k,l)}) = \frac{1}{n^2}\delta_{j,k}$. 
The inclusion $\cM_1\subset \cM_2$ is again a $\displaystyle \frac{1}{n}$-extension. 
We have 
\begin{align*}
    \cM'\cap \cM_2 = \text{span}\{E_{j,j}\otimes E_{k,k}\vert 1\leq j,k\leq n\}. 
\end{align*}
We remark that $\cM'\cap \cM_2$ is a commutative $C^*$-algebra. 
 
 A particular basis of $\cM'\cap \cM_2$ is obtained from Fourier transforming the system of matrix units of $\cN'\cap \cM_1$. 
 For $1\leq j,k\leq n$, we have
 \begin{align*}
     \mathfrak{F}(E_{j,k})= \sqrt{n}E_{j,j}\otimes E_{k,k}. 
 \end{align*}
 The multiplication of matrices is dual to the Schur product of matrices under the Fourier transform $\mathfrak{F}$. 

To summarize, the Jones tower for the inclusion $\mathbb{C}\subset\mathbb{C}^n$ is $\mathbb{C}\subset\mathbb{C}^n\subset M_n(\mathbb{C})\subset M_n(\mathbb{C})\otimes \mathbb{C}^n\subset\cdots$.  
We remark that the full standard invariant of the inclusion $\mathbb{C}\subset \mathbb{C}^n$ is described by spin planar algebra\cite[Example 2.8]{Jones2021}. 
The tensor network also describes the same inclusion. (See also \cite{HJLW23}). 

\subsection{The Inclusion $\bC\subset M_n(\bC)$} 

Denote $\cN = \mathbb{C}$ and $\cM = M_n(\mathbb{C})$. 
Let $\{\vert j\rangle\}^n_{j=1}$ be a orthonormal basis of $\mathbb{C}^n$, and $\{E_{j,k}\}^n_{j,k=1}$ be a system of matrix units of $M_n(\mathbb{C})$ that satisfies $E_{j,k}\vert l\rangle = \delta_{k,l}\vert j\rangle$. 
For $1\leq j,k,p,q\leq n$, we define operators on $L^2(\cM)$ by $E_{(j,k),(p,q)}$ as $E_{(j,k),(p,q)}E_{p',q'}\Omega = \delta_{(p,q),(p',q')}E_{j,k}\Omega$. 
Then $\{E_{(j,k),(p,q)}\}^n_{j,k,p,q=1}$ forms a system of matrix units of $\mathcal{B}(L^2(\cM))$. 
The left regular representation of $E_{j,k}\in \cM$ is $\displaystyle \sum^n_{s=1} E_{(j,s),(k,s)}$. 
The Jones projection is 
\begin{align*}
e_1 = \frac{1}{n}\sum^n_{j,k=1}E_{(j,j),(k,k)}. 
\end{align*}
The basic construction $\cM_1$ is generated by $\cM$ and $e_1$. 
We have $\cM_1 = J\cN'J = \mathcal{B}(L^2(\cM))\cong M_{n^2}(\mathbb{C})$. 
We have $JX^*J = X^{\mathsf{T}}$, where $X\in \cM_1$. 
The trace $\tau_1$ is the unique normalized trace, with respect to which we have $\displaystyle \mathbb{E}_{\cM}(e_1) = \frac{1}{n^2}$. 
Therefore, $\cN\subset \cM$ is a $\displaystyle \frac{1}{n^2}$-extension. 
A natural choice of Pimsner-Popa basis for $\cN\subset\cM$ is $\{\sqrt{n}E_{j,k}\}^n_{j,k=1}$. 

We identify $L^2(\cM_1)$ with $\mathbb{C}^n\otimes\mathbb{C}^n\otimes\mathbb{C}^n\otimes\mathbb{C}^n$ by the unitary transformation 
\begin{align*}
    E_{(j,k),(p,q)}\Omega_1\mapsto \frac{1}{n} \vert j\rangle\otimes \vert k\rangle \otimes \vert p\rangle\otimes \vert q\rangle. 
\end{align*} 
The left action of $E_{(j,k),(p,q)}\in\cM_1$ on $L^2(\cM_1)$ is given by $E_{j,p}\otimes E_{k,q}\otimes I\otimes I$. 
This implies the left action of $E_{j,k}\in\cM$ to be $E_{j,k}\otimes I\otimes I\otimes I$. 
The modular conjugation $J_1$ on $L^2(\cM_1)$ acts as 
\begin{align*}
    J_1\vert j\rangle\otimes \vert k\rangle \otimes \vert p\rangle\otimes \vert q\rangle = \vert p\rangle\otimes \vert q\rangle \otimes \vert j\rangle\otimes \vert k\rangle. 
\end{align*}
The Jones projection $e_2$ is given by 
\begin{align*}
e_2 = \frac{1}{n}\sum^n_{j,k=1} I\otimes E_{j,k}\otimes I\otimes E_{j,k}.
\end{align*}
We therefore have $\cM_2 = M_n(\mathbb{C})\otimes M_n(\mathbb{C})\otimes I\otimes M_n(\mathbb{C})$, and $\cM'\cap \cM_2 = I\otimes M_n(\mathbb{C})\otimes I\otimes M_n(\mathbb{C})$. 
The trace $\tau_2$ is given as $\displaystyle \tau_2(X\otimes Y\otimes I\otimes Z) = \frac{1}{n^3}\Tr(X)\Tr(Y)\Tr(Z)$. 
The inclusion $\cM_1\subset \cM_2$ is again a $\displaystyle \frac{1}{n^2}$-extension. 
The modular conjugation $J_1$ acts on $\cM'\cap \cM_2$ as $J_1(I\otimes X\otimes I\otimes Y)^*J_1 = I\otimes Y^{\mathsf{T}}\otimes I\otimes X^{\mathsf{T}}$. 

Finally the Fourier transform of the system of matrix units of $\cN'\cap \cM_1$ is given by
\begin{align*}
    \mathfrak{F}(E_{(j,k),(p,q)}) = I\otimes E_{k,j}\otimes I\otimes E_{q,p}. 
\end{align*}
The element $E_{(j,k),(p,q)}\in\cM_1$ is depicted in the corresponding planar algebra as 
\begin{align*}
    \vcenter{\hbox{\begin{tikzpicture}[scale=1.2]
        \draw [blue] (0, -0.5)--(0, 0.5) (0.5, -0.5)--(0.5, 0.5);
        \draw [fill=white] (-0.2, -0.2) rectangle (0.2, 0.2);
        \node at (0, 0) {\tiny $E_{j,p}$};
        \begin{scope}[shift={(0.5, 0)}]
         \draw [fill=white] (-0.2, -0.2) rectangle (0.2, 0.2);
        \node at (0, 0) {\tiny $E_{k,q}$};   
        \end{scope}
    \end{tikzpicture}}}.
\end{align*}
Its Fourier transform in $\cM'\cap \cM_2$ is 
\begin{align*}
\vcenter{\hbox{\begin{tikzpicture}[scale=0.65]
    \begin{scope}[shift={(0,1.3)}]
    \draw [blue] (-0.5, 0.8)--(-0.5, 0) .. controls +(0, -0.6) and +(0,-0.6).. (0.5, 0)--(0.5, 0.8);    
\begin{scope}[shift={(0.5, 0.25)}]
\draw [fill=white] (-0.35, -0.35) rectangle (0.35, 0.35);
\node at (0, 0) {\tiny $E_{j,p}$};
\end{scope}
    \end{scope}
\draw [blue] (-0.5, -0.8)--(-0.5, 0) .. controls +(0, 0.6) and +(0,0.6).. (0.5, 0)--(0.5, -0.8);
\begin{scope}[shift={(0.5, -0.25)}]
\draw [fill=white] (-0.35, -0.35) rectangle (0.35, 0.35);
\node at (0, 0) {\tiny $E_{q,k}$};
\end{scope}
\end{tikzpicture}}}.
\end{align*}
Note that $\vcenter{\hbox{\begin{tikzpicture}[scale=0.65]
    \begin{scope}[shift={(0,1.3)}]
    \draw [blue] (-0.5, 0.8)--(-0.5, 0) .. controls +(0, -0.6) and +(0,-0.6).. (0.5, 0)--(0.5, 0.8);    
\begin{scope}[shift={(0.5, 0.25)}]
\draw [fill=white] (-0.35, -0.35) rectangle (0.35, 0.35);
\node at (0, 0) {\tiny $X$};
\end{scope}
    \end{scope}
\draw [blue] (-0.5, -0.8)--(-0.5, 0) .. controls +(0, 0.6) and +(0,0.6).. (0.5, 0)--(0.5, -0.8);
\begin{scope}[shift={(0.5, -0.25)}]
\draw [fill=white] (-0.35, -0.35) rectangle (0.35, 0.35);
\node at (0, 0) {\tiny $X^*$};
\end{scope}
\end{tikzpicture}}}$ is a scalar multiple of a minimal projection for any $X$ and 
\begin{align*}
\vcenter{\hbox{\begin{tikzpicture}[scale=0.65]
    \begin{scope}[shift={(0,1.3)}]
    \draw [blue] (-0.5, 0.8)--(-0.5, 0) .. controls +(0, -0.6) and +(0,-0.6).. (0.5, 0)--(0.5, 0.8);    
\begin{scope}[shift={(0.5, 0.25)}]
\draw [fill=white] (-0.35, -0.35) rectangle (0.35, 0.35);
\node at (0, 0) {\tiny $E_{j,k}$};
\end{scope}
    \end{scope}
\draw [blue] (-0.5, -0.8)--(-0.5, 0) .. controls +(0, 0.6) and +(0,0.6).. (0.5, 0)--(0.5, -0.8);
\begin{scope}[shift={(0.5, -0.25)}]
\draw [fill=white] (-0.35, -0.35) rectangle (0.35, 0.35);
\node at (0, 0) {\tiny $E_{k,j}$};
\end{scope}
\end{tikzpicture}}}
= \vcenter{\hbox{\begin{tikzpicture}[scale=1.2]
        \draw [blue] (0, -0.5)--(0, 0.5) (0.5, -0.5)--(0.5, 0.5);
        \draw [fill=white] (-0.2, -0.2) rectangle (0.2, 0.2);
        \node at (0, 0) {\tiny $E_{j,j}$};
        \begin{scope}[shift={(0.5, 0)}]
         \draw [fill=white] (-0.2, -0.2) rectangle (0.2, 0.2);
        \node at (0, 0) {\tiny $E_{k,k}$};   
        \end{scope}
    \end{tikzpicture}}}.
\end{align*}
To summarize, the Jones tower for the inclusion $\mathbb{C}\subset M_n(\mathbb{C})$ is $\mathbb{C}\subset M_n(\mathbb{C})\subset M_n(\mathbb{C})\otimes M_n(\mathbb{C})\subset M_n(\mathbb{C})\otimes M_n(\mathbb{C})\otimes M_n(\mathbb{C})\subset \cdots$.

\subsection{Finite index inclusions from unitary fusion categories} \label{sec:fusion}

In this section, we recall the canonical way to construct a finite index inclusion of type II$_1$ von Neumann algebras $\cN\subset\cM$ from a Frobenius algebra in a unitary fusion category $\mathcal{C}$. 

We denote the tensor unit of $\mathcal{C}$ by $1$. 
For an object $X$, denote by $\overline{X}$ the dual of $X$. 
Let $\text{ev}_X:\overline{X}\otimes X\rightarrow 1$ be the evaluation map  and $\text{coev}_X:1\rightarrow X\otimes \overline{X}$ the coevaluation map  satisfying the zigzag relations. 
Suppose that $X,Y\in \mathcal{C}$.
We denote by $\text{Hom}_{\mathcal{C}}(X,Y)$ the space of morphisms from $X$ to $Y$, and by $\text{Hom}_{\mathcal{C}}(X)$ the algebra of endomorphisms of $X$, with the identity morphism $id_X$. 
By the fact that $\mathcal{C}$ is unitary fusion category, $\text{Hom}_{\mathcal{C}}(X)$ is a finite dimensional $C^*$-algebra. 
There is a faithful trace on $\text{Hom}_{\mathcal{C}}(X)$ given by the categorical trace as follows: 
\begin{align*}
    \Tr(f) = \text{coev}^*_X(f\otimes id_{\overline{X}})\circ\text{coev}_X,\quad f\in \text{Hom}_\mathcal{C}(X). 
\end{align*}
The quantum dimension of $X$ is defined as $d_X = \Tr(id_X)$. 
In what follows we shall consider the normalized categorical trace $\displaystyle tr = \frac{1}{d_X}\Tr$. 

A $*$-Frobenius algebra is a triple $(\gamma,m,\eta)$ where $\gamma$ is an object in $\mathcal{C}$, $m:\gamma\otimes\gamma\rightarrow\gamma$ is the multiplication and $\eta:1\rightarrow
\gamma$ is the unit. (See \cite[Definition 3.1]{Muger03a} for the details.) 
By the universal construction of Müger, there exists a spherical Morita context $\mathcal{E}$ with objects $\{\mathfrak{A},\mathfrak{B}\}$, such that $\mathcal{C} = \text{END}(\mathfrak{A})$ and $\text{END}(\mathfrak{B})$ is a unitary fusion category. 
Moreoever there is a distinguished $1$-morphism $J:\mathfrak{A}\rightarrow \mathfrak{B}$ such that $\gamma = \overline{J}J$. 
Morita contexts can also be described in terms of module categories \cite{Ostrik03}. 

The construction of the inclusion from $\gamma$ is as follows. 
Let $\ell$ be a non-negative integer. 
For $k =2\ell$, define a finite dimensional $C^*$-algebra $M_k = \text{Hom}((\overline{J}J)^{\ell}\overline{J})$; for $k=2\ell+1$, define $M_k = \text{Hom}(J(\overline{J}J)^{\ell}\overline{J})$. 
There is a natural inclusion $\iota_k:M_{k}\rightarrow M_{k+1}$ defined as $id_J\otimes -$ for $k$ even and $id_{\overline{J}}\otimes -$ for $k$ odd. 
These inclusions preserve the normalized categorical traces, so we obtain a faithful trace $\tau$ on the $*$-algebra $\displaystyle M = \bigcup_{k\geq 0}M_k$.  
For $k=2\ell$, define a $C^*$-subalgebra $N_k\subset M_k$ as $N_k = \text{Hom}((\overline{J}J)^{\ell})\otimes id_{\overline{J}}$; for $k=2\ell+1$, define $N_k = \text{Hom}(J(\overline{J}J)^{\ell})\otimes id_{\overline{J}}$. 
The inclusions $\iota_k$ restricts to inclusions of $N_k$ into $N_{k+1}$, so we obtain an inclusion of $*$-algebras 
\begin{align*}
    \bigcup_{k\geq 0}N_k\subset \bigcup_{k\geq 0}M_k = M. 
\end{align*}
Let $L^2(M,\tau)$ be the GNS-construction of $M$ with respect to $\tau$, we define $\cM$ to be the closure of the left regular representation of $M$ in $\mathcal{B}(L^2(M,\tau))$ with respect to the weak operator topology. 
It is a standard procedure to show that $\tau$ extends to a normal faithful trace on $\cM$. 
Define $\cN\subset \cM$ to be the weak-closure of the $*$-subalgebra $\displaystyle \bigcup_{k\geq 0}N_k$. 
Thus $\cN\subset \cM$ is an inclusion of (hyperfinite) type II von Neumann algebras. 

It follows from the properties of commuting squares \cite[Proposition 5.1.9]{JonesSunder97} that $\cN\subset\cM$ is a $\lambda$-extension with $\displaystyle \lambda = \frac{1}{d_{\gamma}}$, where $\displaystyle d_\gamma$ is the quantum dimension of $\gamma$. 
By Ocneanu's compactness argument \cite[Theorem 5.7.1]{JonesSunder97}, the inclusion $\cN\subset \cM$ is irreducible if $\gamma$ is simple, namely $\text{Hom}_{\mathcal{C}}(1,\gamma)$ is $1$-dimensional. 
The higher relative commutants of $\cN\subset \cM$ can also be computed by this theorem. 
We have for $k=2\ell$, $\cN'\cap \cM_k = \text{Hom}(\overline{J}(J\overline{J})^{\ell})$ and $\cM'\cap \cM_k = \text{Hom}(J(\overline{J}J)^{\ell})$; for $k = 2\ell + 1$,  $\cN'\cap \cM_k = \text{Hom}((\overline{J}J)^{\ell+1})$ and $\cM'\cap \cM_k = \text{Hom}((J\overline{J})^{\ell+1})$. In particular, we have $\cN'\cap \cM_1 = \text{Hom}(\overline{J}J) = \text{Hom}_{\mathcal{C}}(\gamma)$, and $\cM'\cap \cM_2 = \text{Hom}(\gamma\otimes \gamma)$ where $\gamma\otimes \gamma$ is treated as a $\gamma$-bimodule in $\mathcal{C}$. 
By a result of Popa \cite[Corollary 3.7]{Popa1990}, $\cN\subset\cM$ is always extremal for there are only finitely many equivalent classes of simple objects in a fusion category. 
The dual inclusion $\cM\subset \cM_1$ of $\cN\subset \cM$ can be constructed in the same way with $J$ and $\overline{J}$ exchanged. 
For the dual inclusion, we have $\cM'\cap \cM_2 = \text{Hom}(\gamma\otimes \gamma)$ as $\gamma$-bimodules, and $\cM_1'\cap \cM_3 = \text{Hom}_{\mathcal{C}}(\gamma)$. 

We provide some examples of the above construction by specifying the unitary fusion category and the Frobenius algebra. 
Let $\mathcal{C}$ be a unitary fusion category, then so is $\mathcal{C}\boxtimes\mathcal{C}^{\text{op}}$. 
The object $\displaystyle \gamma = \bigoplus_{X\in \text{Irr}(\mathcal{C})}X\boxtimes X^{\text{op}}$ admits the structure of a simple Frobenius algebra, thus produces a irreducible subfactor $\cN\subset\cM$. 
This is known as the quantum double construction\cite{Muger03b}. 
The fusion ring of the underlying category $\mathcal{C}$ is encoded in the triple $(\cN'\cap\cM_1,\cM'\cap \cM_2,\mathfrak{F})$, which is an instance of fusion bialgebras \cite{LPW2021}. 
For a concrete example, consider $\mathcal{C} = \mathcal{H}_3$ being the Haagerup fusion category \cite{AsaedaHaagerup99} with simple objects $\{1,\alpha,\alpha^2,\zeta,\alpha\zeta,\alpha^2\zeta\}$ and non-commutative fusion rules: 
\begin{align*}
    \alpha^3 = 1,\  \zeta\alpha = \alpha^2\zeta, \ \zeta \alpha^2 = \alpha\zeta, \ \zeta^2 = 1+ \zeta + \alpha\zeta + \alpha^2\zeta. 
\end{align*}
It then follows from \cite[Proposition 7.4]{LPW2021} that the associated inclusion $\cN\subset \cM$ has relative commutants $\cN'\cap \cM_1 \cong \mathbb{C}^6$ and $\cM'\cap \cM_2\cong \mathbb{C}^2\oplus M_2(\mathbb{C})$. 
This is an instance of an irreducible inclusion with $1$-side commutativity for $2$-boxes. 
we have that $\overline{\alpha} = \alpha^2$, i.e. the modular conjugation acts non-trivially on $\cN'\cap \cM_1$.
The quantum dimensions $d_\alpha=d_{\alpha^2}=1$ and $\displaystyle d_{\zeta}= d_{\alpha\zeta}=d_{\alpha^2\zeta}=\frac{3+\sqrt{13}}{2}$.
The global dimension $\displaystyle \lambda^{-1}=\frac{9(5+\sqrt{13})}{4}$.
We denote by $\cN^{h}\subset\cM^h$ the inclusion $\cM'\subset \cN'$ on $L^2(\cM)$.
Then ${\cM^h_1}'\cap \cM_2^h=\bC^6$.
Let $p_\alpha, p_{\alpha^2}, p_{\zeta}, p_{\alpha\zeta},p_{\alpha^2\zeta} \in {\cM^h}'\cap \cM_2^h$ the minimal projections corresponding to $\alpha, \alpha^2, \zeta, \alpha\zeta, \alpha^2\zeta$.

\section{Bimodule Quantum Channels}\label{sec:: Bimodule quantum channels} 

Let $\cM$ be a von Neumann algebra.
A linear map $\Phi:\cM\rightarrow \cM$ is called positive if it preserves the positive cone $\cM_+$. 
The map $\Phi$ is called completely positive if $\Phi\otimes id_n$ is positive on $\cM\otimes M_n(\mathbb{C})$ for all $n\geq 1$, and completely bounded if $\displaystyle \sup_{n \geq 1}\Vert \Phi\otimes id_n\Vert$ is finite. 
The map $\Phi$ is unital if $\Phi(1)=1$.
A quantum channel is a normal unital completely positive map on $\cM$. 

For a finite inclusion $\cN\subset \cM$ of finite von Neumann algebras, a linear map $\Phi$ is said to be $\cN$-bimodule, if 
\begin{align*}
\Phi(y_1 x y_2) = y_1\Phi(x)y_2,
\end{align*}
for all $y_1, y_2\in\cN$ and $x\in\cM$. 
The $\cN$-bimodule map can be characterized by an element in the relative commutant $\cM'\cap \cM_2$.
Suppose $\Phi$ is an $\cN$-bimodule map on $\cM$.
The Fourier multiplier $\widehat{\Phi}$ of $\Phi$ is defined as follows: for all $x, y\in\mathcal{M}$,
\begin{equation}\label{eq:fouriermultiple}
    \begin{aligned}
        \widehat{\Phi}  x e_1 y \Omega_1 = & \lambda^{1/2}\sum_{j=1}^m x \eta_j^* e_1\Phi(\eta_j)y\Omega_1, \\
        =& \lambda^{1/2}\sum_{j=1}^m x \bE_{\cN'}(\eta_j^* e_1\Phi(\eta_j))y\Omega_1
    \end{aligned}
\end{equation}
We check that for $x_1,y_1,x_2,y_2\in\cM$: 
\begin{align*}
    \langle \widehat{\Phi}(x_1e_1y_1\Omega_1), x_2e_1y_2\Omega_1\rangle
    &= \lambda^{1/2} \sum^m_{j=1}\langle x_1\eta^*_je_1\Phi(\eta_j)y_1\Omega_1, x_2e_1y_2\rangle\\
    &= \lambda^{1/2}\sum^m_{j=1}\tau_1(y^*_2e_1x^*_2x_1\eta^*_je_1\Phi(\eta_j)y_1)\\
    &= \lambda^{1/2}\sum^m_{j=1}\tau_1(y^*_2e_1\mathbb{E}_{\cN}(x^*_2x_1\eta^*_j)\Phi(\eta_j)y_1)\\
    &= \lambda^{1/2}\tau_1(y^*_2e_1\Phi(x^*_2x_1)y_1) = \lambda^{3/2}\tau(y^*_2\Phi(x^*_2x_1)y_1). 
\end{align*}
Thus $\widehat{\Phi}$ is the unique element in $\cM'\cap\cM_2$ such that 
\begin{equation}\label{eqn:: bilinear form induced by Fourier multiplier}
    \langle \widehat{\Phi}(x_1e_1y_1\Omega_1), x_2e_1y_2\Omega_1\rangle = \lambda^{3/2}\tau(y^*_2\Phi(x^*_2x_1)y_1),\quad x_1, x_2,y_1, y_2\in \cM.
\end{equation}
The complete positivity of a bounded $\cN$-bimodule map $\Phi$ on $\cM$ is equivalent to the positivity of its Fourier multiplier. 
This can be directly seen from the positivity of the bilinear form induced by $\widehat{\Phi}$ as in Equation \eqref{eqn:: bilinear form induced by Fourier multiplier}. 
The bimodule map $\Phi$ can be written in terms of the Fourier multiplier $\widehat{\Phi}$ as follows:
\begin{align}\label{eq:multiplierphi}
    \Phi(x)=\lambda^{-5/2} \bE_{\mathcal{M}}(e_2e_1\widehat{\Phi} x e_1e_2).
\end{align}
This can be proved by using Equation \eqref{eqn:: bilinear form induced by Fourier multiplier} and Lemma \ref{lem:fourier}, as for all $x,y\in\cM$:  
\begin{align*}
    \tau(y^* \Phi(x)) 
    =& \lambda^{-3/2}\langle \widehat{\Phi}xe_1\Omega_1,e_1y\Omega_1\rangle \\
    =& \lambda^{-5/2}\tau_2(y^*e_2e_1\widehat{\Phi}xe_1e_2) \\
    =& \lambda^{-5/2}\tau_2(y^*\mathbb{E}_{\cM}(e_2e_1\widehat{\Phi}xe_1e_2)).
\end{align*}
Moreover, we have that $\Phi(x)e_2=\lambda^{-3/2}e_2e_1\widehat{\Phi}xe_1e_2$.
We shall informally graph $\Phi(x)$ as follows:
\begin{align*}
    \Phi(x)=\raisebox{-0.9cm}{
\begin{tikzpicture}[scale=1.5]
\path [fill=gray!40] (0.5, -0.4) rectangle (0.8, 0.9);
\path [fill=gray!40] (-0.3, 0) rectangle (0, 0.6);
\draw [blue,fill=white] (0,0) rectangle (0.7, 0.5);
\node at (0.35, 0.25) {$\widehat{\Phi}$};
\begin{scope}[shift={(-0.6, 0)}]
\draw [blue,fill=white, rounded corners] (0,0) rectangle (0.35, 0.5);
\node at (0.175, 0.25) {$x$};
\end{scope}
\draw [fill=gray!40] (0.5, 0)--(0.5, -0.4) (0.5, 0.5)--(0.5, 0.9);
\draw[fill=gray!40] (0.2, 0.5) .. controls +(0, 0.4) and +(0, 0.4) .. (-0.425, 0.5);
\draw[fill=gray!40] (0.2, 0) .. controls +(0, -0.4) and +(0, -0.4) .. (-0.425, 0);
\end{tikzpicture}},
\end{align*}
and write $\Phi(x)$ as $x* \widehat{\Phi}$.
We have that $\widehat{\text{id}}=\lambda^{-1/2}e_2$.

\begin{proposition}\label{prop:fouriermultiple2}
Suppose that $\Phi$ is a bimodule bounded map.
Then $\displaystyle \sum_{j=1}^m \Phi(\eta_j^*) e_1 \eta_j \in \cN'\cap \cM_1$ and 
\begin{align*}
    \widehat{\Phi}= \mathfrak{F}\left(\sum_{j=1}^m\bE_{\cN'}(\eta_j^* e_1\Phi(\eta_j))\right) =  \mathfrak{F}^{-1}\left(\sum_{j=1}^m\Phi(\eta_j^*) e_1 \eta_j\right).
\end{align*}
In particular, $\Phi(x)\Omega = \mathfrak{F}(\widehat{\Phi})x\Omega$, for $x\in\cM$. 
\end{proposition}

\begin{proof}
    This follows from Equation \eqref{eq:fouriermultiple} and Lemma \ref{lem:fourier}, here we provide a direct computation. 
    For $x,y\in\cM$, 
    \begin{align*}
        \mathfrak{F}^{-1}\left(\sum^m_{j=1} \Phi(\eta^*_j)e_1\eta_j\right)xe_1y\Omega_1 
        &= \lambda^{-3/2}x \mathbb{E}_{\cM'}\left(e_1e_2\sum_{j=1}^m\Phi(\eta^*_j)e_1\eta_j\right)e_1y
        \Omega_1\\
        &= \lambda^{-1/2}x\sum_{j_1,j_2=1}^m\eta^*_{j_1}e_1\mathbb{E}_{\cM}\left(\Phi(\eta^*_{j_2})e_1\eta_{j_2}\eta_{j_1}e_1 y\right)\Omega_1\\
        &= \lambda^{1/2} x\sum_{j_1=1}^m\eta^*_{j_1} e_1\Phi(\eta_{j_1})y\Omega_1\\
        &=  \widehat{\Phi}xe_1y\Omega_1.
    \end{align*}
    This shows $\displaystyle \widehat{\Phi} =  \mathfrak{F}^{-1}\left(\sum_{j=1}^m\Phi(\eta_j^*) e_1 \eta_j\right)$. 
\end{proof}

Suppose $\Phi, \Psi$ are $\cN$-bimodule maps on $\cM$.
The composition of bimodule maps is characterized by the convolution of the Fourier multipliers: 
\begin{align*}
    \widehat{\Phi\Psi}= \widehat{\Psi}* \widehat{\Phi}.
\end{align*}

For a bounded $\cN$-bimodule map $\Phi$ on $\cM$, we define the adjoint of $\Phi$ (with respect to $\tau$) by $\tau(\Phi^*(y)^*x) = \tau(y^*\Phi(x))$. 
Note that $\Phi^*$ is a trace-preserving bimodule map if and only if $\Phi$ is unital. 
\begin{proposition}
Suppose that $\Phi$ is a bimodule bounded map.
We have that 
    \begin{align}
\widehat{\Phi^*}=\overline{\widehat{\Phi}}^*.
    \end{align}
    In particular, for a completely positive map $\Phi$, we always have $\widehat{\Phi^*} = \overline{\widehat{\Phi}}$.
\end{proposition}
\begin{proof}
For any $x, y\in \cM$, we have that
\begin{align*}
    \tau(y^*\bE_{\cM}(e_2e_1 \widehat{\Phi} x e_1e_2))
    = \tau(\bE_{\cM}(e_2e_1 \widehat{\Phi^*}^*y^* e_1 xe_2)).
\end{align*} 
By removing the conditional expectation $\bE_{\cM}$, we have that
\begin{align*}
 \tau_2(y^*e_2e_1 \widehat{\Phi} x e_1e_2)
    = \tau_2(e_2e_1 \widehat{\Phi^*}^*y^* e_1 xe_2),  
\end{align*}
i.e.
\begin{align*}
 \tau_2(e_2e_1 \widehat{\Phi} x e_1 y^* )
    = \tau_2( \widehat{\Phi^*}^* e_1 e_2 x e_1 y^*).   
\end{align*}
By Lemma \ref{lem:fourier}, we see that
\begin{align*}
  \tau_2(e_2\mathfrak{F}^{-1}(\overline{\widehat{\Phi}}) x e_1 y^* )
    = \tau_2( \mathfrak{F}^{-1}(\widehat{\Phi^*}^*) e_2 x e_1 y^*) .   
\end{align*}
This implies that 
\begin{align*}
    \mathfrak{F}^{-1}(\overline{\widehat{\Phi}})= \mathfrak{F}^{-1}(\widehat{\Phi^*}^*),
\end{align*}
i.e. $\widehat{\Phi^*}=\overline{\widehat{\Phi}}^*$.
This completes the proof of the proposition.
\end{proof}

Suppose $\Phi:\cM\rightarrow\cM$ is a $\cN$-bimodule quantum channel.
The fixed point space of $\Phi$ is given by 
\begin{align*}
\fix(\Phi) = \{x\in \cM \vert \Phi(x) = x\}.
\end{align*} 
The multiplicative domain of $\Phi$ is 
\begin{align*}
\nfix(\Phi) = \{x\in \cM \vert \Phi(x^*x) = \Phi(x)^*\Phi(x),\Phi(xx^*) = \Phi(x)\Phi(x)^*\}.
\end{align*}
Since $\Phi$ is a $\cN$-bimodule map, $\cN\subseteq \fix(\Phi)\cap \nfix(\Phi)$. 
The multiplicative domain of a quantum channel forms a von Neumann subalgebra of $\cM$, while this is not true for the fixed point space in general. 

For an $\cN$-bimodule quantum channel $\Phi$ on $\cM$ with Fourier multiplier $\widehat{\Phi}\in \cM'\cap \cM_2$, we have $\widehat{\Phi^k} = \widehat{\Phi}^{(*k)}$. 
The limit $\displaystyle \mathbb{E}_{\Phi} = \lim_{\ell\rightarrow \infty}\frac{1}{\ell}\sum^\ell_{k=1}\Phi^{k}$ exists as a $\cN$-bimodule quantum channel, with the property that $\mathbb{E}^2_{\Phi} = \mathbb{E}_{\Phi}$. 
Moreover, the image of $\mathbb{E}_{\Phi}$ is $\fix(\Phi)$. 
Taking Fourier multiplier gives: 
\begin{align*}
    \widehat{\mathbb{E}}_{\Phi} = \lim_{\ell \rightarrow \infty}\frac{1}{\ell}\sum^\ell _{k=1} \widehat{\Phi}^{(*k)},
\end{align*}
with $\widehat{\mathbb{E}}_{\Phi}*\widehat{\mathbb{E}}_{\Phi} =  \widehat{\mathbb{E}}_{\Phi}$. 
Therefore, $\mathbb{E}_{\Phi}$ is an idempotent with positive Fourier multiplier. 

Suppose that the image of $\mathbb{E}_{\Phi}$ is a $*$-subalgebra $\mathcal{P}\subset\cM$, then it is known that $\mathbb{E}_{\Phi}$ is a $\mathcal{P}$-bimodule map and is a contraction with respect to the operator norm. 
Namely, $\mathbb{E}_{\Phi}$ is a conditional expectation onto $\mathcal{P}$. 
Let $e_{\mathcal{P}}$ be the projection from $L^2(\cM, \tau)$ onto $L^2(\mathcal{P}, \tau)$.
Then we have  
\begin{align*}
\cR(\mathfrak{F}^{-1}(\widehat{\mathbb{E}}_{\Phi})\mathfrak{F}^{-1}(\widehat{\mathbb{E}}_{\Phi})^*)=e_{\mathcal{P}}.
\end{align*}
In particular, if $\fix(\Phi)=\cN$, we have $\cR(\mathfrak{F}^{-1}(\widehat{\mathbb{E}}_{\Phi})\mathfrak{F}^{-1}(\widehat{\mathbb{E}}_{\Phi})^*)=e_1.$

\begin{remark}
    We cannot conclude that $\mathbb{E}_{\Phi}$ is a trace-preserving conditional expectation, even if its image is a $*$-subalgebra. 
    However this is true when the inclusion $\cN\subset \cM$ is irreducible \cite[Proposition 3.2]{JLW2019}. 
\end{remark}

\begin{definition}[Convolution Support Projection]
Suppose $x\in \cM'\cap \cM_2$.
The convolution support projection $\CS(x)$ is defined to be
\begin{align*}
    \CS(x)=\bigvee_{k \geq 1, \epsilon_1, \ldots, \epsilon_k\in \{1, \overline{\ }\}} \cR(x^{\epsilon_1}*\cdots * x^{\epsilon_k}),
\end{align*}
where $x^\epsilon=x$ if $\epsilon=1$ and $x^\epsilon=\overline{x}$ if $\epsilon=\overline{\ }$.
We say $x$ is connected if $\CS(x)=1$.
We denote by $\displaystyle \CS_0(x)=\bigvee_{k \geq 1} \cR(x^{(*k)})$.
\end{definition}

\begin{remark}
We have that $\cR(x)\leq \CS_0(x) \leq  \CS(x)\leq 1$ if $x$ is positive.
If $\cN\subset \cM$ is irreducible, then $\CS(x)$ is a biprojection by \cite{Liu16} and \cite[Proposition 3.2]{JLW2019}.
\end{remark}

\begin{remark}
Suppose $G$ is an undirected graph and $A_G$ is its adjacent matrix.
By considering the inclusion $\bC\subset \bC^{n}$, where $n=|G|$, the cardinality of $G$, we have that $A_G\in \cN'\cap \cM_1$ and $\mathfrak{F}(A_G)$ is positive.
Moreover, $\CS(\mathfrak{F}(A_G))=1$ if and only if $G$ is connected.
\end{remark}

\begin{remark}
    Suppose $\Phi$ is a bimodule quantum channel.
    We have that $\cR(\widehat{\bE}_{\Phi})\leq \CS_0(\widehat{\Phi})$.
\end{remark}

We recall the relative irreducibility of a bimodule map introduced in \cite{HJLW24}.
\begin{definition}[Relative Irreducibility]
Suppose $\mathcal{N}\subseteq\mathcal{M}$ is a finite inclusion of von Neumann algebras and $\Phi$: $\mathcal{M}\to\mathcal{M}$ is bimodule positive map.
We say $\Phi$ is relatively irreducible if for any projection $p\in\mathcal{M}$ and positive number $c>0$ satisfying $\Phi(p)\leq cp$, we have that $p\in\mathcal{N}$.  
\end{definition}

\begin{remark}
    Theorem 5.8 in \cite{HJLW24} shows that $\CS(\widehat{\Phi})=1$ implies that $\Phi$ is relatively irreducible whenever $\Phi$ is a completely positive bimodule map.
    Lemma 5.9 in \cite{HJLW24} shows that when $\Phi$ is relatively irreducible bimodule quantum channel and $\cN$ is a factor, $\Phi$ is equilibrium with respect to a normal faithful state $\rho$, i.e. $\rho\circ \Phi=\rho$.
\end{remark}

We recall the following useful result for relatively irreducible bimodule quantum channels in \cite{HJLW24}.
\begin{proposition}[Proposition 5.12 in \cite{HJLW24}]\label{prop:wilandt}
    Suppose that $\Phi$ is a relatively irreducible bimodule quantum channel and $\cN$ is a factor. 
   If there exists a nonzero positive element $x \in\mathcal{M}$ such that $\Phi(x)\leq x$ or $\Phi(x)\geq x$, then $ \Phi(x)= x$.
\end{proposition}

Furthermore, Theorem 5.10 in \cite{HJLW24} states that
\begin{proposition}\label{prop:relativeirreducible}
    Suppose $\mathcal{N}\subseteq\mathcal{M}$ is a finite inclusion of finite von Neumann algebras and $\Phi$ is a relatively irreducible bimodule quantum channel.
    Suppose $\mathcal{N}$ is a factor.
  Then the eigenvalues of $\Phi$ with modulus $1$ form a finite cyclic subgroup $\Gamma$ of the unit circle $ U(1)$.
  The fixed points space $\fix(\Phi)=\mathcal{N}$.
  For each $\alpha\in\Gamma$, there exists a unitary $u_{\alpha}\in\fix(\Phi,\alpha)=\{x\in \cM: \Phi(x)=\alpha x\}$ such that $\fix(\Phi,\alpha)=u_{\alpha}\mathcal{N}=\mathcal{N}u_{\alpha}$.
\end{proposition}

\begin{lemma}\label{lem:subalgebra}
Suppose $\Phi$ is an $\cN$-bimodule quantum channel. 
Suppose that there exists a faithful normal  state $\rho$ on $\mathcal{M}$ such that $\Phi$ is equilibrium with respect to $\rho$.
Then $x\in\fix(\Phi)$ if and only if
\begin{align}
x\widehat{\Phi}^{1/2}e_1e_2= \widehat{\Phi}^{1/2}e_1e_2 x.
\end{align}
Consequently, $\fix(\Phi)$ is a von Neumann subalgebra.
\end{lemma}
\begin{proof}
Suppose that $x\in \fix(\Phi)$.
By the Kadison-Schwarz inequality, 
\begin{align*}
\Phi(x^*x)\geq\Phi(x)^*\Phi(x)=x^*x.
\end{align*}
Since $\rho(\Phi(x^*x)-x^*x)=0$, we have $\Phi(x^*x)=x^*x$ as $\rho$ is faithful.
Let $y=x\widehat{\Phi}^{1/2}e_1e_2- \widehat{\Phi}^{1/2}e_1e_2 x$.
Then
\begin{align*}
\tau(\bE_{\mathcal{M}}(y^*y))= &\tau\left( \bE_{\mathcal{M}}\big((x\widehat{\Phi}^{1/2}e_1e_2-  \widehat{\Phi}^{1/2}e_1e_2 x)^*(x\widehat{\Phi}^{1/2}e_1e_2-  \widehat{\Phi}^{1/2}e_1e_2 x)\big)\right)\\
= &\lambda^{3/2}\tau(\Phi(x^*x))+\lambda^{3/2}\tau(xx^*\Phi(\mathbf{1}))- \lambda^{3/2}\tau(x^*\Phi(x))- \lambda^{3/2}\tau(x\Phi(x^*))\\
\leq & \lambda^{3/2}\tau(x^*x)+\lambda^{3/2}\tau(xx^*)-\lambda^{3/2}\tau(x^*x)-\lambda^{3/2}\tau(xx^*)\\
= &0.
\end{align*}
Hence $\bE_{\mathcal{M}}(y^*y)=0$.
By the Pimsner-Popa inequality, $y=0$.
Note that $x^*\in \fix(\Phi)$.
We see that $x^*\widehat{\Phi}^{1/2}e_1e_2= \widehat{\Phi}^{1/2}e_1e_2 x^*$.
This implies that $e_2e_1\widehat{\Phi}^{1/2}x=  xe_2e_1\widehat{\Phi}^{1/2}$.

Suppose that $x\in \cM$ and $x\widehat{\Phi}^{1/2}e_1e_2= \widehat{\Phi}^{1/2}e_1e_2 x.$
We see that $\Phi(x)=x$ by multiplying $e_2e_1\widehat{\Phi}^{1/2}$ from the left hand side and taking the conditional expectation $\bE_{\cM}$.
\end{proof}

\begin{remark}
    It is known that when $\Phi$ is equilibrium with respect to a normal faithful state $\rho$, we have $\fix(\Phi)\subseteq \nfix(\Phi)$. 
\end{remark}

\begin{proposition}\label{prop:subalgebra2}
Suppose $\Phi$ is an $\cN$-bimodule quantum channel equilibrium with respect to a faithful normal  state $\rho$ on $\mathcal{M}$.
Then $x\in\fix(\Phi)$ if and only if 
\begin{align}
\CS(\widehat{\Phi})x e_1e_2= \CS(\widehat{\Phi})e_1 x e_2.
\end{align}
\end{proposition}
\begin{proof}
Suppose that $x\in \mathscr{M}(\Phi)$.
We have that $x^*\in \fix(\Phi)$.
By Lemma \ref{lem:subalgebra}, we have $x\widehat{\Phi}^{1/2}e_1e_2= \widehat{\Phi}^{1/2}e_1e_2 x$, i.e. $x \mathfrak{F}^{-1}(\widehat{\Phi}^{1/2})= \mathfrak{F}^{-1}(\widehat{\Phi}^{1/2})x$.
By considering the adjoint, we have that $x \mathfrak{F}^{-1}(\widehat{\Phi}^{1/2})^*= \mathfrak{F}^{-1}(\widehat{\Phi}^{1/2})^*x$. 
Multiplying $\mathfrak{F}^{-1}(\widehat{\Phi}^{1/2})$, $\mathfrak{F}^{-1}(\widehat{\Phi}^{1/2})^*$ multiple times, we obtain that 
\begin{align*}
x \mathfrak{F}^{-1}(\widehat{\Phi}^{1/2})^{\epsilon_1}\cdots \mathfrak{F}^{-1}(\widehat{\Phi}^{1/2})^{\epsilon_k} = \mathfrak{F}^{-1}(\widehat{\Phi}^{1/2})^{\epsilon_1}\cdots \mathfrak{F}^{-1}(\widehat{\Phi}^{1/2})^{\epsilon_k} x, \quad k \in \bN, \quad \epsilon_1, \ldots, \epsilon_k\in \{1, *\}.
\end{align*}
By Lemma \ref{lem:fourier}, we see that 
\begin{align*}
 (\widehat{\Phi}^{\epsilon_1/2}*\cdots * \widehat{\Phi}^{\epsilon_k/2}) xe_1e_2 =  (\widehat{\Phi}^{\epsilon_1/2}*\cdots *\widehat{\Phi}^{\epsilon_k/2}) e_1xe_2, \quad k  \in \bN, \quad  \epsilon_1, \ldots, \epsilon_k\in \{1, \overline{\ }\}.
\end{align*}
Hence $\CS(\widehat{\Phi}^{1/2})x e_1e_2= \CS(\widehat{\Phi}^{1/2})e_1 x e_2.$
Note that $\CS(\widehat{\Phi}^{1/2})= \CS(\widehat{\Phi})$.
We have 
\begin{align*}
\CS(\widehat{\Phi})x e_1e_2= \CS(\widehat{\Phi})e_1 x e_2.
\end{align*}

Suppose that $\CS(\widehat{\Phi})x e_1e_2= \CS(\widehat{\Phi})e_1 x e_2.$
Then we obtain that $ \widehat{\Phi} xe_1e_2 =  \widehat{\Phi}e_1xe_2$.
Multiplying $e_2e_1$ from the left hand side and taking the conditional expectation $\bE_{\cM}$, we see that $\Phi(x)=x$, i.e $x\in \fix(\Phi)$.
\end{proof}

\begin{remark}\label{rem:subalgebra3}
Suppose $\Phi$ is an $\cN$-bimodule quantum channel equilibrium with respect to a faithful normal  state $\rho$ on $\mathcal{M}$ and $x\in\fix(\Phi)$. 
We have that $\Phi^*(x)=\Phi^*(1)x$ and
\begin{align*}
  \widehat{\bE}_{\Phi}  x e_1e_2= \widehat{\bE}_{\Phi} e_1 x e_2.
\end{align*}
\end{remark}

\begin{corollary}
Suppose $\Phi$ is an $\cN$-bimodule quantum channel equilibrium with respect to a faithful normal state $\rho$ on $\mathcal{M}$.
If one of the following holds:
\begin{enumerate}[(1)]
    \item $\CS(\widehat{\Phi})=1$;
    \item $\widehat{\bE}_{\Phi}$ is invertible,
\end{enumerate}
then $\fix(\Phi)=\cN$.
\end{corollary}
\begin{proof}
    If one of the above conditions holds, we have that $xe_1e_2=e_1xe_2$ for any $x\in \fix(\Phi)$ by Proposition \ref{prop:subalgebra2} and Remark \ref{rem:subalgebra3}.
    This implies that $x\in \cN$.
    Therefore $\fix(\Phi)=\cN$.
\end{proof}

\section{Bimodule GNS Symmetry}\label{sec:: Bimodule Equilibrium and Bimodule GNS Symmetry}

In this section, we will interpret the equilibrium in the bimodule setting and introduce bimodule GNS symmetry for bimodule quantum channels.
Suppose $\rho$ is a normal faithful state on $\cM$.
Let $\Omega_\rho$ be the separating and cyclic vector in the GNS representation Hilbert space $L^2(\cM, \rho)$.

Let $S_{\rho, \tau}$ be the relative modular operator defined by $S_{\rho,\tau} x\Omega=x^*\Omega_\rho$ for any $x\in \cM$.
Let $S_{\rho,\tau}=J \Delta_\rho^{1/2}$, where $\Delta_{\rho}$ is an (unbounded) operator affiliated to $\cM$.  
Let $\sigma_t^\rho$ be the modular automorphism, $t\in \bR$.
We will denote $\sigma_t^{\rho}$ by $\sigma_t$ for simplicity when no confusion arises. 

The state $\rho$ is canonically lifted to $\cM_1$ by $\rho\circ \mathbb{E}_{\cM}$. 
The relative modular operator of $\rho\circ\mathbb{E}_{\cM}$ with respect to $\tau_1$ is just $\Delta_{\rho}$, viewed as an operator affiliated with $\cM_1$. 

The modular operator $\Delta_\rho$ is an unbounded operator in general. 
In the bimodule case, we shall consider the operator $\widehat{\Delta}_\rho$ defined as follows. 
We shall assume that the modular operator is compatible with the basic construction, namely $e_1\in \Dom(\sigma^{\rho}_{-i})$. 
Under this assumption, we define 
\begin{align}
\widehat{\Delta}_\rho
=\lambda^{-1/2}\mathfrak{F}( \bE_{\cN'} (\sigma_{-i}(e_1))). 
\end{align}
We have that $\bE_{\cN'} (\Delta_{\rho} e_1 \Delta_{\rho}^{-1} )\in \cN'\cap \cM_1$ is bounded.
Hence $\widehat{\Delta}_\rho\in \cM'\cap \cM_2$ is bounded. 
\begin{remark}
    We say the normal faithful state $\rho$ is a hypertrace on $\cM$ over $\cN$ if $\rho|_{\cN}$ is a tracial state. 
    For such $\rho$, we have $\Delta_{\rho}$ affiliated with $\cN'\cap\cM$. 
    Note that being a hypertrace does not imply that $\Delta_{\rho}$ is bounded. 
    If the modular automorphism group of $\rho$ keeps $\cN$ fixed globally and the center of $\cN$ is finite dimensional, then $\Delta_{\rho}$ is bounded. 
    For such $\rho$, we have
    \begin{align*}
        \widehat{\Delta}_{\rho} = \lambda^{-1/2}\mathfrak{F}(\Delta_{\rho}e_1\Delta^{-1}_{\rho}). 
    \end{align*}
    If $\rho|_{\cN}=\tau$, we have that $\Delta_\rho e_1=e_1\Delta_\rho$.
This implies that $\widehat{\Delta}_{\rho}=1$. 
If $\cN=\bC$ and $\rho$ is a normal faithful state on $\cM$, then $\rho$ is a hyper-trace on $\cM$ over $\cN$.
In this case, 
\begin{align*}
\widehat{\Delta}_\rho=\lambda^{-2}\bE_{\cM'}(\Delta_\rho e_1 \Delta_\rho^{-1} e_2 e_1),
\end{align*}
depicted as 
\begin{align*}
    \vcenter{\hbox{\begin{tikzpicture}[scale=1.2]
        \draw [blue] (0, -0.5)--(0, 0.5) (0.5, -0.5)--(0.5, 0.5) (-0.5, -0.5)--(-0.5, 0.5);
        \draw [fill=white] (-0.2, -0.2) rectangle (0.2, 0.2);
        \node at (0, 0) {\tiny $\overline{\Delta_\rho}$};
        \begin{scope}[shift={(0.5, 0)}]
         \draw [fill=white] (-0.2, -0.2) rectangle (0.2, 0.2);
        \node at (0, 0) {\tiny $\Delta_\rho^{-1}$};   
        \end{scope}
    \end{tikzpicture}}}.
\end{align*}
The contragredient $\overline{\widehat{\Delta}}_\rho=\lambda^{-2}\bE_{\cM'}(\Delta_\rho^{-1} e_1 \Delta_\rho e_2 e_1)$ is depicted as
\begin{align*}
    \vcenter{\hbox{\begin{tikzpicture}[scale=1.4]
        \draw [blue] (0, -0.5)--(0, 0.5) (0.5, -0.5)--(0.5, 0.5) (-0.5, -0.5)--(-0.5, 0.5);
        \draw [fill=white] (-0.2, -0.2) rectangle (0.2, 0.2);
        \node at (0, 0) {\tiny $\overline{\Delta_\rho^{-1}}$};
        \begin{scope}[shift={(0.5, 0)}]
         \draw [fill=white] (-0.2, -0.2) rectangle (0.2, 0.2);
        \node at (0, 0) {\tiny $\Delta_\rho$};   
        \end{scope}
    \end{tikzpicture}}}.
\end{align*}
\end{remark}

\begin{remark}
    Suppose that $\bC^n\subset  M_n(\bC)$ is the inclusion, $\rho$ is a faithful state on $M_n(\bC)$.
    Then $\widehat{\Delta}_\rho$ is the Schur product of $\Delta_\rho^{\mathsf{T}}$ and $\Delta_\rho^{-1}$.
\end{remark}

\begin{remark}
The normal faithful state $\rho$ can be replaced by a normal semifinite faithful weight.    
\end{remark}

\begin{lemma}
    Suppose $\rho$ is a normal faithful state on $\cM$ such that $e_1\in \Dom(\sigma_{-i})$.
    Then $\cN'\cap \cM_1\subset \Dom(\sigma_{-i})$. 
\end{lemma}
\begin{proof}
For any $y\in \cM'\cap \cM_2$, we have that 
    \begin{align*}
 y\sigma_{-i}(e_1)\Omega_1
 =y \Delta_\rho e_1 \Delta_\rho^{-1} \Omega_1 
 =y \Delta_\rho e_1 \Delta_\rho^{-1} e_2\Omega_1 
 =\Delta_\rho \mathfrak{F}^{-1}(y)\Delta_\rho^{-1}\Omega_1.
    \end{align*}
Note that $\widetilde{y}\in \cM_1$ such that $\widetilde{y}e_2=y\sigma_{-i}(e_1)e_2$.
We see that $\Delta_\rho \mathfrak{F}^{-1}(y)\Delta_\rho^{-1}$ is bounded and $\mathfrak{F}^{-1}(y)\in \Dom(\sigma_{-i})$.
Note that $\mathfrak{F}$ is a unitary transform from $\cN'\cap \cM_1$ onto $\cM'\cap \cM_2$.
We see that $x\in \Dom(\sigma_{-i})$ for any $x\in \cN'\cap \cM_1$.
\end{proof}

\begin{lemma}\label{lem:modular}
    The following inequality holds: 
    \begin{align*}
    \widehat{\Delta}_\rho = \lambda^{-2}\bE_{\cM'}(\Delta_{\rho}^{1/2} e_1e_2 \Delta_{\rho}^{-1} e_2e_1 \Delta_{\rho}^{1/2}) \geq  e_2.
    \end{align*}
    in $\cM'\cap \cM_2$.
    Moreover $e_2\widehat{\Delta}_\rho= \widehat{\Delta}_\rho e_2=e_2$.
\end{lemma}
\begin{proof}
Note that $\sigma_{-i}(e_1)=\Delta_{\rho}e_1\Delta^{-1}_{\rho}$. 
Since $e_1\in \Dom(\Delta_{\rho})$, it is also in $\Dom(\Delta^{1/2}_{\rho})$. 
 \begin{align*}
    \widehat{\Delta}_\rho
    =& \lambda^{-2}\bE_{\cM'}(\bE_{\cN'} (\Delta_{\rho} e_1 \Delta_{\rho}^{-1} )e_2e_1) \\
    =& \lambda^{-2}\bE_{\cM'}(\Delta_{\rho} e_1 \Delta_{\rho}^{-1} e_2e_1) \\
    =& \lambda^{-2}\bE_{\cM'}(\Delta_{\rho}^{1/2} e_1e_2 \Delta_{\rho}^{-1} e_2e_1 \Delta_{\rho}^{1/2}).
 \end{align*}
 Note that 
 \begin{align*}
     \lambda^2 e_2=\bE_{\cM'}(\Delta_{\rho}^{1/2} e_1e_2 \Delta_{\rho}^{-1/2}) \bE_{\cM'}(\Delta_{\rho}^{-1/2} e_2e_1 \Delta_{\rho}^{1/2})\leq \bE_{\cM'}(\Delta_{\rho}^{1/2} e_1e_2 \Delta_{\rho}^{-1} e_2e_1 \Delta_{\rho}^{1/2}). 
 \end{align*}
 and
 \begin{align*}
   \bE_{\cM'}(\Delta_{\rho}^{1/2} e_1e_2 \Delta_{\rho}^{-1} e_2e_1 \Delta_{\rho}^{1/2})e_2
   =&  \lambda  \bE_{\cM'}(\Delta_{\rho}^{1/2} e_1e_2 \Delta_{\rho}^{-1} e_2 \Delta_{\rho}^{1/2})\\
   =& \lambda  \bE_{\cM'}( e_1e_2)\\
   =& \lambda^2 e_2.
 \end{align*}
 This completes the proof of the lemma.
\end{proof}

\begin{lemma}\label{lem:invertible}
We have that $\widehat{\Delta}_\rho$ is invertible.
\end{lemma}
\begin{proof}
Suppose $p$ is a projection in $\cM'\cap \cM_2$ such that $p\widehat{\Delta}_\rho=0$.
Then we have that $p\Delta_\rho^{1/2} e_1e_2\Delta_\rho^{-1} e_2 e_1 \Delta_\rho^{1/2}p=0$.
This implies that $p\Delta_\rho^{1/2} e_1e_2\Delta_\rho^{-1/2}=0$ and $p\Delta_\rho^{1/2} e_1e_2=0$.
Multiplying $e_1$, we obtain that $p\Delta_\rho^{1/2} e_1=0$.
By taking the conditional expectation $\bE_{\cM_1}$, we see that $\bE_{\cM_1}(p)\Delta_\rho^{1/2} e_1=0$.
Note that $\bE_{\cM_1}(p)\in \cM'\cap \cM_1$ and $\Delta_\rho^{1/2}$ is invertible.
We see that $\bE_{\cM_1}(p) e_1=0$.
Applying the Pimsner-Popa basis $\{\eta_j\}_{j=1}^m$ for $\cN \subset \cM$, we see that $\bE_{\cM_1}(p)=0$.
Finally, we see that $p=0$.
Hence $\widehat{\Delta}_\rho$ is invertible.
\end{proof}

\subsection{Equilibrium}

In the section, we shall interpret the equilibrium for bimodule quantum channel in the bimodule setting.

\begin{theorem}\label{thm:equilibrium1}
Suppose that $\Phi$ is a bimodule quantum channel and $\rho$ is a faithful normal state on $\cM$ such that $e_1\in \Dom(\sigma_{-i})$.
Then $\rho$ is an equilibrium state of $\Phi$ if and only if $\bE_{\cM}(e_2 e_1\widehat{\Delta}_\rho \overline{\widehat{\Phi}}e_1e_2)=\lambda^{5/2}$ and 
\begin{align}\label{eq:equi1}
      e_2e_1\overline{\widehat{\Phi}}\Delta_\rho e_1\Delta_\rho^{-1}\Omega_1= e_2e_1\overline{\widehat{\Phi}} \widehat{\Delta}_\rho e_1\Omega_1.
\end{align}
\end{theorem}
\begin{proof}
Suppose that $x\in \cM$.
We have that 
\begin{align*}
    \langle \bE_{\cM}(e_2 e_1 \widehat{\Phi} x e_1 e_2) \Delta_\rho^{1/2}\Omega_1, \Delta_\rho^{1/2} \Omega_1 \rangle =\lambda^{5/2} \langle x\Delta_\rho^{1/2}\Omega_1, \Delta_\rho^{1/2} \Omega_1\rangle.
\end{align*}
Since $\Delta_{\rho}$ is affiliated with $\cM$, $e_2\Delta_\rho^{1/2}\Omega_1=\Delta_\rho^{1/2}\Omega_1$. 
Thus we see that 
\begin{align*}
   \lambda \langle e_2 e_1 \widehat{\Phi} x  e_1  \Delta_\rho^{1/2}\Omega_1, \Delta_\rho^{1/2} \Omega_1 \rangle 
   =\lambda^{5/2} \langle x\Delta_\rho^{1/2}\Omega_1, \Delta_\rho^{1/2} \Omega_1\rangle.
\end{align*}
Taking Fourier transform, we have that 
\begin{align*}
    \langle \mathfrak{F}^{-1} (\overline{\widehat{\Phi}}) x  \Delta_\rho^{1/2}\Omega_1, \Delta_\rho^{1/2} \Omega_1 \rangle 
    =\lambda  \langle x\Delta_\rho^{1/2}\Omega_1, \Delta_\rho^{1/2} \Omega_1\rangle.
\end{align*}
Now we have that 
\begin{align*}
    \int \mathfrak{F}^{-1} (\overline{\widehat{\Phi}}) xe_1 \Delta_\rho d\tau_1' =\lambda \int x\Delta_\rho d\tau_1',
\end{align*}
where $\tau_1'=J_1\tau J_1$ and $\displaystyle \int \cdot d\tau_1$ is the trace-like functional on the Banach space $L^1(\cM_1)$.
By the fact that $e_1\in \Dom(\sigma_{-i})$, we have that 
\begin{align*}
    \int \mathfrak{F}^{-1} (\overline{\widehat{\Phi}}) x \Delta_\rho \sigma_{i}(e_1) d\tau_1'
    =\lambda  \int x\Delta_\rho d\tau_1'.
\end{align*}
Hence for any $\widetilde{x}\in L^1(\cM)$, we have that 
\begin{align*}
    \int \mathfrak{F}^{-1} (\overline{\widehat{\Phi}}) \widetilde{x} \sigma_{i}(e_1) d\tau_1' 
    =\lambda  \int \widetilde{x} d\tau_1'.
\end{align*}
Now we have that 
\begin{align*}
    \int \mathfrak{F}^{-1} (\overline{\widehat{\Phi}}) \widetilde{x} \bE_{\cN'}(\sigma_{i}(e_1) ) d\tau_1' 
    = \lambda \int \widetilde{x} d\tau_1'.
\end{align*}
By taking conditional expectation $\bE_{\cM}$, we have that 
\begin{align*}
    \int   \bE_{\cM}(\bE_{\cN'}(\sigma_{i}(e_1) )\mathfrak{F}^{-1}(\overline{\widehat{\Phi}}) ) \widetilde{x}  d\tau_1' 
    =\lambda \int \widetilde{x} d\tau_1'.
\end{align*}
This implies that 
\begin{align}\label{eq:equilibrium1}
\bE_{\cM}(\bE_{\cN'}(\sigma_{i}(e_1) )\mathfrak{F}^{-1}(\overline{\widehat{\Phi}}) ) =\lambda.
\end{align}
and
\begin{align}\label{eq:equilibrium2}
   \bE_{\cM}(\sigma_{i}(e_1) \mathfrak{F}^{-1}(\overline{\widehat{\Phi}}) ) =\bE_{\cM}( \bE_{\cN'}(\sigma_{i}(e_1) )\mathfrak{F}^{-1}(\overline{\widehat{\Phi}})).
\end{align}
By a direct computation, we have that the Equations \eqref{eq:equilibrium1} and \eqref{eq:equilibrium2} imply that $\Phi$ is equilibrium with respect to $\rho$.

Note that Equation \eqref{eq:equilibrium1} is equivalent to 
\begin{align*}
    e_2e_1\mathfrak{F}^{-1}(\bE_{\cN'}(\sigma_i(e_1))) \overline{\widehat{\Phi}}e_1e_2=\lambda^2 e_2,
\end{align*}
i.e. $e_2e_1\widehat{\Delta}_\rho \overline{\widehat{\Phi}}e_1e_2=\lambda^{3/2} e_2$ and this is equivalent to the first condition.
By taking Fourier transform, we see that Equation \eqref{eq:equilibrium2} is equivalent to the second condition.
This completes the proof of the theorem.
\end{proof}

\begin{remark}
Reformulating Equation \eqref{eq:equi1}, we have that
\begin{align*}
    \cR(\overline{\widehat{\Phi}}e_1 \overline{\widehat{\Phi}}) \Delta_\rho e_1 \Delta_\rho^{-1}\Omega_1 
    =\cR(\overline{\widehat{\Phi}}e_1 \overline{\widehat{\Phi}})\widehat{\Delta}_\rho e_1\Omega_1.
\end{align*}
\end{remark}

\begin{remark}
    In Theorem \ref{thm:equilibrium1}, the first condition $\bE_{\cM}(e_2 e_1\widehat{\Delta}_\rho \overline{\widehat{\Phi}}e_1e_2)=\lambda^{5/2}$ is equivalent to the following statements:
    \begin{enumerate}[(1)]
    \item $\bE_{\cM}(e_2 e_1 \overline{\widehat{\Phi}} \widehat{\Delta}_\rho  e_1e_2)=\lambda^{5/2}$;
        \item $\bE_{\cM}(\bE_{\cN'}(\sigma_{i}(e_1) )\mathfrak{F}^{-1}(\overline{\widehat{\Phi}}) ) =\lambda^{1/2}$;
        \item $\bE_{\cM}(\mathfrak{F}^{-1}(\widehat{\Phi})\bE_{\cN'}(\sigma_{-i}(e_1) ) ) =\lambda^{1/2}$;
        \item $\bE_{\cM}(e_2e_1\widehat{\Delta}_\rho\mathfrak{F}^{-1}(\overline{\widehat{\Phi}}) ) =\lambda $;
        \item $\bE_{\cM_1}(\widehat{\Phi} \overline{\widehat{\Delta}_\rho})=\lambda^{1/2}$.
    \end{enumerate}
    The last statement is obtained from the first condition by $180^\circ$ rotation in the planar algebras.
\end{remark}

\begin{corollary}\label{cor:equilibrium1}
Suppose that $\Phi$ is a bimodule quantum channel and $\rho$ is a hyper-trace on $\cM$ such that $e_1\in \Dom(\sigma_{-i})$.
Then $\Phi$ is equilibrium with respect to $\rho$ if and only if $\bE_{\cM}(e_2 e_1\widehat{\Delta}_\rho \overline{\widehat{\Phi}}e_1e_2)=\lambda^{5/2}$.    
\end{corollary}
\begin{proof}
If $\rho$ is a hyper-trace, then $\Delta_\rho$ is affiliated with $\cN'\cap \cM$.
By Theorem \ref{thm:equilibrium1}, we see that $\Phi$ is equilibrium with respect to $\rho$ if and only if $\bE_{\cM}(e_2 e_1\widehat{\Delta}_\rho \overline{\widehat{\Phi}}e_1e_2)=\lambda^{5/2}$.    
\end{proof}

\subsection{Bimodule GNS Symmetry}
Suppose $\rho$ is a faithful normal state on $\cM$.
A quantum channel $\Phi$ satisfies the $\rho$-detailed balance condition if 
\begin{align}\label{eq:balance}
    \rho(y^*\Phi(x))=\rho(\Phi(y)^*x), \quad x, y\in \cM.
\end{align}
In this case, we also say that the quantum channel $\Phi$ is GNS symmetry with respect to $\rho$.
It is a well-known result that if $\Phi$ is GNS symmetric with respect to $\rho$, then $\Phi$ commutes with the modular group of $\rho$. 
We give a proof of this fact here for completeness.

\begin{lemma}\label{lem:modularinvariant}
Suppose that $\Phi$ is a bimodule quantum channel which is GNS symmetric with respect to a normal faithful state $\rho$.
Then for any $x\in \cM$, 
\begin{align*}
\Phi(\sigma_t(x)) = \sigma_t( \Phi(x) ), \quad t\in \bR. 
\end{align*}
\end{lemma}
\begin{proof}
Suppose that $x\in \Dom(\sigma_{-i})$.
By Equation \eqref{eq:balance}, we have that for any $\cM$,
\begin{align*}
  \langle S_{\rho, \tau} \Phi (x)^*\Omega, S_{\rho, \tau} y^*\Omega\rangle  
  = \langle S_{\rho, \tau} x^* \Omega, S_{\rho, \tau} \Phi (y)^* \Omega\rangle,
\end{align*}
and then
\begin{align*}
  \langle \Delta_\rho^{1/2} \Phi (x)^*\Omega, \Delta_\rho^{1/2} y^*\Omega\rangle  
  = \langle \Delta_\rho^{1/2} x^* \Omega, \Delta_\rho^{1/2}\Phi (y)^* \Omega\rangle.
\end{align*}
Replacing $x$ by $\sigma_{-i}(x)$, we obtain that
\begin{align*}
    \rho(\Phi(\sigma_{-i}(x))^*y)
=& \langle \Delta_\rho^{1/2} \Phi(\Delta_\rho x\Delta_\rho^{-1})^*\Omega, \Delta_\rho^{1/2} y^*\Omega\rangle \\
= & \langle \Delta_\rho^{-1/2} x^* \Delta_\rho\Omega,  \Delta_\rho^{1/2} \Phi (y)^* \Omega\rangle\\
= &  \langle \Delta_\rho ^{1/2} \Phi(y)  \Omega, \Delta_\rho^{1/2} x \Omega \rangle \\
= &   \langle \Delta_\rho^{1/2} y\Omega,  \Delta_\rho ^{1/2}\Phi(x) \Omega\rangle \\
= &  \langle \Phi (x)^* \Delta_\rho ^{1/2} \Omega, y^* \Delta_\rho^{1/2}  \Omega\rangle\\
=& \rho(y\Phi(x)^*).
\end{align*}
This implies that $\Phi(x)\in \Dom(\sigma_{-i})$ and $\Phi(\sigma_{-i}(x)) = \sigma_{-i} \Phi(x). $
Suppose that $x$ is analytic with respect to $\sigma$.
We have that $\Phi(\sigma_{-im}(x)) = \sigma_{-im} \Phi(x)$ for any $m\in \bZ$, where $\bZ$ is the set of all integers.
Therefore, $\Phi(\sigma_{t}(x)) = \sigma_{t} \Phi(x)$ for any $t\in \bR$.
Note that the set of analytic elements is weakly dense in $\cM$.
We have that $\Phi(\sigma_{t}(x)) = \sigma_{t} \Phi(x)$ for any $t\in \bR$ and $x\in \cM$.
This completes the computation.
\end{proof}


In the following, we shall characterize the detailed balance condition in terms of Fourier multiplier.

\begin{theorem}\label{thm:equivbalance}
Suppose $\Phi$ is a bimodule quantum channel and $\rho$ is a normal faithful state on $\cM$ such that $e_1\in \Dom(\sigma_{-i})$.
Then $\Phi$ satisfies the $\rho$-detailed balance condition if and only if $\overline{\widehat{\Phi}}= \widehat{\Phi} \overline{\widehat{\Delta}_\rho}$ and 
\begin{align}\label{eq:condition1}
    \cR(\widehat{\Phi}) \Delta_\rho ^{-1}e_1 \Delta_\rho \Omega_1 = \cR(\widehat{\Phi}) \overline{\widehat{ \Delta}_\rho}   e_1 \Omega_1.
\end{align}
\end{theorem}
\begin{proof}
The $\rho$-detailed balance condition \eqref{eq:balance} implies that 
\begin{align*}
    \langle S_{\rho, \tau} \Phi(x)^*\Omega, S_{\rho, \tau} y^*\Omega\rangle  
  = \langle S_{\rho, \tau} x^* \Omega, S_{\rho, \tau} \Phi(y)^* \Omega\rangle.  
\end{align*}
Reformulating it, we obtain that 
\begin{align*}
  \langle  \Phi (x) \Delta_\rho^{1/2} \Omega, y \Delta_\rho^{1/2}\Omega\rangle  
  = \langle  x \Delta_\rho^{1/2} \Omega, \Phi(y) \Delta_\rho^{1/2}\Omega\rangle.  
\end{align*}
In terms of the Fourier multiplier of $\Phi$, we see that 
\begin{align*}
  \langle  \bE_{\cM}(e_2e_1\widehat{\Phi} x e_1e_2) \Delta_\rho^{1/2} \Omega, y \Delta_\rho^{1/2}\Omega\rangle  
  = \langle  x \Delta_\rho^{1/2} \Omega, \bE_{\cM}(e_2e_1\widehat{\Phi} y e_1e_2) \Delta_\rho^{1/2} \Omega\rangle.  \end{align*}
Now by removing the conditional expectation $\bE_{\cM}$, we have that 
\begin{align*}
\langle  e_2e_1\widehat{\Phi} x e_1e_2 \Delta_\rho^{1/2} \Omega_1, y \Delta_\rho^{1/2} \Omega_1\rangle  
  = \langle  x \Delta_\rho^{1/2} \Omega_1, e_2e_1\widehat{\Phi} y e_1e_2 \Delta_\rho^{1/2} \Omega_1\rangle.      
\end{align*}
Note that $e_2\Omega_1=\Omega_1$.
We have that 
\begin{align*}
\langle  e_1\widehat{\Phi} x e_1 \Delta_\rho^{1/2} \Omega_1, y \Delta_\rho^{1/2}\Omega_1\rangle  
  = \langle   \Delta_\rho^{1/2} \Omega_1, x^* e_1 y \widehat{\Phi} e_1 \Delta_\rho^{1/2} \Omega_1\rangle.      
\end{align*}
Rewriting it in $\cM_1$, we obtain that 
\begin{align*}
\langle \mathfrak{F}^{-1}(\widehat{\Phi})^* x e_1 \Delta_\rho^{1/2} \Omega_1, y \Delta_\rho^{1/2}\Omega_1\rangle  
  = \langle   \Delta_\rho^{1/2} \Omega_1, x^* e_1 y \mathfrak{F}^{-1}(\widehat{\Phi}) \Delta_\rho^{1/2} \Omega_1\rangle.      
\end{align*}
By shifting $e_1$, we have that 
\begin{align*}
\langle \mathfrak{F}^{-1}(\widehat{\Phi})^* x  \Delta_\rho^{1/2} \Omega_1, y \sigma_{-i}(e_1)\Delta_\rho^{1/2}\Omega_1\rangle   
  = \langle   \Delta_\rho^{1/2} \Omega_1, x^* e_1 y \mathfrak{F}^{-1}(\widehat{\Phi}) \Delta_\rho^{1/2} \Omega_1\rangle.      
\end{align*}
By the fact that $\widehat{\Phi} e_1 e_2 =\lambda^{1/2} \mathfrak{F}^{-1}(\widehat{\Phi})e_2$, i.e. Lemma \ref{lem:fourier}, we see that 
\begin{align}\label{eq:bal1}
\langle  \overline{\widehat{\Phi}} e_1 x \Delta_\rho^{1/2} \Omega_1, y\sigma_{-i}(e_1) \Delta_\rho^{1/2} \Omega_1\rangle  
  = \langle  \Delta_\rho^{1/2} \Omega_1, x^* e_1 y \widehat{\Phi}e_1  \Delta_\rho^{1/2} \Omega_1\rangle.      
\end{align}
By taking the conditional expectation $\bE_{\cN'}$, we have that 
\begin{align}\label{eq:bal2}
\langle  \overline{\widehat{\Phi}} y^* e_1 x \Delta_\rho^{1/2} \Omega_1,  \bE_{\cN'}(\sigma_{-i}(e_1)) \Delta_\rho^{1/2} \Omega_1\rangle  
  = \langle  \Delta_\rho^{1/2} \Omega_1, x^* e_1 y \widehat{\Phi}e_1  \Delta_\rho^{1/2} \Omega_1\rangle.      
\end{align}
Combining Equations \eqref{eq:bal1} and \eqref{eq:bal2}, we see that 
\begin{align*}
   \overline{ \widehat{\Phi}} \sigma_{-i}( e_1 )\Delta_\rho^{1/2} \Omega_1
    =\overline{ \widehat{\Phi}}  \bE_{\cN'}(\sigma_{-i}( e_1 )) \Delta_\rho^{1/2} \Omega_1
    = \overline{\widehat{\Phi}}\widehat{\Delta}_\rho e_1 \Delta_\rho^{1/2} \Omega_1
    =\widehat{\Phi}e_1\Omega_1.
\end{align*}
This implies that 
\begin{align*}
\overline{ \widehat{\Phi}} \sigma_{-i}( e_1 )e_2
= \overline{ \widehat{\Phi}}  \bE_{\cN'}(\sigma_{-i}( e_1 ))e_2 =\overline{\widehat{\Phi}}\widehat{\Delta}_\rho e_1e_2=\widehat{\Phi}e_1e_2.
\end{align*}
This indicates that the condition \eqref{eq:condition1} holds.
Note that $\overline{\widehat{\Phi}}\widehat{\Delta}_\rho e_1 e_2=   \widehat{\Phi} e_1e_2$ implies that $\overline{\widehat{\Phi}}\widehat{\Delta}_\rho e_1=   \widehat{\Phi} e_1$
Applying the Pimsner-Popa basis, we see that $\overline{\widehat{\Phi}}\widehat{\Delta}_\rho=   \widehat{\Phi}$.
By taking contragredient, we obtain that 
\begin{align*}
    \overline{\widehat{\Phi}} =  \widehat{\Phi}\overline{\widehat{\Delta}_\rho}.
\end{align*}

Suppose $\overline{\widehat{\Phi}} = \widehat{\Phi} \overline{\widehat{\Delta}_\rho}$.
By the previous computation, we see that 
\begin{align*}
\langle  \overline{\widehat{\Phi}} y^* e_1 x \Delta_\rho^{1/2} \Omega_1,  \bE_{\cN'}(\sigma_{-i}(e_1)) \Delta_\rho^{1/2} \Omega_1\rangle  
  = \langle  \Delta_\rho^{1/2} \Omega_1, x^* e_1 y \widehat{\Phi}e_1  \Delta_\rho^{1/2} \Omega_1\rangle.      
\end{align*}
By the fact that  Equation \eqref{eq:condition1} holds, we see that $\Phi$ satisfies the $\rho$-detailed balance condition.
\end{proof}

\begin{remark}
    Suppose that $\Phi$ is GNS symmetry with respect to a normal faithful state $\rho$ with $e_1\in \Dom(\sigma_{-i})$.
    If $\widehat{\Phi}$ is invertible, by Theorem \ref{thm:equivbalance}, we have that $\sigma_{i}(e_1)=\mathfrak{F}^{-1}(\widehat{\Delta}_\rho)^*$.
\end{remark}

\begin{proposition}
Suppose $\Phi$ is a bimodule quantum channel and $\rho$ is a hyper-trace on $\cM$ over $\cN$ such that $e_1\in \Dom(\sigma_{-i})$.
Then $\Phi$ satisfies $\rho$-detailed balance condition if and only if $\overline{\widehat{\Phi}}= \widehat{\Phi} \overline{\widehat{\Delta}_\rho}$.
\end{proposition}
\begin{proof}
    Note that $\rho$ is hyper-trace.
    We have that $\Delta_\rho e_1 \Delta_\rho^{-1}=\bE_{\cN'}(\Delta_\rho e_1 \Delta_\rho^{-1})$.
    By Theorem \ref{thm:equivbalance}, we see that the proposition holds.
\end{proof}

\begin{remark}
Suppose $\Phi$ is a bimodule quantum channel and $\rho$ is a hyper-trace on $\cM$ with $e_1\in \Dom(\sigma_{-i})$.
Then the $\rho$-detailed balance condition $\overline{\widehat{\Phi}}= \widehat{\Phi}\overline{\widehat{\Delta}_\rho}$ is depicted as
\begin{align*}
      \vcenter{\hbox{\begin{tikzpicture}[scale=1.2]
        \draw [blue] (0, -0.5)--(0, 0.5) (0.5, -0.5)--(0.5, 0.5);        
        \begin{scope}[shift={(0.25, 0)}]
         \draw [fill=white] (-0.45, -0.2) rectangle (0.45, 0.2);
        \node at (0, 0) {\tiny $\overline{\widehat{\Phi}}$};   
        \end{scope}
    \end{tikzpicture}}}=
    \vcenter{\hbox{\begin{tikzpicture}[scale=1.2]
        \draw [blue] (0, -1)--(0, 0.5) (0.5, -1)--(0.5, 0.5);
        \draw [fill=white] (-0.2, -0.2) rectangle (0.2, 0.2);
        \node at (0, 0) {\tiny $\overline{\Delta_\rho^{-1}}$};
        \begin{scope}[shift={(0.5, 0)}]
         \draw [fill=white] (-0.2, -0.2) rectangle (0.2, 0.2);
        \node at (0, 0) {\tiny $\Delta_\rho$};   
        \end{scope}
        \begin{scope}[shift={(0.25, -0.5)}]
         \draw [fill=white] (-0.45, -0.2) rectangle (0.45, 0.2);
        \node at (0, 0) {\tiny $\widehat{\Phi}$};   
        \end{scope}
    \end{tikzpicture}}}.
\end{align*}
\end{remark}



\begin{proposition}\label{prop:equistate}
Suppose $\Phi$ is a bimodule quantum channel and $\rho$ is a normal faithful state on $\cM$ with $e_1\in \Dom(\sigma_{-i})$.
Then $\Phi$ satisfies the $\rho$-detailed balance condition if and only if $\overline{\widehat{\Phi}}= \widehat{\Phi}\overline{\widehat{\Delta}_\rho}$ and
\begin{align}\label{eq:condition2}
    \Delta_\rho ^{-1}\mathfrak{F}^{-1}(\widehat{\Phi}) \Delta_\rho \Omega_1 
    = \overline{\widehat{ \Delta}_\rho}   \mathfrak{F}^{-1}(\widehat{\Phi}) \Omega_1.
\end{align}    
\end{proposition}
\begin{proof}
Suppose $\Phi$ satisfies the $\rho$-detailed balance condition.
Then 
 \begin{align*}
 \widehat{\Phi} \Delta_\rho^{-1} e_1 \Delta_\rho \Omega_1
    = \widehat{\Phi} \overline{\widehat{ \Delta}_\rho} e_1 \Omega_1.
 \end{align*}
 Applying Lemma \ref{lem:fourier}, we see that Equation \eqref{eq:condition2} holds.

Suppose that Equation \eqref{eq:condition2} holds.
 We see that $\Phi$ satisfies the $\rho$-detailed balance condition.

\end{proof}

\begin{remark}
By Proposition \ref{prop:equistate}, we see that $\sigma_{i}(\mathfrak{F}^{-1}(\widehat{\Phi}))\in \cN'\cap \cM_1$.
Note that we do not know if $\sigma_{-i}(e_1)\in \cN'\cap \cM_1$.
\end{remark}

\begin{corollary}\label{cor:modular2}
Suppose $\Phi$ is a bimodule quantum channel satisfying the $\rho$-detailed balance condition, where $\rho$ is a normal faithful state on $\cM$ with $e_1\in \Dom(\sigma_{-i})$. 
Then 
\begin{align}
    \widehat{\Phi}\widehat{\Delta}_\rho =\widehat{\Delta}_\rho  \widehat{\Phi}, \quad \widehat{\Phi} \overline{\widehat{\Delta}_\rho} =\overline{\widehat{\Delta}_\rho } \widehat{\Phi}.
\end{align}
\end{corollary}
\begin{proof}
 By the fact that $\overline{\widehat{\Phi}} =  \widehat{\Phi} \overline{\widehat{\Delta}_\rho}$, we see that $\widehat{\Phi} \overline{\widehat{\Delta}_\rho} =\overline{\widehat{\Delta}_\rho } \widehat{\Phi}$
 and $\overline{\widehat{\Phi}} =  \widehat{\Delta}_\rho^{-1} \widehat{\Phi}$.
    This implies that $\widehat{\Delta}_\rho^{-1} \widehat{\Phi}= \widehat{\Phi} \widehat{\Delta}_\rho^{-1} $.
    Hence $\widehat{\Phi}\widehat{\Delta}_\rho =\widehat{\Delta}_\rho  \widehat{\Phi}$.
This completes the proof of the corollary.
\end{proof}


Now we suggest the following definition of detailed balanced condition for bimodule quantum channels.
\begin{definition}[Bimodule Detailed Balance Condition]
    Suppose $\Phi$ is a bimodule quantum channel, $\widehat{\Delta} \in \cM'\cap \cM_2$ is strictly positive with $\widehat{\Delta} e_2= e_2$.
    We say $\Phi$ satisfies the bimodule detailed balance condition with respect to $\widehat{\Delta}$ if 
    \begin{align}  
    \overline{\widehat{\Phi}}
    = \widehat{\Phi} \overline{\widehat{\Delta}}.
    \end{align}
    We also say that $\Phi$ satisfies bimodule GNS symmetry with respect to $\widehat{\Delta}$.
\end{definition}

\begin{remark}
Suppose that $\Phi$ is a bimodule quantum channel.
If $\cR(\widehat{\Phi})\neq \cR(\overline{\widehat{\Phi}})$, we see that $\Phi$ is not bimodule GNS symmetry to any strictly positive $\widehat{\Delta}\in \cM'\cap \cM_2$ with $\widehat{\Delta}e_2=e_2$. 
\end{remark}

\begin{remark}
    Suppose that $\cN\subset \cM$ is irreducible and $p$ is a projection in $\cM'\cap \cM_2$ such that $p\overline{p}=0$.
    Let $\Phi$ be a bimodule quantum channel such that $\widehat{\Phi}=\kappa p+\kappa^{-1} \overline{p}$, and $\displaystyle \kappa +\kappa^{-1}=\frac{\lambda^{1/2}}{\tau_2(p)}$.
 We have that $\Phi$ is bimodule GNS symmetry with respect to $\widehat{\Delta}=e_2+\kappa^2 p+\kappa^{-2}\overline{p}+s(1-e_2-p-\overline{p})$, $ s> 0$.
\end{remark}

\begin{example}
Suppose that $\bC\subset \bC^4$ is the finite inclusion.
Let $\Phi: \bC^4\to \bC^4$ be a quantum channel such that  
\begin{align*}
\mathfrak{F}(\widehat{\Phi})=
\begin{pmatrix}
0 & \frac{1}{3} &  \frac{1}{3} &  \frac{1}{3} \\
\frac{1}{2} & 0 & \frac{1}{4} & \frac{1}{4}  \\
\frac{1}{4} & \frac{1}{2} & 0 & \frac{1}{4} \\
 \frac{1}{6} &  \frac{1}{2} &  \frac{1}{3} & 0  
\end{pmatrix}.
\end{align*}
Then
\begin{align*}
\mathfrak{F}(\overline{\widehat{\Phi}})=
\begin{pmatrix}
0 & \frac{1}{2} &  \frac{1}{4} &  \frac{1}{6} \\
\frac{1}{3} & 0 & \frac{1}{2} & \frac{1}{2}  \\
\frac{1}{3} & \frac{1}{4} & 0 & \frac{1}{3} \\
 \frac{1}{3} &  \frac{1}{4} &  \frac{1}{4} & 0  
\end{pmatrix}.
\end{align*}
By taking $\widehat{\Delta}$ as follows
\begin{align*}
\mathfrak{F}(\widehat{\Delta})=\frac{1}{2}
\begin{pmatrix}
1 & \frac{2}{3} &  \frac{4}{3} & 2 \\
\frac{3}{2} & 1 & \frac{1}{2} & \frac{1}{2}  \\
\frac{3}{4} & 2 & 1 & \frac{3}{4} \\
 \frac{1}{2} &  2 &  \frac{4}{3} & 1 
\end{pmatrix}
\end{align*}
If there exists a faithful state $\rho$ on $\cM$ with density $(t_1, \ldots, t_n)^{\mathsf{T}}$ such that $\Phi$ is GNS symmetry with respect to $\rho$, then $\widehat{\Delta}=\widehat{\Delta}_\rho$.
This implies that 
\begin{align*}
t_1t_2^{-1}=\frac{2}{3}, t_1t_3^{-1}=\frac{4}{3}, t_1t_4^{-1}=2, \\
t_2t_3^{-1}=\frac{1}{2}, t_2t_4^{-1}=\frac{1}{2}, 
t_3t_4^{-1}=\frac{1}{4}.
\end{align*}
This implies $t_4=4t_3$ and $\displaystyle t_4=\frac{2}{3}t_3$ which leads contradiction.
Hence $\Phi$ is not GNS symmetric with respect to any faithful state.
\end{example}

\begin{proposition}\label{prop:gnsprop}
Suppose $\Phi$ is a bimodule quantum channel bimodule GNS symmetry with respect to $\widehat{\Delta} \in \cM'\cap \cM_2$.  
Then  
\begin{enumerate}[(1)]
    \item $\widehat{\Phi} \overline{\widehat{\Delta} } =\overline{\widehat{\Delta}  } \widehat{\Phi}.$
    \item $\widehat{\Phi}\widehat{\Delta}  =\widehat{\Delta}   \widehat{\Phi}$.
    \item $\widehat{\Phi}\overline{\widehat{\Phi}}=\overline{\widehat{\Phi}}\widehat{\Phi}$. 
    \item $\mathcal{R}(\widehat{\Phi})\widehat{\Delta}\overline{\widehat{\Delta}} = \mathcal{R}(\widehat{\Phi})\overline{\widehat{\Delta}}\widehat{\Delta} = \mathcal{R}(\widehat{\Phi})$. 
\end{enumerate}
\end{proposition}
\begin{proof}
(1) follows from $\overline{\widehat{\Delta}  } \widehat{\Phi} = \overline{\widehat{\Phi}}^* = \overline{\widehat{\Phi}} = \widehat{\Phi} \overline{\widehat{\Delta} }$. (2) follows the same argument applied to $\widehat{\Phi}\widehat{\Delta}^{-1} = \overline{\widehat{\Phi}}$. 
(3) follows from $\widehat{\Phi}\overline{\widehat{\Phi}} = \widehat{\Phi}\widehat{\Phi}\overline{\widehat{\Delta}} = \widehat{\Phi}\overline{\widehat{\Delta}}\widehat{\Phi} = \overline{\widehat{\Phi}}\widehat{\Phi}$. 
For (4), note that we have 
\begin{align*}
    \widehat{\Phi}\widehat{\Delta}\overline{\widehat{\Delta}} = \widehat{\Delta}\widehat{\Phi}\overline{\widehat{\Delta}} = \widehat{\Delta}\overline{\widehat{\Phi}} = \widehat{\Phi}.
\end{align*}
Also, $\widehat{\Phi}\overline{\widehat{\Delta}}\widehat{\Delta} = \overline{\widehat{\Phi}}\widehat{\Delta} = \widehat{\Phi}$. 
Thus the conclusion follows. 
\end{proof}

\begin{proposition}
Suppose that $\Phi$ is a bimodule quantum channel and $\rho$ is a normal faithful state on $\cM$ such that $e_1\in \Dom(\sigma_{-i})$. Then we have that $\Phi$ is bimodule GNS symmetry with respect to $\widehat{\Delta}_\rho$.   
\end{proposition}
\begin{proof}
    It follows from Theorem \ref{thm:equivbalance}.
\end{proof}

\subsection{Bimodule KMS Symmetry} 

Suppose that $\rho$ is a faithful normal state on $\cM$ and $\Phi$ is a bimodule quantum channel.
The bimodule quantum channel $\Phi$ is KMS symmetry with respect to $\rho$ if 
\begin{align*}
    \langle J\Phi(x)^*J\Delta_\rho^{1/2}\Omega, y\Delta_\rho^{1/2}\Omega\rangle = \langle Jx^*J\Delta_\rho^{1/2}\Omega, \Phi(y)\Delta_\rho^{1/2}\Omega\rangle,
\end{align*}
whenever $x, y\in \cM$.
In this section, we shall introduce bimodule KMS symmetry for bimodule quantum channels.
Suppose that $e_1\in \Dom(\sigma_{-i/2}).$
Let 
\begin{align*}
\widehat{\Delta}_{\rho,1/2}
=\lambda^{-1/2}\mathfrak{F}( \bE_{\cN'} (\Delta_{\rho}^{1/2} e_1 \Delta_{\rho}^{-1/2} ) ). 
\end{align*}
By a similar argument in Lemmas \ref{lem:modular} and 
 \ref{lem:invertible} , we have that 
\begin{enumerate}[(1)]
    \item $\widehat{\Delta}_{\rho,1/2}\in \cM'\cap \cM_2$ is positive and $\widehat{\Delta}_{\rho,1/2} e_2=e_2$.
    \item $\widehat{\Delta}_{\rho,1/2}$ is invertible.
\end{enumerate}

\begin{theorem}\label{thm:kms}
Suppose that $\Phi$ is a bimodule quantum channel and $\rho$ is a normal faithful state on $\cM$ with $e_1\in \Dom(\sigma_{-i/2})$.
    Then $\Phi$ is KMS symmetric with respect to $\rho$ if and only if $\widehat{\Phi}\overline{\widehat{\Delta}_{\rho,1/2}}=\widehat{\Delta}_{\rho,1/2} \overline{\widehat{\Phi}}$, $\cR(\widehat{\Phi})\widehat{\Delta}_{\rho, 1/2}^{-1}=\cR(\widehat{\Phi})\overline{\widehat{\Delta}_{\rho, 1/2}}$ and 
    \begin{align*}
     \cR(\widehat{\Phi}) \Delta_\rho ^{-1/2}e_1 \Delta_\rho^{1/2} \Omega_1 = &   \cR(\widehat{\Phi}) \overline{\widehat{ \Delta}_{\rho, 1/2}}   e_1 \Omega_1.
    \end{align*}
    Equivalently,
        \begin{align*}
    \Delta_\rho ^{-1/2} \mathfrak{F}^{-1}(\widehat{\Phi}) \Delta_\rho^{1/2} \Omega_1 
    = & \overline{\widehat{\Delta}_{\rho, 1/2} }  
 \mathfrak{F}^{-1}(\widehat{\Phi})^* \Omega_1.
    \end{align*}
\end{theorem}

\begin{proof}
We will follow a similar argument in Theorem \ref{thm:equivbalance}.
For any $x, y\in \cM$, we have
\begin{align}\label{eq:kms0}
    \langle \Delta_\rho^{1/2} \Phi(x)\Omega_1, y\Delta_\rho^{1/2}\Omega_1\rangle = \langle \Delta_\rho^{1/2}x\Omega_1, \Phi(y)\Delta_\rho^{1/2}\Omega_1\rangle, 
\end{align}
Reformulating it in terms of the Fourier multiplier, we obtain that 
\begin{align*}
    \langle \Delta_\rho^{1/2} e_1\widehat{\Phi}xe_1\Omega_1, y\Delta_\rho^{1/2}\Omega_1\rangle = \langle \Delta_\rho^{1/2}x\Omega_1, e_1\widehat{\Phi}ye_1\Delta_\rho^{1/2}\Omega_1\rangle.
\end{align*}
Assuming that $x\in \Dom(\sigma_{-i/2})$, we obtain that 
\begin{align*}
    \langle y^*  \sigma_{-i/2}(e_1) \sigma_{-i/2}(x) \Delta_\rho^{1/2} \Omega_1, \Delta_\rho^{1/2} \mathfrak{F}^{-1}(\widehat{\Phi})^*\Omega_1\rangle 
= \langle y^* e_1\sigma_{-i/2}(x) \Delta_\rho^{1/2}\Omega_1, \widehat{\Phi}e_1\Delta_\rho^{1/2}\Omega_1\rangle.
\end{align*}
Hence
\begin{align*}
\langle y^*  \sigma_{-i/2}(e_1)  x \Delta_\rho^{1/2} \Omega_1, \Delta_\rho^{1/2} \overline{\widehat{\Phi}}e_1\Omega_1\rangle 
= \langle y^* e_1 x\Delta_\rho^{1/2}\Omega_1, \widehat{\Phi}e_1\Delta_\rho^{1/2}\Omega_1\rangle.
\end{align*}
By taking the conditional expectation $\bE_{\cN'}$ and the Fourier transform, we have that 
\begin{align}\label{eq:kms11}
\langle y^*  \widehat{\Delta}_{\rho, 1/2}e_1  x \Delta_\rho^{1/2} \Omega_1, \Delta_\rho^{1/2} \overline{\widehat{\Phi}}e_1\Omega_1\rangle 
= \langle y^* e_1 x\Delta_\rho^{1/2}\Omega_1, \widehat{\Phi}e_1\Delta_\rho^{1/2}\Omega_1\rangle.    
\end{align}
Meanwhile, we have that 
\begin{align}
   \overline{\widehat{\Phi}}\sigma_{-i/2}(e_1)e_2= \overline{\widehat{\Phi}} \widehat{\Delta}_{\rho, 1/2}e_1e_2.
\end{align}
This is equivalent to 
\begin{align*}
\overline{\widehat{\Phi}}\sigma_{-i/2}(e_1)\Omega_1= \overline{\widehat{\Phi}} \widehat{\Delta}_{\rho, 1/2}e_1\Omega_1,
\end{align*}
or
\begin{align}\label{eq:kms12}
\widehat{\Phi}\sigma_{i/2}(e_1)\Omega_1= \widehat{\Phi} \overline{\widehat{\Delta}_{\rho, 1/2}}e_1\Omega_1,
\end{align}
by taking the conjugation $J_1$ on both sides.
Similarly, we have that
\begin{align*}
\langle y^*  \widehat{\Delta}_{\rho, 1/2}e_1  x \Delta_\rho^{1/2} \Omega_1,  \overline{\widehat{\Phi}}\widehat{\Delta}_{\rho, 1/2}e_1 \Delta_\rho^{1/2}\Omega_1\rangle 
= \langle y^* e_1 x\Delta_\rho^{1/2}\Omega_1, \widehat{\Phi}e_1\Delta_\rho^{1/2}\Omega_1\rangle.    
\end{align*}
By removing the inner product, we see that 
\begin{align*}
  \widehat{\Delta}_{\rho, 1/2} \overline{\widehat{\Phi}}\widehat{\Delta}_{\rho, 1/2}e_1e_2 
  =\widehat{\Phi}e_1e_2.
\end{align*}
By removing $e_1, e_2$, we obtain that 
\begin{align}\label{eq:kms21}
  \widehat{\Delta}_{\rho, 1/2} \overline{\widehat{\Phi}}\widehat{\Delta}_{\rho, 1/2} 
  =\widehat{\Phi}.  
\end{align}
Reformulating Equation \eqref{eq:kms11}, we see that 
\begin{align*}
 \langle y^*  \widehat{\Delta}_{\rho, 1/2}e_1  x \Delta_\rho^{1/2} \Omega_1,  \overline{\widehat{\Phi}} \Delta_\rho^{1/2} e_1  \Omega_1\rangle 
= \langle y^* e_1 x\Delta_\rho^{1/2}\Omega_1, \widehat{\Phi}\Delta_\rho^{1/2}\overline{\widehat{\Delta}_{\rho, 1/2}}e_1\Omega_1\rangle.   
\end{align*}
This implies that
\begin{align}\label{eq:kms22}
  \widehat{\Delta}_{\rho, 1/2} \overline{\widehat{\Phi}} 
  =\widehat{\Phi}\overline{\widehat{\Delta}_{\rho, 1/2}}.    
\end{align}
Moreover, we have that 
\begin{align}\label{eq:kms13}
    \widehat{\Phi}\sigma_{i/2}(e_1)e_2 =\widehat{\Phi}\overline{\widehat{\Delta}_{\rho,1/2}}e_1e_2.
\end{align}
Combining Equations \eqref{eq:kms21} and \eqref{eq:kms22}, we obtain that 
\begin{align}\label{eq:kms23}
\cR(\widehat{\Phi})\widehat{\Delta}_{\rho, 1/2}^{-1}=\cR(\widehat{\Phi})\overline{\widehat{\Delta}_{\rho, 1/2}}.
\end{align}

If Equations \eqref{eq:kms12}, \eqref{eq:kms21}, \eqref{eq:kms22}, \eqref{eq:kms23} hold, we see that Equation \eqref{eq:kms0} holds.
This completes the proof of the theorem.
\end{proof}


Theorem \ref{thm:kms} inspires us to introduce bimodule KMS symmetry as follows.

\begin{definition}[Bimodule KMS Symmetry]
Suppose $\Phi$ is a bimodule quantum channel and $\widehat{\Delta}\in \cM'\cap \cM_2$ is strictly positive such that $\widehat{\Delta} e_2=e_2$, $\cR(\widehat{\Phi})\overline{\widehat{\Delta}}=\cR(\widehat{\Phi})\widehat{\Delta}^{-1}$.
We say $\Phi$ is bimodule KMS symmetry with respect to $\widehat{\Delta}$ if 
\begin{align}
\overline{\widehat{\Phi}} =\overline{\widehat{\Delta}} \widehat{\Phi}\overline{\widehat{\Delta}}.
\end{align}
\end{definition}

\begin{proposition}
Suppose that $\Phi$ is a bimodule quantum channel and $\rho$ is a normal faithful state on $\cM$ such that $e_1\in \Dom(\sigma_{-i/2})$. Then $\Phi$ is bimodule KMS symmetry with respect to $\widehat{\Delta}_{\rho, 1/2}$.   
\end{proposition}
\begin{proof}
    It follows from Theorem \ref{thm:kms}.
\end{proof}

\section{Bimodule Quantum Markov Semigroups}\label{sec:: Bimodule Quantum Markov Semigroups}

In this section, we study the quantum Markov semigroups for finite inclusion and the associated first-order differential structure. 
In the end of this section, a bimodule version of Poincar\'{e} inequality is proved.

\subsection{Bimodule Quantum Markov Semigroup}

We begin by the definition of bimodule quantum Markov semigroups as follows.
\begin{definition}[Bimodule Quantum Markov Semigroups]
    Suppose $\cN\subset \cM$ is a finite inclusion of finite von Neumann algebras and $\{\Phi_t:\cM\to \cM\}_{t\geq 0}$ is a quantum Markov semigroup.
    We say $\{\Phi_t\}_{t\geq 0}$ is a bimodule quantum Markov semigroup if $\Phi_t$ is a bimodule quantum channel for $t\geq 0$.
\end{definition}

\begin{proposition}\label{prop:uniform}
  Suppose $\{\Phi_t\}_{t\geq 0}$ is a bimodule quantum Markov semigroup for a finite inclusion $\cN\subset \cM$ of finite von Neumann algebras.
  Then $\{\Phi_t\}_{t\geq 0}$ is uniformly continuous.
\end{proposition}
\begin{proof}
    Suppose $\tau$ is a normal faithful tracial state on $\cM$ and $\Omega$ is the cyclic and separating vector in the GNS Hilbert space $L^2(\cM, \tau)$.
    We define $y(t)$ on $L^2(\cM)$ as follows
    \begin{align}\label{eq:tracephi}
        y(t)x\Omega=\Phi_t(x)\Omega, \quad x\in \cM.
    \end{align}
    Then we see that $y(t)\in \cN'\cap \cM_1$.
    Note that $\cN\subset \cM$ is a finite inclusion.
    We find that $\cN'\cap \cM_1$ is finite dimensional.
    Hence $\{y(t)\}_{t\geq 0}$ is a uniformly continuous one-parameter semigroup.
    We denote by $h$ the generator of this semigroup.
    Then $h\in \cN'\cap \cM_1$ and $y(t)=e^{ht}$.
    We have that 
    \begin{align*}
        \|\Phi_t(x)-x\|=& \sup_{\tau(a^*a)=1} \|(\Phi_t(x)-x)a\Omega\| \\
        =& \sup_{\tau(a^*a)=1} \|Ja^* (\Phi_t(x^*)-x^*) \Omega\| \\
        =& \sup_{\tau(a^*a)=1} \|Ja^* (y(t)-1)x^* \Omega\| \\
        =& \lambda^{-1}\sup_{\tau(a^*a)=1} \|Ja^* \bE_\cM((y(t)-1)x^*e_1) \Omega\| \\
        =& \lambda^{-1}\sup_{\tau(a^*a)=1} \| \bE_\cM(e_1 x (y(t)-1)^*) a \Omega\| \\
        \leq & \lambda^{-1} \|\bE_\cM(e_1 x (y(t)-1)^*)\|.
    \end{align*}
    This implies that $\Phi_t$ is uniformly continuous.
\end{proof}

Let $\cL$ be the generator of a bimodule quantum Markov semigroup $\{\Phi_t\}_{t\geq 0}$ such that $e^{-t\cL}=\Phi_t$, which is also called Lindbladian.
By Proposition \ref{prop:uniform}, we see that $\{\Phi_t\}_{t\geq 0}$ is uniformly continuous.
So the generator $\cL$ is a bounded map on $\cM$.
The Fourier multiplier $\widehat{\cL}$ of $\cL$ is as follows
\begin{align*}
    \widehat{\cL}  x e_1 y \Omega_1 = \lambda^{1/2}\sum_{j=1}^m x \eta_j^* e_1\cL(\eta_j)y\Omega_1, \quad \text{ for all } x, y\in\mathcal{M}. 
\end{align*}
We have that $\widehat{\cL}\in \cM'\cap \cM_2$ and $\widehat{\cL}$ is related to the Fourier multipliers $\{\widehat{\Phi}_t\}_{t\geq 0}$ as $\displaystyle \widehat{\cL}=\lim_{t\to 0} \frac{\lambda^{-1/2} e_2- \widehat{\Phi}_t}{t}$.

\begin{proposition}\label{prop:generator}
    Suppose that $\{\Phi_t\}_{t\geq 0}$ is a bimodule quantum Markov semigroup for $\lambda$-extension $\cN\subset \cM$ and $\cL$ is its generator.
    Then
    \begin{enumerate}[(1)]
        \item $\bE_{\cM}(e_2e_1\widehat{\cL}e_1e_2)=0$. Pictorially, $\vcenter{\hbox{\begin{tikzpicture}[scale=0.6]
        \draw [blue] (0.2, -0.8) --(0.2, 0.8);
            \draw [fill=white] (-0.5, -0.4) rectangle (0.5, 0.4);
            \node at (0, 0) {\tiny $\widehat{\mathcal{L}}$};
            \draw [blue] (-0.2, 0.4).. controls +(0, 0.35) and +(0, 0.35)..(-0.8, 0.4)--(-0.8, -0.4);
            \draw [blue] (-0.2, -0.4).. controls +(0, -0.35) and +(0, -0.35)..(-0.8, -0.4);
        \end{tikzpicture}}}=0$.
        \item $\widehat{\cL}^*=\widehat{\cL}$.
        \item $ -(1-e_2)\widehat{\cL}(1-e_2)\geq 0$.
    \end{enumerate}
\end{proposition}
\begin{proof}
$(1)$: 
Let $h$ be the generator of $\{y(t)\}_{t\geq 0}$ defined in Equation \eqref{eq:tracephi}.
We see that $h\in \cN'\cap \cM$ and $y(t)=e^{th}$, $t\geq 0$.
Note that $\Phi_t=e^{-t\cL}$.
We have that $hx\Omega=-\cL(x)\Omega$ for any $x\in \cM$.
By the fact that $\Phi_t(1)=1$, we have that $e^{-t\cL(1)}=1$.
This implies that $\cL(1)=0$ and $h\Omega=0$.
Hence $he_1=0$.
Reformulating it in terms of the Fourier multiplier, we obtain that $\bE_{\cM}(e_2e_1\widehat{\cL}e_1e_2)=0$.

$(2)$:
By the fact that $\Phi_t(x)^*=\Phi_t(x^*)$, we have that $\cL(x)^*=\cL(x^*)$ for all $x\in \cM$.
This implies that $\widehat{\cL}^*=\widehat{\cL}$.

$(3)$: Suppose that for any $n\in\bN$ and $x_j, y_k\in \cM$ with $\displaystyle \sum_{j=1}^n x_jy_j=0$.
Then 
\begin{align*}
  -\sum_{j,k=1}^n y_j^*\cL(x_j^*x_k)y_k =&
  \sum_{j,k=1}^n \lim_{t\to 0} \frac{y_j^*\Phi_t(x_j^*x_k)y_k-y_j^*x_j^*x_ky_k}{t}  \\
  = & \lim_{t\to 0} \sum_{j,k=1}^n 
 \frac{y_j^*\Phi_t(x_j^*x_k)y_k}{t}  \geq 0.
\end{align*}
Reformulating it in terms of the Fourier multiplier, we have that 
\begin{align*}
   0\leq & -\sum_{j,k=1}^n y_j^*\bE_{\cM}(e_2e_1x_j^*\widehat{\cL}x_ke_1e_2)y_k \\
= & -\bE_{\cM}\left(e_2\sum_{j=1}^ny_j^*e_1x_j^*\widehat{\cL}\sum_{k=1}^nx_ke_1y_ke_2\right).
\end{align*}
This implies that for any $\tilde{x}\in \cM_1$ with $\bE_{\cM}(\tilde{x})=0$,  
\begin{align}\label{eq:qm1}
-\bE_{\cM}(e_2 \tilde{x}^* \widehat{\cL} \tilde{x} e_2) \geq 0.
\end{align}
Since $\bE_{\cM}(\tilde{x})=0$ is equivalent to $e_2\tilde{x}e_2=0$, Equation \eqref{eq:qm1} can be reformulated as
\begin{align}\label{eq:qm2}
 -\bE_{\cM}(e_2 \tilde{x}^* (1-e_2) \widehat{\cL} (1-e_2)\tilde{x} e_2) \geq 0.
\end{align}
For any $\tilde{y}\in \cM_1$, we have that $\bE_{\cM}(\tilde{y}-\bE_{\cM}(\tilde{y}))=0$.
Replacing $\tilde{x}$ by $\tilde{y}-\bE_{\cM}(\tilde{y})$ in Equation \eqref{eq:qm2}, we have that
\begin{align}\label{eq:qm3}
     -\bE_{\cM}(e_2 \tilde{y}^* (1-e_2) \widehat{\cL} (1-e_2)\tilde{y} e_2) \geq 0.
\end{align}
Then
\begin{align*}
    0\leq -\langle \bE_{\cM}(e_2 \tilde{y}^* (1-e_2) \widehat{\cL} (1-e_2)\tilde{y} e_2)\Omega_1, \Omega_1\rangle. 
\end{align*}
This implies that 
\begin{align}\label{eq:qm4}
-\tau_2(\tilde{y}^*e_2\tilde{y}(1-e_2) \widehat{\cL} (1-e_2))\geq 0.
\end{align}
Let $\tilde{\eta}_1, \ldots, \tilde{\eta}_m$ be a Pimsner-Popa basis for $\cM\subset \cM_1$.
Then for any $\tilde{x}\in \cM_2$, we have
\begin{align*}
    \tilde{x}=\sum_{k=1}^m x_ke_2 \tilde{\eta}_k, \quad x_{k}\in \cM_1.
\end{align*}
This implies that 
\begin{align*}
  \tilde{x}\tilde{x}^* =  \sum_{k=1}^m x_ke_2 \tilde{\eta}_k\tilde{\eta}_k^* e_2 x_k^*.
\end{align*}
Finally, we see that Equation \eqref{eq:qm4} implies that $-(1-e_2) \widehat{\cL} (1-e_2)\geq 0$.
\end{proof}

Suppose $\cL:\cM\to \cM$ is bounded bimodule map.
We say $\cL$ is conditional negative if  $\displaystyle -\sum_{j,k=1}^n y_j^*\cL(x_j^*x_k)y_k \geq 0$ for any $\displaystyle \sum_{j=1}^n x_j y_j=0$, where $x_1, \ldots, x_n$, $y_1, \ldots, y_n\in \cM$.
This is equivalent to the following condition:
\begin{align*}
    -y^*\cL_{n}(x^*x)y \geq 0, \quad x, y\in \cM\otimes M_n(\bC) \text{ with } xy=0,
\end{align*}
where $\cL_n=\cL\otimes I_n$ on $\cM\otimes M_n(\bC)$ for $n \in \bN$. 
The gradient form $\Gamma$ of $\{\Phi_t=e^{-t\mathcal{L}}\}_{t\geq 0}$ is defined as 
\begin{align*}
    \Gamma(x, y)=\frac{1}{2} (y^*\cL(x)+\cL(y)^*x-\cL(y^*x)), \quad x, y\in \cM.
\end{align*}
Let
\begin{align*}
    \Gamma_n(x, y)=\frac{1}{2} (y^*\cL_n(x)+\cL_n(y)^*x-\cL_n(y^*x)), \quad x, y\in \cM\otimes M_n(\bC).
\end{align*}
We say $\cL$ is completely dissipative if for any $n \in \bN$ and $x\in \cM\otimes M_n(\bC)$, $\Gamma_n(x, x) \geq 0$, i.e.
\begin{align*}
    x^*\cL_n(x)+\cL_n(x)^*x -\cL_n(x^*x) \geq 0.
\end{align*}

\begin{proposition}\label{prop:positvity1}
Suppose $\cL:\cM\to \cM$ is a bounded bimodule map with $\cL(1)=0$ and $\widehat{\cL}^*=\widehat{\cL}$. 
Then the following are equivalent:
\begin{enumerate}[(1)]
\item $\Phi_t = e^{-t\mathcal{L}}$ is completely positive for $t \geq 0$ and $\Phi_t(1)=1$;
\item $\cL$ is completely dissipative;
\item $\cL$ is conditionally negative;
\item $-(1-e_2)\widehat{\cL}(1-e_2)\geq 0$.
\end{enumerate}
\end{proposition}
\begin{proof}
By Proposition \ref{prop:generator}, we see that (1)$\Rightarrow$(3)$\Rightarrow$(4) is true.

(3)$\Rightarrow$(2):
Suppose that $x\in \cM\otimes M_n(\bC)$.
Let $\tilde{x}=\begin{pmatrix}
    x & 1 \\0 & 0
\end{pmatrix}$
and $\tilde{y}=\begin{pmatrix}
    -1 & 0 \\ x & 0
\end{pmatrix}$.
Then we have that $\tilde{y}^*\cL_{2n}(\tilde{x}^*\tilde{x})\tilde{y} \geq 0$.
This implies that $\tilde{y}^*\cL_n(\tilde{x})+\cL_n(\tilde{y})^*\tilde{x}-\cL_n(\tilde{y}^*\tilde{x}))\geq 0$.
Hence we see that $\cL$ is completely dissipative.

(2)$\Rightarrow$(3):
Suppose $x_1, \ldots, x_n$, $y_1, \ldots, y_n\in \cM$ with $\displaystyle \sum_{j=1}^nx_jy_j=0$. 
Let $x=(x_1, \ldots, x_n)$ and $y=(y_1, \ldots, y_n)$.
Then $x^*\cL_n(x)+\cL_n(x)^*x-\cL_n(x^*x) \geq 0$ and
\begin{align*}
    y^*x^*\cL_n(x)y+y^*\cL_n(x)^*xy-y^*\cL_n(x^*x)y \geq 0. 
\end{align*}
This implies that $\cL$ is conditional negative.

(2)$\Rightarrow$(1): 
Note that $\displaystyle \|I+t\cL\|=\sup_{u\in \mathscr{U}(\cM)}\|u+t\cL(u)\|$, where $\mathscr{U}(\cM)$ is the set of all unitaries in $\cM$.
Then by the fact that $\cL$ is completely dissipative, we have that $\|I+t\cL\|\leq 1+t^2\|\cL\|^2$.
This implies that $\displaystyle \lim_{t\to 0 }\frac{\|I+t\cL\|-1}{t}\leq 0$.
Hence $\Phi_t$ is contractive.
By the fact that $\cL(1)=0$, we see that $\Phi_t(1)=1$.
This implies that $\Phi_t$ is positive.

(4)$\Rightarrow$(3):
The assumption implies that Equation \eqref{eq:qm2} is true.
The Equation \eqref{eq:qm1} is true. 
Hence $\cL$ is conditionally negative.
\end{proof}

\begin{remark}
The equivalence of (1), (2) is discussed in \cite{Lin76}.  
\end{remark}

\begin{theorem}\label{thm:generatorpos}
Suppose $u\in \cM'\cap \cM_2$ is self-adjoint and satisfies the following two conditions: 
\begin{enumerate}[(1)]
        \item $\bE_{\cM}(e_2e_1ue_1e_2)=0$.
        \item $ -(1-e_2)u(1-e_2)\geq 0$.
\end{enumerate}
Then there exists a bimodule quantum Markov semigroup $\{\Phi_t\}_{t\geq 0}$ with generator $\cL$ such that $u=\widehat{\cL}$.
\end{theorem}
\begin{proof}
We define a map $\cL:\cM\to \cM$ such that 
\begin{align*}
    \cL(x)=\lambda^{-5/2} \bE_{\cM}(e_2e_1 u x e_1e_2), \quad x\in \cM.
\end{align*}
Then $\Phi_t=e^{-t\cL}$ is a bimodule semigroup.
By Proposition \ref{prop:positvity1}, we see that $\Phi_t$ is completely positive.
Then $\{\Phi_t\}_{t\geq 0}$ is a bimodule quantum Markov semigroup.
\end{proof}

We define the two components of $\widehat{\cL}$ as follows:
\begin{align*}
    \widehat{\cL}_0= & -(1-e_2)\widehat{\cL}(1-e_2). \\
    \widehat{\cL}_1 =& e_2\widehat{\cL}(1-e_2).
\end{align*}
By the fact that $\lambda^{-3/2}\bE_{\cM}(e_2e_1\widehat{\cL}e_1e_2)=\lambda\vcenter{\hbox{\begin{tikzpicture}[scale=0.6]
        \draw [blue] (0.2, -0.8) --(0.2, 0.8);
            \draw [fill=white] (-0.5, -0.4) rectangle (0.5, 0.4);
            \node at (0, 0) {\tiny $\widehat{\mathcal{L}}$};
            \draw [blue] (-0.2, 0.4).. controls +(0, 0.35) and +(0, 0.35)..(-0.8, 0.4)--(-0.8, -0.4);
            \draw [blue] (-0.2, -0.4).. controls +(0, -0.35) and +(0, -0.35)..(-0.8, -0.4);
        \end{tikzpicture}}}=0$, we have that 
\begin{align*}
 \lambda^2 \vcenter{\hbox{\begin{tikzpicture}[scale=0.6]
        \draw [blue] (0.7, -0.8) --(0.7, 0.8);
            \draw [fill=white] (-0.5, -0.4) rectangle (0.5, 0.4);
            \node at (0, 0) {\tiny $\widehat{\mathcal{L}}$};
            \draw [blue] (-0.2, 0.4).. controls +(0, 0.35) and +(0, 0.35)..(0.2, 0.4);
            \draw [blue] (-0.2, -0.4).. controls +(0, -0.35) and +(0, -0.35)..(0.2, -0.4);
        \end{tikzpicture}}} 
                + \lambda \vcenter{\hbox{\begin{tikzpicture}[scale=0.6]
            \draw [blue] (-0.2, 0.4).. controls +(0, 0.35) and +(0, 0.35)..(-0.8, 0.4)--(-0.8, -0.8);
            \draw [blue] (0.2, 0.8)--(0.2, -0.4).. controls +(0, -0.35) and +(0, -0.35)..(-0.2, -0.4);
            \draw [fill=white] (-0.5, -0.4) rectangle (0.5, 0.4);
            \node at (0, 0) {\tiny $\widehat{\mathcal{L}}_1$};
        \end{tikzpicture}}}
                +  \lambda \vcenter{\hbox{\begin{tikzpicture}[scale=0.6] 
            \draw [blue] (-0.2, -0.4).. controls +(0, -0.35) and +(0, -0.35)..(-0.8, -0.4)--(-0.8, 0.8);
            \draw [blue] (0.2, -0.8)--(0.2, 0.4).. controls +(0, 0.35) and +(0, 0.35)..(-0.2, 0.4);
            \draw [fill=white] (-0.5, -0.4) rectangle (0.5, 0.4);
            \node at (0, 0) {\tiny $\widehat{\mathcal{L}}_1^*$};
        \end{tikzpicture}}}
  - \vcenter{\hbox{\begin{tikzpicture}[scale=0.6]
        \draw [blue] (0.2, -0.8) --(0.2, 0.8);
            \draw [fill=white] (-0.5, -0.4) rectangle (0.5, 0.4);
            \node at (0, 0) {\tiny $\widehat{\mathcal{L}}_0$};
            \draw [blue] (-0.2, 0.4).. controls +(0, 0.35) and +(0, 0.35)..(-0.8, 0.4)--(-0.8, -0.4);
            \draw [blue] (-0.2, -0.4).. controls +(0, -0.35) and +(0, -0.35)..(-0.8, -0.4);
        \end{tikzpicture}}}=0, 
\end{align*}
i.e.
\begin{align*}
\lambda\vcenter{\hbox{\begin{tikzpicture}[scale=0.6]
        \draw [blue] (0.2, -0.8) --(0.2, 0.8);
            \draw [fill=white] (-0.5, -0.4) rectangle (0.5, 0.4);
            \node at (0, 0) {\tiny $\widehat{\mathcal{L}}_0$};
            \draw [blue] (-0.2, 0.4).. controls +(0, 0.35) and +(0, 0.35)..(-0.8, 0.4)--(-0.8, -0.4);
            \draw [blue] (-0.2, -0.4).. controls +(0, -0.35) and +(0, -0.35)..(-0.8, -0.4);
            \draw [blue] (0.8, -0.8) --(0.8, 0.8);
        \end{tikzpicture}}}
        =&\lambda^{1/2}\bE_{\cM}(e_2 \widehat{\cL} e_2)+\lambda^{-1/2} \bE_{\cM}(e_2e_1\widehat{\cL}_1^*)+ \lambda^{-1/2} \bE_{\cM}(\widehat{\cL}_1e_1e_2),  \\
        =&\lambda^{1/2}\bE_{\cM}(e_2 \widehat{\cL} e_2) + \lambda \bE_{\cM}(\mathfrak{F}^{-1}(\widehat{\cL}_1))+ \lambda \bE_{\cM}(\mathfrak{F}(\widehat{\cL}_1^*)).
\end{align*}
where $\bE_{\cM}(\widehat{\cL}_1e_1e_2)\in \cN'\cap \cM$.

In what follows we shall impose that $Z(\cM) = \mathbb{C}$, i.e. $\cM$ is a factor. 
This condition is equivalent to $e_2$ being a minimal projection in $\cM'\cap \cM_2$. 
Let $\mathscr{S}=z_{e_2}\cM'\cap \cM_2$, where $z_{e_2}$ is the central carrier of the Jones projection $e_2$ in $\cM'\cap \cM_2$.
We have that $\mathscr{S}$ is isomorphic to a matrix algebra.
When $e_2\widehat{\cL}=\widehat{\cL}e_2$, we have that $\widehat{\cL}_1=0$.
In particular, when $\mathcal{N}\subset \cM$ is irreducible, i.e. $\cN'\cap \cM=\bC$, we have that $\widehat{\cL}_1=0$.
In general, we have that $\widehat{\cL}_1=e_2 \widehat{\cL} z_{e_2}$. 

The spectral decomposition of $\widehat{\cL}_0$ is given as follows
\begin{align*}
  \widehat{\cL}_0=\sum_{j\in \sI_0\cup\sI_1}\omega_{j} p_j ,
\end{align*}
where $\{p_j\}_{j\in \sI_0}$ is an orthogonal family of minimal projections in $(1-e_2)z_{e_2} \cM'\cap \cM_2$ and $\{p_j\}_{j\in \sI_1}$ is an orthogonal family of minimal projections in $(1-z_{e_2} )\cM'\cap \cM_2$.
We denote by $\sI=\sI_0\cup \sI_1$ for convenience.
For each $j\in \sI_0$, there exists $v_j\in \cN'\cap \cM$ such that 
\begin{align*}
 \vcenter{\hbox{\begin{tikzpicture}[scale=0.6]
        \draw [blue] (0.2, -0.8) --(0.2, 0.8);
        \draw [blue] (-0.2, -0.8) --(-0.2, 0.8);
        \draw [fill=white] (-0.5, -0.4) rectangle (0.5, 0.4);
    \node at (0, 0) {\tiny $p_j$};
        \end{tikzpicture}}}  
    =\vcenter{\hbox{\begin{tikzpicture}[scale=0.65]
    \begin{scope}[shift={(0,1.5)}]
    \draw [blue] (-0.5, 0.8)--(-0.5, 0) .. controls +(0, -0.6) and +(0,-0.6).. (0.5, 0)--(0.5, 0.8);    
\begin{scope}[shift={(0.5, 0.3)}]
\draw [fill=white] (-0.3, -0.3) rectangle (0.3, 0.3);
\node at (0, 0) {\tiny $v_j$};
\end{scope}
    \end{scope}
\draw [blue] (-0.5, -0.8)--(-0.5, 0) .. controls +(0, 0.6) and +(0,0.6).. (0.5, 0)--(0.5, -0.8);
\begin{scope}[shift={(0.5, -0.3)}]
\draw [fill=white] (-0.3, -0.3) rectangle (0.3, 0.3);
\node at (0, 0) {\tiny $v_j*$};
\end{scope}
\end{tikzpicture}}},
\end{align*}
where the algebraic expression of the right hand side is 
\begin{align*}
 \lambda^{-1/2}  \Theta(v_j^*) e_2 \Theta(v_j)
  =\lambda^{-13/2}\bE_{\cM_1}(e_2e_1\bE_{\cM'}(v_j^*e_2e_1))e_2\bE_{\cM_1}(e_2e_1\bE_{\cM'}(v_je_2e_1)).
\end{align*} 
Note that we also denote $\Theta(v_j)$ by $\overline{v_j}\in \cM'\cap \cM_1$.
Then by the fact that $p_j e_2=0$ and $p_jp_{j'}=\delta_{j,j'}p_j$, we have that 
\begin{align*}
    \tau(v_j)=0, \quad \tau(v_jv_{j'}^*)= \lambda^{1/2}\delta_{j,j'}.
\end{align*}
Moreover, each $p_j$ implements a completely positive bimodule map on $\cM$ as $x*p_j = v^*_jxv_j$ for $x\in\cM$. 
\begin{proposition}\label{prop:generatorformula}
Suppose $\{\Phi_t\}_{t\geq 0}$ is a bimodule quantum Markov semigroup.
Then for $x\in \cM$,
\begin{align}
 \cL(x)=\lambda^{-1/2}\bE_{\cM}(e_2 \widehat{\cL} e_2)x+\lambda^{-3/2} \bE_{\cM}(e_2e_1\widehat{\cL}_1^*) x+ \lambda^{-3/2} x \bE_{\cM}(\widehat{\cL}_1 e_1 e_2)-  x*\widehat{\cL}_0.   
\end{align}
\end{proposition}
\begin{proof}
We have that 
\begin{align*}
    \cL(x)=&\lambda^{-5/2}\bE_{\cM}(e_2 e_1 \widehat{\cL} x e_1 e_2) \\
    =& \lambda^{-5/2}\bE_{\cM}(e_2 e_1 e_2\widehat{\cL} e_2 x e_1 e_2) + \lambda^{-5/2}\bE_{\cM}(e_2 e_1 e_2\widehat{\cL} (1-e_2)x e_1 e_2) +\\
    & + \lambda^{-5/2}\bE_{\cM}(e_2 e_1 (1-e_2)\widehat{\cL} x e_2e_1 e_2) 
    +\lambda^{-5/2}\bE_{\cM}(e_2 e_1 (1-e_2)\widehat{\cL} (1-e_2)x e_1 e_2)\\
=& \lambda^{-1/2}\bE_{\cM}(e_2  \widehat{\cL} e_2) x + \lambda^{-3/2} x \bE_{\cM}(\widehat{\cL}_1 e_1 e_2)  + \lambda^{-3/2} \bE_{\cM}(e_2e_1\widehat{\cL}_1)^*x-x*\widehat{\cL}_0. 
\end{align*}
This completes the computation.
\end{proof}

\begin{remark}
Suppose that $e_2\widehat{\cL}=\widehat{\cL}e_2$.
We have that $\widehat{\cL}_1=0$ and 
\begin{align*}
  \cL(x)=\lambda^{-1/2}\bE_{\cM}(e_2 \widehat{\cL} e_2)x  -  x*\widehat{\cL}_0 = (1*\widehat{\cL}_0) x - x*\widehat{\cL}_0.
\end{align*}
If the inclusion $\cN\subset \cM$ is irreducible, we see that $e_2\widehat{\cL}=\widehat{\cL}e_2$.
\end{remark}

\begin{remark}
    This is a bimodule version of Proposition 5 in \cite{Lin76}.
\end{remark}

\begin{definition}[Laplacian]
Suppose $\{\Phi_t\}_{t\geq 0}$ is a bimodule quantum Markov semigroup.
The bimodule map $\cL_a$ defined by
\begin{align}
    \cL_a(x)= 
        \frac{1}{2}(1*\widehat{\cL}_0) x
        + \frac{1}{2} x  (1*\widehat{\cL}_0)
        -x* \widehat{\cL}_0, \quad x\in \cM,
\end{align}
is called the Laplacian of $\{\Phi_t\}_{t\geq 0}$.
\end{definition}
\begin{remark}
    Define 
    \begin{align*}
\cL_w (x)
=&  i[x, \Im\bE_{\cM}(\mathfrak{F}^{-1}(\widehat{\cL}_1))], \\
\Im\bE_{\cM}(\mathfrak{F}^{-1}(\widehat{\cL}_1))
= & \frac{i}{2}\left(\bE_{\cM}(\mathfrak{F}^{-1}(\widehat{\cL}_1))^*- \bE_{\cM}(\mathfrak{F}^{-1}(\widehat{\cL}_1))\right).
    \end{align*}
    Then $\mathcal{L}$ is decomposed as 
    \begin{align*}
        \cL=\cL_a+ \cL_w,
    \end{align*}  
The planar algebraic presentation of $\widehat{\mathcal{L}}_a$ is as follows: 
\begin{align*}
    \widehat{\cL}_a=\vcenter{\hbox{\begin{tikzpicture}[scale=0.65]
    \begin{scope}[shift={(0,1.5)}]
    \draw [blue] (-0.5, 0.8)--(-0.5, 0) .. controls +(0, -0.6) and +(0,-0.6).. (0.5, 0)--(0.5, 0.8);    
\begin{scope}[shift={(0.5, 0.3)}]
\end{scope}
    \end{scope}
\draw [blue] (-0.5, -0.8)--(-0.5, 0) .. controls +(0, 0.6) and +(0,0.6).. (0.5, 0)--(0.5, -0.8);
\begin{scope}[shift={(0.5, -0.3)}]
\draw [fill=white] (-0.3, -0.3) rectangle (0.3, 0.3);
\node at (0, 0) {\tiny $\mathbf{y}$};
\end{scope}
\end{tikzpicture}}}
+
\vcenter{\hbox{\begin{tikzpicture}[scale=0.65]
    \begin{scope}[shift={(0,1.5)}]
    \draw [blue] (-0.5, 0.8)--(-0.5, 0) .. controls +(0, -0.6) and +(0,-0.6).. (0.5, 0)--(0.5, 0.8);    
\begin{scope}[shift={(0.5, 0.3)}]
\draw [fill=white] (-0.3, -0.3) rectangle (0.3, 0.3);
\node at (0, 0) {\tiny $\mathbf{y}$};
\end{scope}
    \end{scope}
\draw [blue] (-0.5, -0.8)--(-0.5, 0) .. controls +(0, 0.6) and +(0,0.6).. (0.5, 0)--(0.5, -0.8);
\begin{scope}[shift={(0.5, -0.3)}]
\end{scope}
\end{tikzpicture}}}
-  \vcenter{\hbox{\begin{tikzpicture}[scale=0.6]
        \draw [blue] (0.2, -0.8) --(0.2, 0.8);
        \draw [blue] (-0.2, -0.8) --(-0.2, 0.8);
        \draw [fill=white] (-0.5, -0.4) rectangle (0.5, 0.4);
    \node at (0, 0) {\tiny $\widehat{\cL}_0$};
        \end{tikzpicture}}} ,
\end{align*}
where $\displaystyle \mathbf{y}=\frac{1}{2} (1*\widehat{\cL}_0)$. 
The graphical representation of $\widehat{\mathcal{L}}_w$ is 
\begin{align*}
    \widehat{\cL}_w=i\vcenter{\hbox{\begin{tikzpicture}[scale=0.65]
    \begin{scope}[shift={(0,1.5)}]
    \draw [blue] (-0.5, 0.8)--(-0.5, 0) .. controls +(0, -0.6) and +(0,-0.6).. (0.5, 0)--(0.5, 0.8);    
\begin{scope}[shift={(0.5, 0.3)}]
\end{scope}
    \end{scope}
\draw [blue] (-0.5, -0.8)--(-0.5, 0) .. controls +(0, 0.6) and +(0,0.6).. (0.5, 0)--(0.5, -0.8);
\begin{scope}[shift={(0.5, -0.3)}]
\draw [fill=white] (-0.3, -0.3) rectangle (0.3, 0.3);
\node at (0, 0) {\tiny $\mathbf{x}$};
\end{scope}
\end{tikzpicture}}}
-
i \vcenter{\hbox{\begin{tikzpicture}[scale=0.65]
    \begin{scope}[shift={(0,1.5)}]
    \draw [blue] (-0.5, 0.8)--(-0.5, 0) .. controls +(0, -0.6) and +(0,-0.6).. (0.5, 0)--(0.5, 0.8);    
\begin{scope}[shift={(0.5, 0.3)}]
\draw [fill=white] (-0.3, -0.3) rectangle (0.3, 0.3);
\node at (0, 0) {\tiny $\mathbf{x}$};
\end{scope}
    \end{scope}
\draw [blue] (-0.5, -0.8)--(-0.5, 0) .. controls +(0, 0.6) and +(0,0.6).. (0.5, 0)--(0.5, -0.8);
\begin{scope}[shift={(0.5, -0.3)}]
\end{scope}
\end{tikzpicture}}},    
\end{align*}
where $\mathbf{x}=\Im\bE_{\cM}(\mathfrak{F}^{-1}(\widehat{\cL}_1))$.
In particular, the Laplacian is related to $\widehat{\cL}_0$ by
    \begin{align*}
        (1-e_2)\widehat{\cL_a}(1-e_2) = (1-e_2)\widehat{\cL}(1-e_2) = -\widehat{\cL}_0. 
    \end{align*}
\end{remark}

\begin{proposition}
Suppose that $\bE_{\cM}(\mathfrak{F}^{-1}(\widehat{\cL}_1))^*= \bE_{\cM}(\mathfrak{F}^{-1}(\widehat{\cL}_1))$.
Then $\cL_w=0$.
If $\cM$ is a factor, then we have $\cL_w=0$ if and only if $\bE_{\cM}(\mathfrak{F}^{-1}(\widehat{\cL}_1))^*= \bE_{\cM}(\mathfrak{F}^{-1}(\widehat{\cL}_1))$.
\end{proposition}
\begin{proof}
If $\bE_{\cM}(\mathfrak{F}^{-1}(\widehat{\cL}_1))^*= \bE_{\cM}(\mathfrak{F}^{-1}(\widehat{\cL}_1))$, we have that $\Im\bE_{\cM}(\mathfrak{F}^{-1}(\widehat{\cL}_1))=0$.
Hence $\cL_w=0$.

For the converse, note $\cL_w=0$ if and only if $\Im\bE_{\cM}(\mathfrak{F}^{-1}(\widehat{\cL}_1))$ is in the center of $\cM$.   
If $\cM$ is a factor, then $\Im\bE_{\cM}(\mathfrak{F}^{-1}(\widehat{\cL}_1))$ is in the center of $\cM$ if and only if $\Im\bE_{\cM}(\mathfrak{F}^{-1}(\widehat{\cL}_1))$ is a multiple of scalars.
Since $\Im\bE_{\cM}(\mathfrak{F}^{-1}(\widehat{\cL}_1))\in \cN'\cap \cM$ and $\tau(\Im\bE_{\cM}(\mathfrak{F}^{-1}(\widehat{\cL}_1)))=0$, we have that $\cL_w=0$ if and only if $\bE_{\cM}(\mathfrak{F}^{-1}(\widehat{\cL}_1))^*= \bE_{\cM}(\mathfrak{F}^{-1}(\widehat{\cL}_1))$.
\end{proof}


\begin{corollary}
    Suppose $\cN=\bC$ and $\cM=M_n(\bC)$ and $\{\Phi_t\}_{t\geq 0}$ is a quantum Markov semigroup on $M_n(\bC)$.
Then there exist $v_j$, $w$ in $M_n(\bC)$ with hermitian $w$ such that 
    \begin{align*}
        \cL(x)=&\sum_{j\in \sI_0} \frac{1}{2} \omega_j\{v_j^*v_j, x\} - \omega_j v_j^* x v_j +i [w, x]\\
        =& \frac{1}{2}\sum_{j \in \sI_0} \omega_j v_j^*[x, v_j] +\frac{1}{2} \sum_{j \in \sI_0} \omega_j [v_j^*, x] v_j +i [w, x], 
    \end{align*}
subject to $\tau(v_jv_k^*)=\lambda^{1/2} \delta_{j,k}$, $\tau(v_j)=0$,  and $\omega_j \geq 0$.
\end{corollary}
\begin{proof}
Note that $\cL=\cL_a+\cL_w$. 
Let $\displaystyle w=\frac{1}{2i}\left(\lambda^{-3/2} \bE_{\cM}(e_2e_1\widehat{\cL}_1^*) - \lambda^{-3/2} \bE_{\cM}(\widehat{\cL}_1e_1e_2)\right)\in M_n(\bC)$.
Then $\cL_w(x)=i[x, w]$ for any $x\in \cM$.
Writing
$\vcenter{\hbox{\begin{tikzpicture}[scale=0.6]
        \draw [blue] (0.2, -0.8) --(0.2, 0.8);
        \draw [blue] (-0.2, -0.8) --(-0.2, 0.8);
        \draw [fill=white] (-0.5, -0.4) rectangle (0.5, 0.4);
    \node at (0, 0) {\tiny $p_j$};
        \end{tikzpicture}}}  
    =\vcenter{\hbox{\begin{tikzpicture}[scale=0.65]
    \begin{scope}[shift={(0,1.5)}]
    \draw [blue] (-0.5, 0.8)--(-0.5, 0) .. controls +(0, -0.6) and +(0,-0.6).. (0.5, 0)--(0.5, 0.8);    
\begin{scope}[shift={(0.5, 0.3)}]
\draw [fill=white] (-0.3, -0.3) rectangle (0.3, 0.3);
\node at (0, 0) {\tiny $v_j$};
\end{scope}
    \end{scope}
\draw [blue] (-0.5, -0.8)--(-0.5, 0) .. controls +(0, 0.6) and +(0,0.6).. (0.5, 0)--(0.5, -0.8);
\begin{scope}[shift={(0.5, -0.3)}]
\draw [fill=white] (-0.3, -0.3) rectangle (0.3, 0.3);
\node at (0, 0) {\tiny $v_j^*$};
\end{scope}
\end{tikzpicture}}},$
and $\displaystyle \widehat{\cL}_0=\sum_{j \in \sI_0}\omega_j p_j$ is the spectral decomposition with $p_j$ minimal projection, we have that 
\begin{align*}
    x*\widehat{\cL}_0= 
\sum_{j \in \sI_0} \omega_j\vcenter{\hbox{\begin{tikzpicture}[scale=0.65]
\begin{scope}[shift={(0,1.5)}]
\draw [blue] (-0.7, -0.5)--(-0.7, 0.4)..controls +(0,0.4) and +(0, 0.4)..(-0.3, 0.4)--(-0.3, 0) .. controls +(0, -0.4) and +(0,-0.4).. (0.3, 0)--(0.3, 0.8);    
\begin{scope}[shift={(0.3, 0.3)}]
\draw [fill=white] (-0.3, -0.3) rectangle (0.3, 0.3);
\node at (0, 0) {\tiny $v_j$};
\end{scope}
    \end{scope}
\draw [blue] (-0.7, 0.5)--(-0.7, -0.4)..controls +(0,-0.4) and +(0, -0.4)..(-0.3, -0.4)--(-0.3, 0) .. controls +(0, 0.4) and +(0,0.4).. (0.3, 0)--(0.3, -0.8);
\begin{scope}[shift={(0.3, -0.3)}]
\draw [fill=white] (-0.3, -0.3) rectangle (0.3, 0.3);
\node at (0, 0) {\tiny $v_j^*$};
\end{scope}
\begin{scope}[shift={(-0.7, 0.75)}]
\draw [fill=white] (-0.3, -0.3) rectangle (0.3, 0.3);
\node at (0, 0) {\tiny $x$};
\end{scope}
\end{tikzpicture}}}
=\sum_{j \in \sI_0}\omega_j v_j^*x v_j.
\end{align*}
Hence
\begin{align*}
\cL(x)
=\frac{1}{2} \sum_{j \in \sI_0} \omega_j v_j^* v_j x+\frac{1}{2}\sum_{j \in \sI_0} \omega_j x v_j^* v_j -\sum_{j \in \sI_0}  \omega_j v_j^* xv_j +i [w, x].    
\end{align*}
This completes the computation.
\end{proof}

\begin{remark}
For the inclusion $\bC\subset M_n(\bC)$, the  
\begin{align*}
    \mathfrak{F}^{-1}(\widehat{\cL}_0)=\sum_{j \in \sI_0}\omega_j \vcenter{\hbox{\begin{tikzpicture}[scale=1.2]
        \draw [blue] (0, -0.5)--(0, 0.5) (0.5, -0.5)--(0.5, 0.5);
        \draw [fill=white] (-0.2, -0.2) rectangle (0.2, 0.2);
        \node at (0, 0) {\tiny $v_j$};
        \begin{scope}[shift={(0.5, 0)}]
         \draw [fill=white] (-0.2, -0.2) rectangle (0.2, 0.2);
        \node at (0, 0) {\tiny $\overline{v_j^*}$};   
        \end{scope}
    \end{tikzpicture}}}.
\end{align*}
\end{remark}

In the following, we give a completely characterization of the Lindbladian for bimodule quantum Markov semigroups.
\begin{theorem}\label{thm:generatorform}
Suppose $0\leq L_0\in \cM'\cap \cM_2$ and $L_1^*=L_1\in \cN'\cap \cM$.
Define a bounded bimodule map $\cL:\cM\to \cM$ as follows: 
\begin{align}
\cL(x)=   
\frac{1}{2}(1* L_0)   x + \frac{1}{2} x (1* L_0)   
        + iL_1 x-ix L_1-x* L_0.
\end{align}
Then $\cL$ is the generator of a bimodule quantum Markov semigroup.
\end{theorem}
\begin{proof}
By a direct computation, we have that 
\begin{align*}
    \widehat{\cL}=   \vcenter{\hbox{\begin{tikzpicture}[scale=0.65]
    \begin{scope}[shift={(0,1.5)}]
    \draw [blue] (-0.5, 0.8)--(-0.5, 0) .. controls +(0, -0.6) and +(0,-0.6).. (0.5, 0)--(0.5, 0.8);    
    \end{scope}
\draw [blue] (-0.5, -0.8)--(-0.5, 0) .. controls +(0, 0.6) and +(0,0.6).. (0.5, 0)--(0.5, -0.8);
\begin{scope}[shift={(0.5, -0.3)}]
\draw [fill=white] (-0.3, -0.3) rectangle (0.3, 0.3);
\node at (0, 0) {\tiny $L_2^*$};
\end{scope}
        \end{tikzpicture}}}
+\vcenter{\hbox{\begin{tikzpicture}[scale=0.6]
    \begin{scope}[shift={(0,1.5)}]
    \draw [blue] (-0.5, 0.8)--(-0.5, 0) .. controls +(0, -0.6) and +(0,-0.6).. (0.5, 0)--(0.5, 0.8);  
    \begin{scope}[shift={(0.5, 0.3)}]
\draw [fill=white] (-0.3, -0.3) rectangle (0.3, 0.3);
\node at (0, 0) {\tiny $L_2$};
\end{scope}
    \end{scope}
\draw [blue] (-0.5, -0.8)--(-0.5, 0) .. controls +(0, 0.6) and +(0,0.6).. (0.5, 0)--(0.5, -0.8);
        \end{tikzpicture}}}
        -\vcenter{\hbox{\begin{tikzpicture}[scale=0.6]
        \draw [blue] (0.2, -0.8) --(0.2, 0.8);
        \draw [blue] (-0.2, -0.8) --(-0.2, 0.8);
        \draw [fill=white] (-0.5, -0.4) rectangle (0.5, 0.4);
    \node at (0, 0) {\tiny $L_0$};
        \end{tikzpicture}}},
\end{align*}
where $L_2=\frac{1}{2}\vcenter{\hbox{\begin{tikzpicture}[scale=0.6]
        \draw [blue] (0.2, -0.8) --(0.2, 0.8);
            \draw [fill=white] (-0.5, -0.4) rectangle (0.5, 0.4);
            \node at (0, 0) {\tiny $L_0$};
            \draw [blue] (-0.2, 0.4).. controls +(0, 0.35) and +(0, 0.35)..(-0.8, 0.4)--(-0.8, -0.4);
            \draw [blue] (-0.2, -0.4).. controls +(0, -0.35) and +(0, -0.35)..(-0.8, -0.4);
        \end{tikzpicture}}} +iL_1$.
Now it is clear that $-(1-e_2)\widehat{\cL}(1-e_2)=(1-e_2)L_0(1-e_2)\geq 0$.
By Theorem \ref{thm:generatorpos}, we see that $\cL$ is the generator of a bimodule quantum Markov semigroup.
\end{proof}

\subsection{Derivation}

In \cite{CiprianiSauvageot2003}, the Cipriani and Sauvageot considered the first-order differential structure of trace-symmetric quantum Markov semigroups. 
Here we define the derivation $\partial: \cM\to \cM_1$ for the bimodule quantum Markov semigroup $\{\Phi_t\}_{t \geq 0}$ as
\begin{align}
    \partial x=[x, \mathfrak{F}^{-1}(\widehat{\cL}_0^{1/2})], \quad x\in \cM.
\end{align}
The conjugated derivation $\overline{\partial}:\cM \to \cM_1$ is defined as 
\begin{align}
    \overline{\partial} x=[x, \mathfrak{F}^{-1}(\widehat{\cL}_0^{1/2})^*], \quad x\in \cM.
\end{align}
\begin{remark}
    The Fourier multiplier of $\partial$ is in the 3-box space of the corresponding planar algebra and depicted as follows:
    \begin{align*}
        \vcenter{\hbox{\scalebox{0.8}{
        \begin{tikzpicture}[scale=1.2]
           \draw [blue] (-0.2, 0.4)--(-0.2, 0.9) (0.2, 0.4)--(0.2, 0.9);
           \draw [blue] (-0.2, -0.4)--(0.7, -1.2);
           \draw [blue] (0.2, -0.4).. controls +(0, -0.3) and +(0, -0.3) .. (0.7, -0.4)--(0.7, 0.9);
           \draw [blue] (-0.2, -1.2).. controls +(0, 0.4) and +(0, 0.4).. (0.3, -1.2);
           \draw [fill=white] (-0.5, -0.4) rectangle (0.5, 0.4);
           \node at (0, 0) {\tiny $\mathfrak{F}^{-1}(\widehat{\cL}_0^{1/2})$};
        \end{tikzpicture}
        }}}
        -
              \vcenter{\hbox{\scalebox{0.8}{
        \begin{tikzpicture}[scale=1.2]
           \draw [blue] (-0.2, -0.4)--(-0.2, -0.9) (0.2, -0.4)--(0.2, -0.9);
           \draw [blue] (-0.2, 0.4)--(0.7, 1.2);
           \draw [blue] (0.2, 0.4).. controls +(0, 0.3) and +(0, 0.3) .. (0.7, 0.4)--(0.7, -0.9);
           \draw [blue] (-0.3, 1.2).. controls +(0, -0.4) and +(0, -0.4).. (0.3, 1.2);
           \draw [fill=white] (-0.5, -0.4) rectangle (0.5, 0.4);
           \node at (0, 0) {\tiny $\mathfrak{F}^{-1}(\overline{\widehat{\cL}_0^{1/2}})$};
        \end{tikzpicture}
        }}}
        =
        \vcenter{\hbox{\scalebox{0.8}{
        \begin{tikzpicture}[scale=1.2]
           \draw [blue] (-0.2, 0.4)--(-0.2, 0.9) (0.2, 0.9)--(0.2, -1.2);
           \draw [blue] (-0.2, -0.4).. controls +(0, -0.3) and +(0, -0.3) .. (-0.7, -0.4)--(-0.7, 0.9);
           \draw [blue] (-0.2, -1.2).. controls +(0, 0.4) and +(0, 0.4).. (-0.7, -1.2);
           \draw [fill=white] (-0.5, -0.4) rectangle (0.5, 0.4);
           \node at (0, 0) {\tiny $\overline{\widehat{\cL}_0^{1/2}}$};
        \end{tikzpicture}}}}
        -
           \vcenter{\hbox{\scalebox{0.8}{
        \begin{tikzpicture}[scale=1.2]
           \draw [blue] (-0.2, -0.4)--(-0.2, -0.9) (0.2, -0.9)--(0.2, 1.2);
           \draw [blue] (-0.2, 0.4).. controls +(0, 0.3) and +(0, 0.3) .. (-0.7, 0.4)--(-0.7, -0.9);
           \draw [blue] (-0.2, 1.2).. controls +(0, -0.4) and +(0, -0.4).. (-0.7, 1.2);
           \draw [fill=white] (-0.5, -0.4) rectangle (0.5, 0.4);
           \node at (0, 0) {\tiny $\widehat{\cL}_0^{1/2}$};
        \end{tikzpicture}}}}.
    \end{align*}
\end{remark}

We see that the adjoint $\partial^*:\cM_1\to \cM$ of $\partial$ is 
\begin{align*}
  \partial ^* x= \bE_{\cM}( [x, \mathfrak{F}^{-1}(\widehat{\cL}_0^{1/2})^*]), \quad x\in \cM_1.
\end{align*}
and the adjoint $\overline{\partial}^*:\cM_1\to \cM$ of $\overline{\partial}$ is 
\begin{align*}
  \overline{\partial} ^* x= \bE_{\cM}( [x, \mathfrak{F}^{-1}(\widehat{\cL}_0^{1/2})]), \quad x\in \cM_1.
\end{align*}
We have that for any $x\in \cM$ and $y\in \cM_1$,
\begin{align*}
    \tau_1(y^* \partial x)= & \tau((\partial^* y)^* x), \\
    \tau_1(y^* \overline{\partial} x)= & \tau((\overline{\partial}^* y)^* x).
\end{align*}
In particular, we have that $\tau(y^*\partial^*\partial x)=\tau_1((\partial y)^*(\partial x))$ for any $x, y\in \cM$. 
Now we define the directional derivation $\partial_j$ as follows
\begin{align*}
  \partial_j x =&  \omega_j^{1/2} [x, \mathfrak{F}^{-1}(p_j)],\quad j \in \sI.
\end{align*}
It then follows that
\begin{align*}
    \overline{\partial_j} x =&  \omega_j^{1/2} [x, \mathfrak{F}^{-1}(p_j)^*]
    = \omega_j^{1/2} [x, \mathfrak{F}^{-1}(\overline{p_j})],\quad j \in \sI.
\end{align*}
By considering the adjoint, we obtain
\begin{align*}
    \partial_j^* x =& \omega_j^{1/2} \bE_{\cM}([x, \mathfrak{F}^{-1}(p_j)^*])
    =\omega_j^{1/2} \bE_{\cM}([x, \mathfrak{F}^{-1}(\overline{p_j})]),\quad j \in \sI.
\end{align*}
and 
\begin{align*}
 \overline{\partial_j}^* x =&  \omega_j^{1/2}\bE_{\cM}( [x, \mathfrak{F}^{-1}(p_j)]),\quad j \in \sI.
\end{align*}
By the construction, we have $\displaystyle \partial =\sum_{j\in \sI_0\cup\sI_1} \partial_j$. 

\begin{remark}
Suppose that $\cN=\bC$ and $\cM=M_n(\bC)$.
We have that 
\begin{align*}
    \partial x =  \sum_{j\in \sI_0} \omega_j^{1/2} [x, \mathfrak{F}^{-1}(p_j)] 
    = & \sum_{j\in \sI_0} \omega_j^{1/2} [x, v_j]\otimes \overline{v_j^*}, \\
    \overline{\partial} x =  \sum_{j\in \sI_0} \omega_j^{1/2} [x, \mathfrak{F}^{-1}(p_j)^*] 
    = & \sum_{j\in \sI_0} \omega_j^{1/2} [x, v_j^*]\otimes \overline{v_j}.
\end{align*}
This indicates $\overline{\partial}$ can be characterized by $\partial_j^*$ completely.
\end{remark}

\begin{lemma}
    We have that 
    \begin{align}
        (\partial x )e_2=& \lambda^{-1/2}\widehat{\cL}_0^{1/2}[x, e_1]e_2, \nonumber\\
        e_2(\partial x )=& \lambda^{-1/2}e_2[x, e_1] \overline{\widehat{\cL}_0^{1/2}}.
    \end{align}
\end{lemma}
\begin{proof}
    By Lemma \ref{lem:fourier}, we have that 
    \begin{align*}
         (\partial x )e_2=\lambda^{-1/2}x\widehat{\cL}_0^{1/2}e_1e_2- \lambda^{-1/2}\widehat{\cL}_0^{1/2}e_1e_2x=\lambda^{-1/2}\widehat{\cL}_0^{1/2}[x, e_1]e_2.
    \end{align*}
    The rest follows from a similar computation.
\end{proof}

\begin{remark}
    We have that $\cM\subset \Ker\partial_j^*$.
    In fact, for any $x\in \cM$, we see that 
    \begin{align*}
        (\partial_j^* x )e_2=e_2[x, \mathfrak{F}^{-1}(p_j)^*] e_2=\lambda^{-1/2}e_2\overline{p_j}[x, e_1] e_2=0,
    \end{align*}
    by the fact that $p_je_2=0$.
\end{remark}

\begin{lemma}
    For $j\neq k$, 
\begin{align}
\partial_j^* \partial_k =0, \quad \overline{\partial_j}^* \overline{\partial_k}=0. 
\end{align}
and
\begin{align*}
    \partial_j^* \partial_j x=
    \lambda^{1/2}\omega_j (1* p_j)x
        +  \lambda^{1/2} \omega_j x (1* \overline{p_j})  
        -\lambda^{1/2}\omega_j x* (\overline{p_j}+p_j).
\end{align*}
\end{lemma}
\begin{proof}
For any $x\in \cM$, we have that 
\begin{align*}
   e_2\partial_k^* \partial_j(x)
   =&e_2\omega_k^{1/2} \omega_{j}^{1/2}\bE_{\cM}([[x, \mathfrak{F}^{-1}(p_j)], \mathfrak{F}^{-1}(\overline{p_k}) ])\\
   =&e_2\omega_k^{1/2} \omega_{j}^{1/2}[[x, \mathfrak{F}^{-1}(p_j)], \mathfrak{F}^{-1}(\overline{p_k}) ] e_2\\
   =& \lambda^{-1}  \omega_k^{1/2} \omega_{j}^{1/2} (xe_2e_1\overline{p_j}\overline{p_k}e_1e_2 - e_2e_1\overline{p_j}x\overline{p_k}e_1e_2 \\
   & - e_2e_1 p_k x p_j e_1e_2+ e_2e_1 p_k  p_je_1e_2x)\\
   =& 0,
\end{align*}
as $p_kp_j = 0$ for $j\neq k$.
Similarly, we have that $\overline{\partial_j}^* \overline{\partial_k}=0$.
Moreover, we have that 
\begin{align*}
   e_2\partial_j^* \partial_j(x)  
 = & \lambda^{-1}   \omega_{j}(xe_2e_1\overline{p_j}e_1e_2 - e_2e_1\overline{p_j}x e_1e_2 \\
   & - e_2e_1  x p_j e_1e_2+ e_2e_1 p_j e_1e_2x).
\end{align*}
Applying the conditional expectaion $\bE_{\cM}$, we obtain that 
\begin{align*}
    \partial_j^* \partial_j(x)
    =& \lambda^{-2}\omega_{j}(\bE_{\cM}(xe_2e_1\overline{p_j}e_1e_2) - \bE_{\cM}(e_2e_1  x\overline{p_j}e_1e_2)\\
    &-\bE_{\cM}(e_2e_1  x p_j e_1e_2) +\bE_{\cM}( e_2e_1 p_je_1e_2x))\\
    =& \lambda^{1/2}\omega_j (x(1*\overline{p_j})-x*\overline{p_j}-x* p_j +(1*p_j)x). 
\end{align*}
This completes the proof of the Lemma.
\end{proof}
\begin{corollary}
    We have 
    \begin{align*}
        \partial^* \partial x=\sum_{j\in \sI_0\cup \sI_1} \partial_j^* \partial_j x=  \lambda^{1/2}((1* \widehat{\cL}_0)x+x(1* \overline{\widehat{\cL}_0})- x* (\widehat{\cL}_0+\overline{\widehat{\cL}_0})).
    \end{align*}
\end{corollary}
\begin{proof}
It follows from previous lemma.
\end{proof}
\begin{remark}
    
Suppose that $\mathcal{A}$ is a *-subalgebra of $\mathcal{B}(\cH)$ and $D$ is an (unbounded) self-adjoint operator on $\cH$.
The triple $(\mathcal{A}, \cH, D)$ is a spectral triple \cite{Con94} if the commutator $[D, x]$ is bounded for all $x\in \mathcal{A}$.
Note that $\mathfrak{F}^{-1}(\widehat{\cL}_0^{1/2})$ is not self-adjoint in general.
Hence the triple $(\cM, L^2(\cM, \tau), \partial)$ is not a spectral triple in general.

In \cite{LiJungeLaRacuente2024}, the derivation triple was introduced as a modification of the spectral triple, which has played a crucial role in deriving abstract Ricci curvature lower bound and functional inequalities\cite{BrannanGaoJunge2022,BrannanGaoJunge2022riccicurvatureII}.

Suppose that $\cM\subset \widetilde{\cM}$ is an inclusion of finite von Neumann algebra with a normal faithful tracial state $\tau$ and $\delta$ is a closable *-preserving derivation on $\cM$ valued in $L^2(\widetilde{\cM}, \tau)$.
The triple $(\cM\subset \widetilde{\cM}, \tau, \delta)$ is a derivation triple if there exists a weakly dense *-subalgebra $\mathcal{A}\subset \cM$ such that $\mathcal{A}$ is in the domain of $\delta$.

For a finite inclusion $\cN\subset\cM$ and a  $\cN$-bimodule Markov semigroup $e^{-t\mathcal{L}}$ on $\cM$, we see that $(\cM\subset\cM_1,\tau_1,\partial)$ is a derivation triple. 
In our setting, the derivation triple captures the dual structure via the Fourier transform of $\mathcal{L}$. 
\end{remark}

\begin{lemma}\label{lem:laplacian1}
 For any $x\in \cM$,
 \begin{align}
 \cL_a(x)= -\frac{\lambda^{-1/2}}{2} \bE_{\cM}(\mathfrak{F}^{-1}(\overline{\widehat{\cL}_0^{1/2}}) (\partial x )) + \frac{\lambda^{-1/2}}{2} \bE_{\cM}(  (\overline{\partial} x) \mathfrak{F}^{-1}(\widehat{\cL}_0^{1/2} ) ).   
 \end{align}
\end{lemma}
\begin{proof}
We have that 
    \begin{align*}
    \cL_a(x) 
    =&  \frac{1}{2}(1*\widehat{\cL}_0) x
        + \frac{1}{2} x  (1*\widehat{\cL}_0)
        -x* \widehat{\cL}_0\\
    =& \frac{\lambda^{-5/2}}{2} \left(\bE_{\cM}(e_2e_1 \widehat{\cL}_0 e_1e_2 x) - \bE_{\cM}(e_2e_1 \widehat{\cL}_0 xe_1e_2 ) \right) \\
    &+ \frac{\lambda^{-5/2}}{2} \left(\bE_{\cM}( x e_2e_1 \widehat{\cL}_0 e_1e_2 ) - \bE_{\cM}(e_2e_1 \widehat{\cL}_0 xe_1e_2 ) \right)\\
    =& -\frac{\lambda^{-5/2}}{2} \bE_{\cM}(e_2e_1 \widehat{\cL}_0 [x, e_1]e_2) + \frac{\lambda^{-5/2}}{2} \bE_{\cM}( e_2[x, e_1]   \widehat{\cL}_0 e_1e_2 ) \\
    =& -\frac{\lambda^{-2}}{2} \bE_{\cM}(e_2e_1 \widehat{\cL}_0^{1/2} (\partial x )e_2) + \frac{\lambda^{-2}}{2} \bE_{\cM}( e_2 (\overline{\partial} x) \widehat{\cL}_0^{1/2} e_1e_2 )\\
    =& -\frac{\lambda^{-3/2}}{2} \bE_{\cM}(e_2\mathfrak{F}^{-1}(\overline{\widehat{\cL}_0^{1/2}}) (\partial x )e_2) + \frac{\lambda^{-3/2}}{2} \bE_{\cM}( e_2 (\overline{\partial} x) \mathfrak{F}^{-1}(\widehat{\cL}_0^{1/2} )e_2 )\\
    =& -\frac{\lambda^{-1/2}}{2} \bE_{\cM}(\mathfrak{F}^{-1}(\overline{\widehat{\cL}_0^{1/2}}) (\partial x )) + \frac{\lambda^{-1/2}}{2} \bE_{\cM}(  (\overline{\partial} x) \mathfrak{F}^{-1}(\widehat{\cL}_0^{1/2} ) ).
\end{align*}
This completes the computation.
\end{proof}

\begin{remark}
Suppose that $\cN=\bC$, $\cM=M_n(\bC)$.
We have that 
\begin{align*}
    \cL_a(x)=& \frac{1 }{2}\sum_{j\in \sI_0} \omega_j [x, v_j^*] v_j -\omega_j v_j^* [x, v_j],
\end{align*}
subject to $\tau(v_jv_k^*)=\lambda^{1/2} \delta_{j,k}$, $\tau(v_j)=0$,  and $\omega_j \geq 0$.
\end{remark}

\begin{lemma}\label{lem:laplacian20}
  For any $x\in \cM$,
 \begin{align}
 \cL_a^*(x)= -\frac{\lambda^{-1/2}}{2} \bE_{\cM}\left( [ \mathfrak{F}^{-1}( \widehat{\cL}_0^{1/2})y, \mathfrak{F}^{-1}(\widehat{\cL}_0^{1/2})^*] \right)+ \frac{\lambda^{-1/2}}{2} \bE_{\cM}\left( [y\mathfrak{F}^{-1}(\overline{\widehat{\cL}_0}^{1/2} ), \mathfrak{F}^{-1}(\widehat{\cL}_0^{1/2})] \right).  
 \end{align}   
\end{lemma}
\begin{proof}
Suppose $x, y\in \cM$.
We have that 
\begin{align*}
    \tau(\cL_a^*(y)^*x)=& \tau(y^*\cL_a(x))\\
    =&  -\frac{\lambda^{-1/2}}{2} \tau(y^* \bE_{\cM}(\mathfrak{F}^{-1}(\overline{\widehat{\cL}_0^{1/2}}) (\partial x )) + \frac{\lambda^{-1/2}}{2}\tau(y^* \bE_{\cM}(  (\overline{\partial} x) \mathfrak{F}^{-1}(\widehat{\cL}_0^{1/2} ) ))\\
    =& -\frac{\lambda^{-1/2}}{2} \tau_1(y^* \mathfrak{F}^{-1}(\overline{\widehat{\cL}_0^{1/2}}) (\partial x )) + \frac{\lambda^{-1/2}}{2}\tau_1(y^* (\overline{\partial} x) \mathfrak{F}^{-1}(\widehat{\cL}_0^{1/2} ) )\\
=& -\frac{\lambda^{-1/2}}{2} \tau_1((\partial^* \mathfrak{F}^{-1}( \widehat{\cL}_0^{1/2})y)^*  x) + \frac{\lambda^{-1/2}}{2}\tau_1(( \overline{\partial}^* y\mathfrak{F}^{-1}(\overline{\widehat{\cL}_0}^{1/2} ))^* x ).
\end{align*}
This implies that 
\begin{align*}
    \cL_a^*(y)=&  -\frac{\lambda^{-1/2}}{2} \partial^* \mathfrak{F}^{-1}( \widehat{\cL}_0^{1/2})y + \frac{\lambda^{-1/2}}{2} \overline{\partial}^* y\mathfrak{F}^{-1}(\overline{\widehat{\cL}_0}^{1/2} ) \\
    =& -\frac{\lambda^{-1/2}}{2} \bE_{\cM}\left( [ \mathfrak{F}^{-1}( \widehat{\cL}_0^{1/2})y, \mathfrak{F}^{-1}(\widehat{\cL}_0^{1/2})^*] \right)+ \frac{\lambda^{-1/2}}{2} \bE_{\cM}\left( [y\mathfrak{F}^{-1}(\overline{\widehat{\cL}_0}^{1/2} ), \mathfrak{F}^{-1}(\widehat{\cL}_0^{1/2})] \right).
\end{align*}
We see that the Lemma is true.
\end{proof}

\begin{remark}
Suppose that $\cN=\bC$ and $\cM=M_n(\bC)$.
For any $x\in \cM$, we have that 
\begin{align*}
    \cL_{a}^*(x)=-\frac{1}{2}\sum_{j \in \sI_0} \omega_j [v_j x, v_j^*]+  \omega_j [x v_j^*, v_j],
\end{align*}
subject to $\tau(v_jv_k^*)=\lambda^{1/2} \delta_{j,k}$, $\tau(v_j)=0$,  and $\omega_j \geq 0$.
\end{remark}

We define the completely bounded map $\cL_{\overline{a}}$ as follows:
\begin{align*}
    \cL_{\overline{a}}(x)=\frac{1}{2}(1* \overline{\widehat{\cL}_0})x + \frac{1}{2} x(1* \overline{\widehat{\cL}_0})   -x* \overline{\widehat{\cL}_0}.
\end{align*}
By Lemma \ref{lem:laplacian1}, we have that
\begin{align*}
 \cL_{\overline{a}}(x)= -\frac{\lambda^{-1/2}}{2} \bE_{\cM}(\mathfrak{F}^{-1}( \widehat{\cL}_0^{1/2}) (\overline{\partial} x )) + \frac{\lambda^{-1/2}}{2} \bE_{\cM}(  (\partial x) \mathfrak{F}^{-1}(\overline{\widehat{\cL}_0 }^{1/2}) ).   
 \end{align*}

 If $\tau(y^*\cL(x))=\tau(\cL(y)^*x)$ for any $x, y\in \cM$, i.e. $\cL$ is symmetric with respect to $\tau$, then $\cL_{\overline{a}}=\cL_a^*$.
 
\begin{proposition}\label{prop:laplacian2}
    We have 
    \begin{align}
        \frac{\lambda^{-1/2}}{2} \left(\partial^* \partial + \overline{\partial}^* \overline{\partial}\right)=\cL_a+\cL_{\overline{a}}. 
    \end{align}
\end{proposition}
\begin{proof}
For any $x\in \cM$, we have that 
\begin{align*}
   \cL_{\overline{a}}(x)+ \cL_a(x)=&  \frac{1}{2}(1* \widehat{\cL}_0)x + \frac{1}{2} x(1* \widehat{\cL}_0)   -x* \widehat{\cL}_0\\
  &  + \frac{1}{2}(1* \overline{\widehat{\cL}_0})x + \frac{1}{2} x(1* \overline{\widehat{\cL}_0})   -x* \overline{\widehat{\cL}_0}\\
        =&\frac{\lambda^{-1/2}}{2}\sum_{j \in \sI_0\cup \sI_1} \overline{\partial_j}^*\overline{\partial_j}x +  \partial_j^* \partial_j x\\
        =& \frac{\lambda^{-1/2}}{2}\overline{\partial}^* \overline{\partial}x + \frac{\lambda^{-1/2}}{2} \partial^*  \partial x.
\end{align*}
This completes the computation.
\end{proof}

\begin{remark}
    We have that $\Ker(\cL_a+\cL_{\overline{a}})$ is a von Neumann subalgebra.
\end{remark}

We recall the divergences and gradients for a bimodule quantum Markov semigroup $\{\Phi_t\}_{t\geq 0}$.
For any $x\in \cM$, we define the noncommutative gradients $\nabla:\cM \to \bigoplus_j \cM_1$ and $\overline{\nabla}:\cM \to \bigoplus_j \cM_1$ as follows
\begin{align*}
    \nabla x=(\partial_j x)_{j \in \sI_0\cup \sI_1}, \quad \overline{\nabla} x=(\overline{\partial_j} x)_{j \in \sI_0\cup \sI_1}.
\end{align*}
For any $(x_j)_{j\in \sI_0\cup\sI_1}\in \bigoplus_{j}\cM_1$, we define the divergences $\displaystyle \Div:\bigoplus_j \cM_1 \to \cM$ and $\displaystyle \overline{\Div}:\bigoplus_j \cM_1 \to \cM$ as follows
\begin{align*}
    \Div(x_j)_{j\in \sI_0\cup\sI_1}=\sum_{j\in \sI_0\cup \sI_1} \partial_j^* x_j, \quad \overline{\Div}(x_j)_{j\in \sI_0\cup\sI_1}=\sum_{j\in \sI_0\cup \sI_1} \overline{\partial_j}^* x_j.
\end{align*}
By construction, we have $\Div\nabla = \partial^*\partial$ and $\overline{\Div\nabla} = \overline{\partial}^*\overline{\partial}$. 
Thus by Proposition \ref{prop:laplacian2}, 
\begin{align*}
    \cL_a+\cL_{\overline{a}}=\frac{\lambda^{-1/2}}{2}(\overline{\Div}\ \overline{\nabla} + \Div  \nabla). 
\end{align*}

\begin{remark}
    Suppose $\cN=\bC$, $\cM=M_n(\bC)$.
    We have that 
    \begin{align*}
        \cL_a(x)+\cL_{\overline{a}}(x) =\frac{1 }{2}\sum_{j \in \sI_0} \omega_j[ [x, v_j^*], v_j]+\omega_j[[x, v_j], v_j^*]
    \end{align*}
subject to $\tau(v_jv_k^*)=\lambda^{1/2} \delta_{j,k}$, $\tau(v_j)=0$,  and $\omega_j \geq 0$.
\end{remark}
\begin{proposition}\label{prop:gradientform}
For any $x, y\in \cM$, we have that 
    \begin{align}
        2\Gamma(x, y)
        &= y^*(1*\widehat{\cL}_0) x- y^*(x*\widehat{\cL}_0)-(y^**\widehat{\cL}_0)x+(y^*x)*\widehat{\cL}_0\\
        &= \lambda^{-1/2} \sum_{j\in \sI_0\cup \sI_1}\bE_{\cM}((\partial_j y)^*(\partial_j x)).
    \end{align}
\end{proposition}
\begin{proof}
By Proposition \ref{prop:generatorformula}, we have that 
\begin{align*}
   2 \Gamma(x, y)= & y^*\cL(x)+\cL(y)^* x -\cL(y^*x) \\
  =&  \lambda^{-1/2}\bE_{\cM}(e_2 \widehat{\cL} e_2)y^*x + \lambda^{-3/2} y^*\bE_{\cM}(e_2e_1\widehat{\cL}_1^*) x+ \lambda^{-3/2} y^*\bE_{\cM}(\widehat{\cL}_1e_1e_2)x\\
  & -y^*(x*\widehat{\cL}_0)-(y^**\widehat{\cL}_0)x+(y^*x)*\widehat{\cL}_0\\
  =& y^*(1*\widehat{\cL}_0) x- y^*(x*\widehat{\cL}_0)-(y^**\widehat{\cL}_0)x+(y^*x)*\widehat{\cL}_0.
\end{align*}
On the other hand, we have that 
\begin{align*}
& \bE_{\cM}([y, \widehat{\cL}_0^{1/2}e_1e_2]^*[x, \widehat{\cL}_0^{1/2}e_1e_2]) \\
=& \bE_{\cM}(e_2e_1 y^*\widehat{\cL}_0xe_1e_2)-\bE_{\cM}(y^*e_2e_1\widehat{\cL}_0xe_1e_2)\\
& -\bE_{\cM}(e_2e_1 y^*\widehat{\cL}_0e_1e_2x)+\bE_{\cM}(y^*e_2e_1 \widehat{\cL}_0e_1e_2x) \\
=& \lambda^{5/2} (y^*(1*\widehat{\cL}_0) x- y^*(x*\widehat{\cL}_0)-(y^**\widehat{\cL}_0)x+(y^*x)*\widehat{\cL}_0).
\end{align*}
Note that $\widehat{\cL}_0^{1/2}e_1e_2=\lambda^{1/2}\mathfrak{F}^{-1}(\widehat{\cL}_0^{1/2})e_2$.
We obtain that 
\begin{align*}
 \bE_{\cM}([y, \widehat{\cL}_0^{1/2}e_1e_2]^*[x, \widehat{\cL}_0^{1/2}e_1e_2])
 =& \lambda \bE_{\cM}(e_2[y, \mathfrak{F}^{-1}(\widehat{\cL}_0^{1/2})]^*[x, \mathfrak{F}^{-1}(\widehat{\cL}_0^{1/2})]e_2)\\
 =& \lambda^2 \bE_{\cM}([y, \mathfrak{F}^{-1}(\widehat{\cL}_0^{1/2})]^*[x, \mathfrak{F}^{-1}(\widehat{\cL}_0^{1/2})])\\
 =& \lambda^2\bE_{\cM}((\partial y)^*(\partial x)).
\end{align*}
This completes the computation.
\end{proof}

\begin{corollary}\label{cor:convineq1}
Suppose the inclusion $\cN\subset \cM$ is irreducible. 
Then for any positive operator $x\in \cM$ and $y\in \cM'\cap \cM_2$, we have
\begin{align}
    (x^*x)*y+x^*(1*y)x-x^*(x*y)-(x^**y)x\geq \lambda^{1/2} |[x, \mathfrak{F}^{-1}(y^{1/2})]|^2.
\end{align}
In particular, if $x*y=0$, then 
\begin{align*}
    (x^*x)*y+x^*(1*y)x\geq \lambda^{1/2} |[x, \mathfrak{F}^{-1}(y^{1/2})]|^2.
\end{align*}
and 
\begin{align*}
   \|[x, \mathfrak{F}^{-1}(y^{1/2})]\|_2\leq  \sqrt{2}\|y\|^{1/2}\|x\|_2.
\end{align*}
\end{corollary}
\begin{proof}
  Note that by Pimsner-Popa inequality: 
  \begin{align*}
      &  (x^*x)*y+x^*(1*y)x-x^*(x*y)-(x^**y)x \\
       =&  \lambda^{-5/2}\bE_{\cM} ([x, y^{1/2}e_1 e_2]^*[x, y^{1/2}e_1 e_2]) \\
       =&  \lambda^{-1/2}\bE_{\cM} ([x, \mathfrak{F}^{-1}(y^{1/2})]^*[x, \mathfrak{F}^{-1}(y^{1/2})])\\
       \geq  &  \lambda^{1/2}[x, \mathfrak{F}^{-1}(y^{1/2})]^*[x, \mathfrak{F}^{-1}(y^{1/2})]= \lambda^{1/2}|[x, \mathfrak{F}^{-1}(y^{1/2})]|^2.
  \end{align*}

  Suppose that $x*y=0$.
  We see that $(x^*x)*y+x^*(1*y)x\geq  \lambda^{1/2} |[x, \mathfrak{F}^{-1}(y^{1/2})]|^2.$
  By taking trace, we obtain that 
  \begin{align*}
  \tau((x^*x)*y)+\tau(x^*(1*y)x)\geq  \lambda^{1/2} \|[x, \mathfrak{F}^{-1}(y^{1/2})]\|_2^2.
  \end{align*}
  Note that 
  \begin{align*}
  \tau((x^*x)*y)\leq \lambda^{1/2}\|x\|_2^2\|y\|, \quad 1*y\leq \lambda^{1/2}\|y\|.
  \end{align*}
  We see that $\|[x, \mathfrak{F}^{-1}(y^{1/2})]\|_2\leq  \sqrt{2}\|y\|^{1/2}\|x\|_2 $.
  If $\cN\subset \cM$ is irreducible, then 
    \begin{align*}
  \tau((x^*x)*y)=\lambda^{1/2}\|x\|_2^2\|y\|_1, \quad 1*y= \lambda^{1/2}\|y\|_1.
  \end{align*}
  This completes the computation.
\end{proof}

\begin{corollary}\label{cor:convineq2}
Suppose that $\cN\subset \cM$ is irreducible and $\Phi$ is bimodule completely positive map.
Then for any $x\in \cM$, we have that 
\begin{align}
    \Phi(x^*x)+ x^*\Phi(1)x-x^*\Phi(x)-\Phi(x^*)x \geq \lambda^{1/2} |[x, \mathfrak{F}^{-1}(\widehat{\Phi}^{1/2})]|^2
\end{align}
In particular, if $\Phi(x)=0$, we have that 
\begin{align}
    \Phi(x^*x)+ x^*\Phi(1)x \geq \lambda^{1/2} |[x, \mathfrak{F}^{-1}(\widehat{\Phi}^{1/2})]|^2
\end{align}
\end{corollary}
\begin{proof}
It follows from the same argument as Corollary \ref{cor:convineq1} by taking $y = \widehat{\Phi}$. 
\end{proof}

\subsection{Invariant Subalgebras} 

In this section, we investigate the invariant subspace of a bimodule quantum Markov semigroup $\{\Phi_t\}_{t\geq 0}$.

\begin{proposition}\label{prop:gradientform2}
Suppose that $\{\Phi_t\}_{t\geq 0}$ is a bimodule quantum Markov semigroup.
Then 
\begin{align*}
   \Ker\Gamma
   :=& \{x\in \cM: \Gamma(x, x)=0\}, \\
    =& \{x\in \cM: [x, \mathfrak{F}^{-1}(\widehat{\cL}_0^{1/2})]=0\},\\
     =& \{x\in \cM:  \widehat{\cL}_0 xe_1e_2= \widehat{\cL}_0e_1 x e_2\},\\
    =& \{x\in \cM:  \CS_0(\widehat{\cL}_0) xe_1e_2= \CS_0(\widehat{\cL}_0)e_1 x e_2\},\\
    =& \{x\in \cM: \partial_j x=0, j\in \sI_0\cup \sI_1\}.
\end{align*}
\end{proposition}
\begin{proof}
It follows from Proposition \ref{prop:gradientform}.    
\end{proof}

\begin{definition}
Suppose that $\{\Phi_t\}_{t\geq 0}$ is a bimodule quantum Markov semigroup.
The fixed point subspace $\fix(\{\Phi_t\}_{t\geq 0})$ is defined to be
\begin{align*}
    \fix(\{\Phi_t\}_{t\geq 0})=\{x\in \cM: \Phi_t(x)=x, t\geq 0\}
\end{align*}
The multiplicative domain $\nfix(\{\Phi_t\}_{t\geq 0})$ of $\{\Phi_t\}_{t\geq 0}$ is defined to be
\begin{align*}
    \nfix(\{\Phi_t\}_{t\geq 0})=\{x\in \cM: \Phi_t(x^*)\Phi_t(x)=\Phi_t(x^*x) , \Phi_t(x)\Phi_t(x^*)=\Phi_t(xx^*) , t\geq 0\}.
\end{align*}
\end{definition}

\begin{remark}
By the definition of Lindbladian, we see that $\fix(\{\Phi_t\}_{t\geq 0})=\Ker(\cL)$.   
\end{remark}

\begin{remark}
    By differentiating $\Phi_t(x^*)\Phi_t(x)=\Phi_t(x^*x)$ with respect to $t$ at $t=0$, we have that $\cL(x^*x)=\cL(x^*)x+x^*\cL(x)$.
    This implies that $\nfix(\{\Phi_t\}_{t\geq 0})\subset \Ker\Gamma$.  
\end{remark}

\begin{proposition}
Suppose that $\{\Phi_t\}_{t\geq 0}$ is a bimodule quantum Markov semigroup such that for each $t\geq0$, $\Phi_t$ admits equilibrium state $\rho_t$. 
Then $\fix(\{\Phi_t\}_{t\geq 0})$ is a von Neumann algebra.
\end{proposition}
\begin{proof}
    It follows from Lemma \ref{lem:subalgebra}.
\end{proof}

\begin{remark}
Suppose that $\{\Phi_t\}_{t\geq 0}$ is a bimodule quantum Markov semigroup such that $\Phi_t$ is equilibrium with respect to $\rho_t$.
We have that $\fix(\{\Phi_t\}_{t\geq 0})\subset \nfix(\{\Phi_t\}_{t\geq 0})$.
\end{remark}


\begin{theorem}\label{thm:bilimit}
 Suppose that $\{\Phi_t\}_{t\geq 0}$ is a bimodule quantum Markov semigroup and $\Phi_t$ is equilibrium with respect to a normal faithful state $\rho_t$.
 Then $\fix(\{\Phi_t\}_{t \geq 0})= \nfix(\{\Phi_t\}_{t \geq 0})$ if and only if 
 \begin{align}
     \lim_{t\to\infty} \widehat{\Phi}_t=\lim_{t\to\infty} \frac{1}{t}\int_0^t \widehat{\Phi}_s ds.
 \end{align}
\end{theorem}
\begin{proof}
Suppose that $\displaystyle \lim_{t\to\infty} \Phi_t(x)=\lim_{t\to\infty} \frac{1}{t}\int_0^t \Phi_s(x)ds$ for any $x\in \cM$.
Note that the limit $\displaystyle \lim_{t\to\infty} \frac{1}{t}\int_0^t \Phi_s(x)ds$ exists. 
We define the map $\bE_{\Phi}: \cM \to \cM$ to be
 \begin{align*}
     \bE_{\Phi}(x)= \lim_{t\to\infty} \frac{1}{t}\int_0^t \Phi_s(x)ds, \quad x\in \cM.
 \end{align*}
 Then $\bE_{\Phi}^2=\bE_{\Phi}$ is a conditional expectation onto the fixed point space of $\{\Phi_t\}_{t \geq 0}$, which is a von Neumann subalgebra.
By the assumption, we see that $\displaystyle \lim_{t\to\infty} \Phi_t(x)=\bE_{\Phi}(x)$ for any $x\in \cM$.
 
 Suppose $x\in \nfix(\Phi)$.
 Then $\Phi_t(x^*x)=\Phi_t(x^*)\Phi_t(x)$.
 By taking limit as $t\to \infty$, we obtain that $\bE_{\Phi}(x^*x)=\bE_{\Phi}(x^*)\bE_{\Phi}(x)$.
 Hence
 \begin{align*}
 \bE_{\Phi}((x-\bE_\Phi(x))^* (x-\bE_\Phi(x)))=\bE_{\Phi}(x^*x)-\bE_{\Phi}(x^*)\bE_{\Phi}(x)=0.
 \end{align*}
 This implies that $x=\bE_\Phi(x)$, i.e. $x\in \fix(\{\Phi_t\}_{t \geq 0})$.

Suppose that $\fix(\{\Phi_t\}_{t \geq 0})= \nfix(\{\Phi_t\}_{t \geq 0})$.
Let $x\in \cM$ with $\|x\|=1$.
Then $\{\Phi_t(x)\}_{t\geq 0}$ is bounded. 
By the Banach-Alaoglu theorem, there exists a subnet $\mathcal{I}$ and $x_0\in\cM$ such that the $\displaystyle \text{w}^*\text{-} \lim_{s\in \mathcal{I}} \Phi_s(x) = x_0$.
Suppose that $\rho_t(\cdot)=\langle \cdot\xi_t, \xi_t\rangle$.
We have that 
\begin{align*}
    \langle (\Phi_t(y_1^*y_2)-\Phi_t(y_1)^*\Phi_t(y_2))\xi_t,  \xi_t\rangle 
    =\langle y_2\xi_t, y_1\xi_t\rangle- \langle \Phi_t(y_2)\xi_t, \Phi_t(y_1)\xi_t\rangle,
\end{align*}
and
\begin{align*}
 &\langle (\Phi_t(\Phi_s(x)^*\Phi_s(x))-\Phi_t(\Phi_s(x))^*\Phi_t(\Phi_s(x)))\xi_t,  \xi_t\rangle   \\
 =& \| \Phi_s(x)\xi_t \|^2- \|\Phi_{t+s}(x)\xi_t\|^2\geq 0.  
\end{align*}
By taking limit as $s\in \mathcal{I}$, we have that 
\begin{align*}
\rho_t( (\Phi_t(x_0^*x_0)-\Phi_t(x_0)^*\Phi_t(x_0))) =0.
\end{align*}
Hence $x_0\in \nfix(\{\Phi_t\}_{t \geq 0})$.
Note that $\fix(\{\Phi_t\}_{t \geq 0})= \nfix(\{\Phi_t\}_{t \geq 0})$.
We see that $x_0\in \fix(\{\Phi_t\}_{t \geq 0})$.
Taking the conditional expectation $\bE_{\Phi}$, we obtain that $x_0=\bE_{\Phi}(x)$.
Therefore $\displaystyle \text{w}^*\text{-}\lim_{s\in \mathscr{I}}\Phi_s(x)=\bE_{\Phi}(x)$, which is independent of the choice of the subnet, i.e. 
\begin{align*}
\text{w}^*\text{-}\lim_{t\to\infty }\Phi_t(x)=\bE_{\Phi}(x).
\end{align*}
By the fact that the inclusion is finite, we see that $\displaystyle \lim_{t\to\infty }\Phi_t(x)=\bE_{\Phi}(x)$ for any $x\in \cM$.
\end{proof}

\begin{remark}
Theorem \ref{thm:bilimit} generalizes Theorem 3.1 in \cite{Fri78} and is a bimodule version of Theorem 3.4 in \cite{FriVer82}.
\end{remark}

\begin{definition}[Relative Irreducibility]
Suppose that $\{\Phi_t\}_{t\geq 0}$ is a bimodule quantum Markov semigroup.
We say $\{\Phi_t\}_{t\geq 0}$ is relatively irreducible if for any projection $p\in \cM$ satisfying $\Phi_t(p)\leq c_t p$ for some $c_t>0$ and for all $t\geq 0$, we have that $p\in \cN$.
\end{definition}


\begin{remark}
Suppose that $\{\Phi_t\}_{t\geq 0}$ is relative irreducibilty and $\cN$ is a factor.
Then $\Phi_t$ is equilibrium for all $t\geq 0$ follows from Lemma 5.9 in \cite{HJLW24}.
\end{remark}


\begin{definition}[Relative Ergodicity]
    Suppose $\{\Phi_t\}_{t\geq 0}$ is a bimodule quantum Markov semigroup.
    We say $\{\Phi_t\}_{t\geq 0}$ is relatively ergodic if $\Ker\cL=\cN$ or equivalently $\fix(\{\Phi_t\}_{t\geq 0})=\cN$.
\end{definition}

\begin{proposition}
 Suppose $\{\Phi_t\}_{t\geq 0}$ is a bimodule quantum Markov semigroup.
Then $\{\Phi_t\}_{t\geq 0}$ is relatively ergodic if and only if $\mathcal{R}(\mathfrak{F}^{-1}(\widehat{\cL}))=1-e_1$.
\end{proposition}

\begin{proof}
Note that  $x\in \Ker\cL$ if and only if $\cL(x)\Omega=0$, i.e. $\mathfrak{F}(\widehat{\cL})x\Omega=0$.
This shows that $\Ker\cL=\cN$ if and only if $\Ker(\mathfrak{F}(\widehat{\cL}))=\overline{\cN\Omega}^{\|\cdot\|}$.
Hence $\{\Phi_t\}_{t\geq 0}$ is relatively ergodic if and only if $\mathcal{R}(\mathfrak{F}^{-1}(\widehat{\cL}))=1-e_1$.
\end{proof}




\subsection{Poincar\'{e} Inequality}

In this section, we shall obtain the Poincar\'{e} inequalities for bimodule quantum Markov semigroups.

\begin{theorem}[Poincar\'{e} Inequality]\label{thm:poincare0}
    Suppose $\{\Phi_t\}_{t\geq 0}$ is a bimodule quantum Markov semigroup, $\cN$ is a factor and $\CS_0(\widehat{\cL}_0)=1$.
    Let $2\beta$ be the second maximal eigenvalue of $\mathfrak{F}^{-1}(\widehat{\cL}_0+\overline{\widehat{\cL}_0})$.
    Then for any $x\in \cM$ with $\bE_{\cN}(x)=0$,
    \begin{align}
        \tau(\Gamma(x,x)) \geq  (\widehat{\beta}-\beta)\tau(x^*x),
    \end{align}
    where $0\neq  \widehat{\beta}$ is the minimal eigenvalue of $\lambda^{-1/2} |\mathfrak{F}^{-1}(\widehat{\cL}_0^{1/2})|^2$.
\end{theorem}
\begin{proof}
We have that 
   \begin{align*}
 \tau((\Gamma(x,x))
 =& \frac{\lambda^{-1/2}}{2}\tau_1(( \partial x)^*( \partial x)) \\
 =& \frac{\lambda^{-3/2}}{2}\tau_2(e_2([x, \mathfrak{F}^{-1}(\widehat{\cL}_0^{1/2})])^*([x, \mathfrak{F}^{-1}(\widehat{\cL}_0^{1/2})]e_2)) \\
  =& \frac{\lambda^{-5/2}}{2}\tau_2(e_2[x, e_1]^*  \widehat{\cL}_0 [x, e_1]e_2)) \\
  =& \lambda^{-5/2}\frac{1}{2}(\tau_2(x^*x\bE_{\cM}(e_2e_1 \overline{\widehat{\cL}_0}  e_1e_2)) +\tau_2(xx^*\bE_{\cM}(e_2e_1 \widehat{\cL}_0  e_1e_2)))\\
  -& \frac{\lambda^{-1}}{2}\tau_2(xe_1x^* \mathfrak{F}^{-1}(\widehat{\cL}_0+\overline{\widehat{\cL}_0})).
   \end{align*}   
Now by the fact that $\cN$ is factor, we see that $e_1$ is a rank-one projection in $\cN'\cap \cM_1$.
By the fact that $\widehat{\cL}_0$ is connected and the Perron-Frobenius theorem for $\mathfrak{F}$-positive elements, there is a unique positive eigenvector for $\mathfrak{F}^{-1}(\widehat{\cL}_0+\overline{\widehat{\cL}_0})$ and
\begin{align*}
\mathfrak{F}^{-1}(\widehat{\cL}_0+\overline{\widehat{\cL}_0})\leq 2\|\mathfrak{F}^{-1}(\widehat{\cL}_0)\|e_1 +2\beta(1-e_1).     
\end{align*}
This implies that for any $x\in \cM$ with $\bE_{\cN}(x)=0$,
\begin{align*}
\frac{1}{2}\tau_2(xe_1x^* \mathfrak{F}^{-1}(\widehat{\cL}_0+\overline{\widehat{\cL}_0}))\leq \beta\tau_2(xe_1 x^*(1-e_1))\leq \lambda \beta \tau(x^*x).
\end{align*}
On the other hand, we have that 
\begin{align*}
\lambda^{-5/2}\tau_2(xx^*\bE_{\cM}(e_2e_1 \widehat{\cL}_0  e_1e_2))
= & \lambda^{-1/2}\tau_2(xx^*|\mathfrak{F}^{-1}(\widehat{\cL}_0^{1/2})|^2) \geq  \widehat{\beta}\tau_2(xx^*), \\
\lambda^{-5/2}\tau_2(x^*x\bE_{\cM}(e_2e_1 \widehat{\cL}_0  e_1e_2))
= & \lambda^{-1/2}\tau_2(x^*x|\mathfrak{F}^{-1}(\widehat{\cL}_0^{1/2})^*|^2) \geq  \widehat{\beta}\tau_2(x^*x).
\end{align*}
Now we shall show that $\widehat{\beta}>0$.
Suppose that $p$ is a projection in $\cN'\cap \cM$ such that $p\bE_{\cM}(e_2e_1 \widehat{\cL}_0  e_1e_2)=0$.
Then we have that 
\begin{align*}
p(\widehat{\cL}_0*\widehat{\cL}_0) p\leq \lambda^{-5/2} p \bE_{\cM}(e_2e_1 \widehat{\cL}_0  e_1e_2) p=0
\end{align*}
by the Schur product theorem.
Repeating the process, we see that $p \widehat{\cL}_0^{(*k)}=0$ for any $k\geq 2$.
Moreover, we have that $p\CS_0(\widehat{\cL}_0)=0$, i.e. $p=0$.
\end{proof}

If the inclusion is irreducible, we have much better estimation for $\bE_{\cM}(e_2e_1 \widehat{\cL}_0  e_1e_2)$.
\begin{theorem}[Poincar\'{e} Inequality]\label{thm:poincare1}
    Suppose $\{\Phi_t\}_{t\geq 0}$ is a bimodule quantum Markov semigroup.
    Assume that $\cN\subset \cM$ is irreducible.
    Let $B$ be spectral projection of $\mathfrak{F}^{-1}(\widehat{\cL}_0+\overline{\widehat{\cL}_0})$ corresponding to its maximal eigenvalue and $2\beta$ is the second maximal eigenvalue.
    Then for any $x\in \cM$ with $Bxe_1=0$,
    \begin{align}
        \tau(\Gamma(x,x)) \geq  (\lambda^{-1/2}\tau_2(\widehat{\cL}_0))-\beta)\tau(x^*x),
    \end{align}
\end{theorem}
\begin{proof}
By Corollary 6.12 in \cite{JLW16} and the Hausdorff-Young inequality, we have that $B$ is a biprojection and
\begin{align*}
  & \tau_2(xe_1x^* \mathfrak{F}^{-1}(\widehat{\cL}_0+\overline{\widehat{\cL}_0})) \\
  \leq &\tau_2(xe_1x^*( 2\lambda^{-1/2}\tau_2(\widehat{\cL}_0)B+2\beta(1-B)))\\
  =& 2\beta \tau_2(xe_1x^*)=2\lambda\beta\tau(x^*x).
\end{align*}
This implies that 
\begin{align*}
  \tau(\Gamma(x,x)) \geq& \lambda^{-1/2}\tau_2(\widehat{\cL}_0) \tau(x^*x) -   \beta\tau(x^*x)\\
  =&  (\lambda^{-1/2}\tau_2(\widehat{\cL}_0)-\beta)\tau(x^*x).
\end{align*}
Note that
\begin{align*}
   2\beta <  \|\mathfrak{F}^{-1}(\widehat{\cL}_0+\overline{\widehat{\cL}_0})\|=\lambda^{-1/2} \|\widehat{\cL}_0+\overline{\widehat{\cL}_0}\|_1=2\lambda^{-1/2} \tau(\widehat{\cL}_0).
\end{align*}
This completes the proof of the theorem.
\end{proof}

\section{Gradient Flow}\label{sec:: Gradient Flow}

In this section, we investigate bimodule quantum Markov semigroups that satisfy the bimodule detailed balance condition and establish the corresponding gradient flow equation.

\begin{definition}[Bimodule GNS Symmetric Semigroups]
Suppose $\{\Phi_t\}_{t\geq 0}$ is a quantum bimodule Markov semigroup and $\widehat{\Delta} \in \cM'\cap \cM_2$ is strictly positive and $\widehat{\Delta} e_2= e_2$.
We say $\{\Phi_t\}_{t\geq 0}$ is bimodule GNS symmetric with respect to $\widehat{\Delta}$ if $\overline{\widehat{\cL}}
= \widehat{\cL} \overline{\widehat{\Delta}}$ and $\Phi_t$ is equilibrium for all $t\geq 0$. 
\end{definition}

\begin{remark}
By Theorem \ref{thm:equivbalance}, we have that if $\{\Phi_t\}_{t\geq 0}$ satisfies $\rho$-detailed balance condition for some normal faithful state $\rho$ and $e_1\in \Dom(\sigma_{-i})$, then $\{\Phi_t\}_{t\geq 0}$ is bimodule GNS symmetry with respect to $\widehat{\Delta}_\rho$.
\end{remark}

\begin{remark}
Suppose that $\Phi_t$ is equilibrium and bimodule GNS symmetry with respect to $\widehat{\Delta}$ for all $t\geq 0$.
Then $\{\Phi_t\}_{t\geq 0}$ is bimodule GNS symmetry with respect to $\widehat{\Delta}$.
\end{remark}


\begin{proposition}\label{prop:bimoduleder}
    Suppose $\{\Phi_t\}_{t\geq 0}$ is bimodule GNS symmetric with respect to $\widehat{\Delta}$.
    Then
    \begin{align*}
        \widehat{\cL}\widehat{\Delta} =\widehat{\Delta} \widehat{\cL}, \quad 
    \widehat{\cL} \overline{\widehat{\Delta}} =\overline{\widehat{\Delta} } \widehat{\cL}, \quad 
    \cR(\widehat{\cL})\overline{\widehat{\Delta} }\widehat{\Delta} = \cR(\widehat{\cL}).
    \end{align*}
    Moreover,
    \begin{align*}
   \overline{\widehat{\cL}_0} =  \widehat{\cL}_0 \overline{\widehat{\Delta}}= 
   \overline{\widehat{\Delta}}\widehat{\cL}_0, \quad 
\overline{\widehat{\cL}_0\overline{\widehat{\Delta}^{1/2}}} =\overline{\widehat{\cL}_0\widehat{\Delta}^{-1/2}}
=\widehat{\cL}_0\widehat{\Delta}^{-1/2} = \widehat{\cL}_0\overline{\widehat{\Delta}^{1/2}}  
\end{align*}
\end{proposition}
\begin{proof}
By Proposition \ref{prop:gnsprop}, we have $\overline{\widehat{\cL}}= \widehat{\cL}\overline{\widehat{\Delta}}$, $\widehat{\cL}\widehat{\Delta} =\widehat{\Delta} \widehat{\cL}$, $
    \widehat{\cL} \overline{\widehat{\Delta}} =\overline{\widehat{\Delta} } \widehat{\cL}$.
Furthermore, we have $\cR(\widehat{\cL})\overline{\widehat{\Delta} }\widehat{\Delta} = \cR(\widehat{\cL})$.

Note that 
\begin{align*}
   \overline{\widehat{\cL}_0} 
   =& -(1-e_2)\overline{\widehat{\cL}}(1-e_2) 
   =  -(1-e_2) \widehat{\cL}\overline{\widehat{\Delta}}(1-e_2) \\
= & -(1-e_2) \widehat{\cL}(1-e_2)\overline{\widehat{\Delta}} 
= \widehat{\cL}_0\overline{\widehat{\Delta}}.
\end{align*}
The rest equalities are true by similar arguments.
\end{proof}

\begin{remark}
Suppose $\{\Phi_t\}_{t\geq 0}$ is a quantum bimodule Markov semigroup and $\widehat{\Delta} \in \cM'\cap \cM_2$ is strictly positive and $\widehat{\Delta} e_2= e_2$.
We say $\{\Phi_t\}_{t\geq 0}$ is bimodule KMS symmetric with respect to $\widehat{\Delta}$ if $\overline{\widehat{\cL}}
= \overline{\widehat{\Delta}}\widehat{\cL} \overline{\widehat{\Delta}}$, $\cR(\widehat{\cL})\overline{\widehat{\Delta}}=\cR(\widehat{\cL})\widehat{\Delta}^{-1}$ and $\Phi_t$ is equilibrium for all $t\geq 0$.

By Theorem \ref{thm:equivbalance}, we have that if $\{\Phi_t\}_{t\geq 0}$ is KMS symmetry with respect to a normal faithful state $\rho$ with $e_1\in \Dom(\sigma_{-i/2})$, then $\{\Phi_t\}_{t\geq 0}$ is bimodule KMS symmetry with respect to $\widehat{\Delta}_{\rho,1/2}$.

Suppose $\Phi_t$ is bimodule KMS symmetric with respect to $\widehat{\Delta}$ for all $t\geq 0$.
Then by differentiating with respect to $t$, we have that
\begin{align*}
\overline{\widehat{\cL}}= \overline{\widehat{\Delta}} \widehat{\cL}\overline{\widehat{\Delta}} ,\quad \widehat{\Delta}\overline{\widehat{\cL}}=  \widehat{\cL}\overline{\widehat{\Delta}} .
\end{align*}
By multiplying $1-e_2$, we have that
    \begin{align*}
   \overline{\widehat{\cL}_0} = \overline{\widehat{\Delta}} \widehat{\cL}_0 \overline{\widehat{\Delta}} , \quad 
\overline{\overline{\widehat{\Delta}^{1/2}} \widehat{\cL}_0\overline{\widehat{\Delta}^{1/2}}} 
=   \widehat{\Delta}^{1/2}\overline{\widehat{\cL}_0} \widehat{\Delta}^{1/2}
= \widehat{\Delta}^{-1/2}\widehat{\cL}_0{\widehat{\Delta}^{-1/2}}.
\end{align*}
If $\widehat{\Delta}\overline{\widehat{\Delta}}=\overline{\widehat{\Delta}} \widehat{\Delta}$, we see that $\overline{\widehat{\Delta}^{1/2}} \widehat{\cL}_0\overline{\widehat{\Delta}^{1/2}}$ will produce a $\tau$-symmetric bimodule quantum Markov semigroup.
\end{remark}

\begin{proposition}\label{prop:keeplap}
  Suppose $\{\Phi_t\}_{t\geq 0}$ is bimodule GNS symmetric with respect to $\widehat{\Delta}$.
  Then $\cL_w=0$.
\end{proposition}
\begin{proof}  
By Proposition \ref{prop:bimoduleder}, we have that $e_2\overline{\widehat{\cL}}(1-e_2)= e_2\widehat{\cL}(1-e_2)$, i.e. $\overline{\widehat{\cL}_1^*}=\widehat{\cL}_1$.
This implies that $\mathfrak{F}^{-1}(\widehat{\cL}_1)^*=\mathfrak{F}^{-1}(\widehat{\cL}_1)$.
Hence $\cL_w=0$ and $\cL=\cL_a$.    
\end{proof}

\begin{remark}
Suppose $\{\Phi_t\}_{t\geq 0}$ is bimodule KMS symmetry with respect to $\widehat{\Delta}$. 
Note that $e_2\overline{\widehat{\cL}}(1-e_2)=e_2\widehat{\cL}(1-e_2)\overline{\widehat{\Delta}}$.
This implies that $\overline{\widehat{\cL}_1^*}=\widehat{\cL}_1\overline{\widehat{\Delta}}$.
\end{remark}

\begin{proposition}
Suppose that $\cN\subseteq \cM$ is irreducible and $\{\Phi_t\}_{t \geq 0}$ is a bimodule quantum Markov semigroup.
Suppose that there exists a strictly positive element $\widehat{\Delta}\in \cM'\cap \cM_2$ with $\widehat{\Delta}e_2=e_2$ such that $\overline{\widehat{\cL}_0}=\overline{\widehat{\Delta}}\widehat{\cL}_0$.
Then $\{\Phi_t\}_{t\geq 0}$ is bimodule GNS symmetry with respect to $\widehat{\Delta}$.
 If $\cM'\cap \cM_2$ is commutative, then $\Phi_t$ is bimodule GNS symmetry.
\end{proposition}
\begin{proof}
By the fact that $\cN\subseteq \cM$ is irreducible, we have that $\cL_w=0$. 
Then we obtain that 
\begin{align*}
    \widehat{\cL}=\lambda^{-1} \tau_2(\widehat{\cL}_0)e_2-\widehat{\cL}_0.
\end{align*}
It is clear that $\widehat{\cL}\overline{\widehat{\Delta}}=\overline{\widehat{\cL}}$, i.e. $\cL$ is bimodule GNS symmetry with respect to $\widehat{\Delta}$.
Note that $\Phi_t$ is symmetry with respect to $\tau$.
We see that $\Phi_t$ is equilibrium for all $t\geq 0$.
Hence $\{\Phi_t\}_{t\geq 0}$ is bimodule GNS symmetry with respect to $\widehat{\Delta}$.
Note that $\cR(\widehat{\cL})=\cR(\overline{\widehat{\cL}})$.
We obtain that for $t\neq 0$, 
\begin{align*}
    \cR(\widehat{\Phi}_t)=\bigvee_{k \geq 0} \cR(\widehat{\cL}_0^{(*k)})= \bigvee_{k \geq 0} \cR(\overline{\widehat{\cL}_0}^{(*k)})=\cR(\overline{\widehat{\Phi}_t}).
\end{align*}
Note that $\widehat{\Phi}e_2=\overline{\widehat{\Phi}}e_2=\lambda^{-1}\tau_2(\widehat{\Phi}e_2)e_2$.
If $\cM'\cap \cM_2$ is commutative, then $\Phi_t$ is bimodule GNS symmetry.
\end{proof}

\begin{remark}\label{rem:haag1}
Suppose that $\cN\subset \cM$ is irreducible and $p$ is projection in $\cM'\cap \cM_2$ such that $p\overline{p}=0$.
Let $\widehat{\cL}_0=\kappa p+\kappa^{-1} \overline{p}$.
Then $\widehat{\cL}=\lambda^{-1}(\kappa+\kappa^{-1})\tau_2(p)e_2 -\kappa p-\kappa^{-1} \overline{p}$.
Then the corresponding bimodule quantum Markov semigroup is bimodule GNS symmetry with respect to $\widehat{\Delta}$, where $\widehat{\Delta}=e_2+\kappa^2 p+\kappa^{-2} \overline{p}+(1-e_2-p-\overline{p})$.

Suppose that $\mathcal{C}$ is a unitary fusion category with a generating object $\zeta$ such that $\zeta^*\neq \zeta$.
For the associated inclusion described in Section \ref{sec:fusion}, we let
\begin{align*}
\widehat{\cL}_{0, \kappa}=\kappa p_{\zeta}+\kappa^{-1} p_{\zeta^*}, \quad \kappa>0.
\end{align*}
Let $\widehat{\Delta}_{\kappa}=e_2+\kappa^2 p_{\zeta}+\kappa^{-2} p_{\zeta^*}+(1-e_2-p_{\zeta}-p_{\zeta^*})$.
We see that the semigroup $\{\Phi_t=e^{-t\cL}\}_{t\geq 0}$ given by $\displaystyle \widehat{\cL}=\tau_2(\widehat{\cL}_0)e_2-\kappa p_{\zeta}-\kappa^{-1} p_{ \zeta^*}$ is bimodule GNS symmetry.
\end{remark}

\begin{proposition}\label{prop:gnsformula}
Suppose that $\{\Phi_t\}_{t \geq 0}$ is a bimodule quantum Markov semigroup, $\Phi_t$ is equilibrium for all $t\geq 0$ and $\widehat{\Delta}\in \cM'\cap \cM_2$ is strictly positive element with $\widehat{\Delta}e_2=e_2$.
Then $\{\Phi_t\}_{t\geq 0}$ is bimodule GNS symmetry with respect to $\widehat{\Delta}$ if and only if $\overline{\widehat{\cL}_0}=\overline{\widehat{\Delta}}\widehat{\cL}_0$  and
\begin{align*}
    \mathfrak{F}^{-1}(\widehat{\Delta})* (1*\widehat{\cL}_0)=1*\widehat{\cL}_0,
\end{align*}
where the first convolution is taken in $\cN'\cap \cM_1$.
Pictorially, $\vcenter{\hbox{\begin{tikzpicture}[scale=0.95]
\draw [blue] (-0.5, -1.5)--(-0.5, 0) .. controls +(0, 0.6) and +(0,0.6).. (0.5, 0)--(0.5, -1.5);
\begin{scope}[shift={(0.5, -0.3)}]
\draw [fill=white] (-0.4, -0.3) rectangle (0.4, 0.3);
\node at (0, 0) {\tiny $1*\widehat{\cL}_0$};
\end{scope}
\begin{scope}[shift={(0, -1)}]
\draw [fill=white] (-0.7, -0.3) rectangle (0.7, 0.3);
\node at (0, 0) {\tiny $\overline{\widehat{\Delta}}$};    
\end{scope}
\end{tikzpicture}}}= \vcenter{\hbox{\begin{tikzpicture}[scale=0.95]
\draw [blue] (-0.5, -0.8)--(-0.5, 0) .. controls +(0, 0.6) and +(0,0.6).. (0.5, 0)--(0.5, -0.8);
\begin{scope}[shift={(0.5, -0.3)}]
\draw [fill=white] (-0.4, -0.3) rectangle (0.4, 0.3);
\node at (0, 0) {\tiny $1*\widehat{\cL}_0$};
\end{scope}
\end{tikzpicture}}}$.

\end{proposition}
\begin{proof}
Note that     
\begin{align*}
    \widehat{\cL}=\vcenter{\hbox{\begin{tikzpicture}[scale=0.65]
    \begin{scope}[shift={(0,1.5)}]
    \draw [blue] (-0.5, 0.8)--(-0.5, 0) .. controls +(0, -0.6) and +(0,-0.6).. (0.5, 0)--(0.5, 0.8);    
\begin{scope}[shift={(0.5, 0.3)}]
\end{scope}
    \end{scope}
\draw [blue] (-0.5, -0.8)--(-0.5, 0) .. controls +(0, 0.6) and +(0,0.6).. (0.5, 0)--(0.5, -0.8);
\begin{scope}[shift={(0.5, -0.3)}]
\draw [fill=white] (-0.3, -0.3) rectangle (0.3, 0.3);
\node at (0, 0) {\tiny $\mathbf{y}$};
\end{scope}
\end{tikzpicture}}}
+
\vcenter{\hbox{\begin{tikzpicture}[scale=0.65]
    \begin{scope}[shift={(0,1.5)}]
    \draw [blue] (-0.5, 0.8)--(-0.5, 0) .. controls +(0, -0.6) and +(0,-0.6).. (0.5, 0)--(0.5, 0.8);    
\begin{scope}[shift={(0.5, 0.3)}]
\draw [fill=white] (-0.3, -0.3) rectangle (0.3, 0.3);
\node at (0, 0) {\tiny $\mathbf{y}$};
\end{scope}
    \end{scope}
\draw [blue] (-0.5, -0.8)--(-0.5, 0) .. controls +(0, 0.6) and +(0,0.6).. (0.5, 0)--(0.5, -0.8);
\begin{scope}[shift={(0.5, -0.3)}]
\end{scope}
\end{tikzpicture}}}
-  \vcenter{\hbox{\begin{tikzpicture}[scale=0.6]
        \draw [blue] (0.2, -0.8) --(0.2, 0.8);
        \draw [blue] (-0.2, -0.8) --(-0.2, 0.8);
        \draw [fill=white] (-0.5, -0.4) rectangle (0.5, 0.4);
    \node at (0, 0) {\tiny $\widehat{\cL}_0$};
        \end{tikzpicture}}} ,
\end{align*}
where $\displaystyle \mathbf{y}=\frac{1}{2} (1*\widehat{\cL}_0)$.
We see that 
\begin{align*}
    \overline{\widehat{\Delta}}\widehat{\cL}
=\vcenter{\hbox{\begin{tikzpicture}[scale=0.65]
    \begin{scope}[shift={(0,1.5)}]
    \draw [blue] (-0.5, 0.8)--(-0.5, 0) .. controls +(0, -0.6) and +(0,-0.6).. (0.5, 0)--(0.5, 0.8);    
\begin{scope}[shift={(0.5, 0.3)}]
\end{scope}
    \end{scope}
\draw [blue] (-0.5, -1.5)--(-0.5, 0) .. controls +(0, 0.6) and +(0,0.6).. (0.5, 0)--(0.5, -1.5);
\begin{scope}[shift={(0.5, -0.3)}]
\draw [fill=white] (-0.3, -0.3) rectangle (0.3, 0.3);
\node at (0, 0) {\tiny $\mathbf{y}$};
\end{scope}
\begin{scope}[shift={(0, -1)}]
\draw [fill=white] (-0.7, -0.3) rectangle (0.7, 0.3);
\node at (0, 0) {\tiny $\overline{\widehat{\Delta}}$};    
\end{scope}
\end{tikzpicture}}}
+
\vcenter{\hbox{\begin{tikzpicture}[scale=0.65]
    \begin{scope}[shift={(0,1.5)}]
    \draw [blue] (-0.5, 0.8)--(-0.5, 0) .. controls +(0, -0.6) and +(0,-0.6).. (0.5, 0)--(0.5, 0.8);    
\begin{scope}[shift={(0.5, 0.3)}]
\draw [fill=white] (-0.3, -0.3) rectangle (0.3, 0.3);
\node at (0, 0) {\tiny $\mathbf{y}$};
\end{scope}
    \end{scope}
\draw [blue] (-0.5, -0.8)--(-0.5, 0) .. controls +(0, 0.6) and +(0,0.6).. (0.5, 0)--(0.5, -0.8);
\begin{scope}[shift={(0.5, -0.3)}]
\end{scope}
\end{tikzpicture}}}
-  \vcenter{\hbox{\begin{tikzpicture}[scale=0.6]
        \draw [blue] (0.2, -1.8) --(0.2, 0.8);
        \draw [blue] (-0.2, -1.8) --(-0.2, 0.8);
        \draw [fill=white] (-0.5, -0.4) rectangle (0.5, 0.4);
    \node at (0, 0) {\tiny $\widehat{\cL}_0$};
    \begin{scope}[shift={(0, -1)}]
  \draw [fill=white] (-0.5, -0.4) rectangle (0.5, 0.4);
    \node at (0, 0) {\tiny $\overline{\widehat{\Delta}}$};  
    \end{scope}
        \end{tikzpicture}}}.
\end{align*}
By the assumption, we have that $\overline{\widehat{\cL} }=\overline{\widehat{\Delta}}\widehat{\cL}$.
Similarly, we have that $\overline{\widehat{\cL}}=\widehat{\cL} \overline{\widehat{\Delta}}$.
We see that $\{\Phi_t\}_{t\geq 0}$ is bimodule GNS symmetry with respect to $\widehat{\Delta}$.

Suppose that $\{\Phi_t\}_{t\geq 0}$ is bimodule GNS symmetry with respect to $\widehat{\Delta}$.
We see that $\overline{\widehat{\cL}_0}=\overline{\widehat{\Delta}}\widehat{\cL}_0$  and $\mathfrak{F}^{-1}(\widehat{\Delta})* (1*\widehat{\cL}_0)=1*\widehat{\cL}_0$ by previous computation.
\end{proof}

\begin{corollary}
Suppose that $\cM'\cap \cM_2$ is commutative,  $\{\Phi_t\}_{t \geq 0}$ is a bimodule quantum Markov semigroup, $\Phi_t$ is equilibrium for all $t\geq 0$ and $\widehat{\Delta}\in \cM'\cap \cM_2$ is strictly positive element with $\widehat{\Delta}e_2=e_2$.
Then $\{\Phi_t\}_{t\geq 0}$ is bimodule GNS symmetry with respect to $\widehat{\Delta}$ if and only if $\overline{\widehat{\cL}_0}=\overline{\widehat{\Delta}}\widehat{\cL}_0$. 
\end{corollary}
\begin{proof}
If $\cM'\cap \cM_2$ is commutative, then 
\begin{align*}
\vcenter{\hbox{\begin{tikzpicture}[scale=0.95]
\draw [blue] (-0.5, -1.5)--(-0.5, 0) .. controls +(0, 0.6) and +(0,0.6).. (0.5, 0)--(0.5, -1.5);
\begin{scope}[shift={(0.5, -0.3)}]
\draw [fill=white] (-0.4, -0.3) rectangle (0.4, 0.3);
\node at (0, 0) {\tiny $1*\widehat{\cL}_0$};
\end{scope}
\begin{scope}[shift={(0, -1)}]
\draw [fill=white] (-0.7, -0.3) rectangle (0.7, 0.3);
\node at (0, 0) {\tiny $\overline{\widehat{\Delta}}$};    
\end{scope}
\end{tikzpicture}}}
=\vcenter{\hbox{\begin{tikzpicture}[scale=0.95]
\draw [blue] (-0.5, -1.5)--(-0.5, 0) .. controls +(0, 0.6) and +(0,0.6).. (0.5, 0)--(0.5, -1.5);
\begin{scope}[shift={(0.5, -1)}]
\draw [fill=white] (-0.4, -0.3) rectangle (0.4, 0.3);
\node at (0, 0) {\tiny $1*\widehat{\cL}_0$};
\end{scope}
\begin{scope}[shift={(0, -0.3)}]
\draw [fill=white] (-0.7, -0.3) rectangle (0.7, 0.3);
\node at (0, 0) {\tiny $\overline{\widehat{\Delta}}$};    
\end{scope}
\end{tikzpicture}}}
= \vcenter{\hbox{\begin{tikzpicture}[scale=0.95]
\draw [blue] (-0.5, -0.8)--(-0.5, 0) .. controls +(0, 0.6) and +(0,0.6).. (0.5, 0)--(0.5, -0.8);
\begin{scope}[shift={(0.5, -0.3)}]
\draw [fill=white] (-0.4, -0.3) rectangle (0.4, 0.3);
\node at (0, 0) {\tiny $1*\widehat{\cL}_0$};
\end{scope}
\end{tikzpicture}}}.
\end{align*}
Now by Proposition \ref{prop:gnsformula}, we see that the corollary is true.
\end{proof}

\begin{remark}
Suppose that $\{\Phi_t\}_{t \geq 0}$ is a bimodule quantum Markov semigroup, $\Phi_t$ is equilibrium for all $t\geq 0$ with $1*\widehat{\cL}_0=1$ and $\widehat{\Delta}\in \cM'\cap \cM_2$ is strictly positive element with $\widehat{\Delta}e_2=e_2$.
Then $\{\Phi_t\}_{t\geq 0}$ is bimodule GNS symmetry with respect to $\widehat{\Delta}$ if and only if $\overline{\widehat{\cL}_0}=\overline{\widehat{\Delta}}\widehat{\cL}_0$.

Suppose that $\bC\subset \bC^4$ is the inclusion.
Let $\widehat{\cL}_0\in \cM'\cap \cM_2$ such that $\mathfrak{F}(\widehat{\cL}_0)=
\begin{pmatrix}
0 & \frac{1}{3} &  \frac{1}{3} &  \frac{1}{3} \\
\frac{1}{2} & 0 & \frac{1}{4} & \frac{1}{4}  \\
\frac{1}{4} & \frac{1}{2} & 0 & \frac{1}{4} \\
 \frac{1}{6} &  \frac{1}{2} &  \frac{1}{3} & 0  
\end{pmatrix}.$
Then the associated semigroup is bimodule GNS symmetry, but not GNS symmetry.
\end{remark}

\begin{lemma}\label{lem:gnssupport}
  Suppose $\{\Phi_t\}_{t\geq 0}$ is bimodule GNS symmetric with respect to $\widehat{\Delta}$.
Then
\begin{align}
    \Ker\Gamma=\Ker\cL=\Ker\cL_a = \nfix(\{\Phi_t\}_{t\geq 0}) =\fix(\{\Phi_t\}_{t\geq 0}).
\end{align}
\end{lemma}
\begin{proof}
By Proposition \ref{prop:keeplap}, we have that $\Ker\cL_a=\Ker\cL$.
For any $x\in \Ker\Gamma$, we have that $\widehat{\cL}_0^{1/2} x e_1e_2=\widehat{\cL}_0^{1/2} e_1 xe_2$ by Proposition \ref{prop:gradientform2}.
By Proposition \ref{prop:bimoduleder}, we have that $\overline{\widehat{\cL}_0}^{1/2} x e_1e_2=\overline{\widehat{\cL}_0}^{1/2} e_1 xe_2$, i.e. $x \mathfrak{F}^{-1}(\widehat{\cL}_0^{1/2})^*= \mathfrak{F}^{-1}(\widehat{\cL}_0^{1/2})^* x$.
This implies that $x e_2e_1 \widehat{\cL}_0^{1/2}=e_2e_1\widehat{\cL}_0^{1/2} x$.
Now we have that 
\begin{align*}
    \cL(x)=\frac{1}{2}(1*\widehat{\cL}_0)x +\frac{1}{2} x(1*\widehat{\cL}_0) -x* \widehat{\cL}_0=0,
\end{align*}
i.e. $x\in \Ker\cL$.

Suppose that $x\in \Ker\cL$.
Then we have that $x\in \fix(\{\Phi_t\}_{t\geq 0})$.
Note that $\fix(\{\Phi_t\}_{t\geq 0})\subset \nfix(\{\Phi_t\}_{t\geq 0})\subset \Ker\Gamma$.
We see the lemma is true.
\end{proof}

\begin{remark}
Suppose that $\{\Phi_t\}_{t\geq 0}$ is bimodule GNS symmetry with respect to $0<\widehat{\Delta}\in \cM'\cap \cM_2$.
If $\CS(\widehat{\cL}_0)=\CS_0(\widehat{\cL}_0)=1$, then the semigroup is relatively ergodic.
In fact, $\CS_0(\widehat{\cL}_0)=1$ implies that $\ker\Gamma=\cN$.
By Lemma \ref{lem:gnssupport}, we see that the semigroup is relative ergodic.
Note that $\CS(\widehat{\cL}_{0, \kappa})=1$.
We see that the semigroup defined in Remark \ref{rem:haag1} is relatively ergodic.
\end{remark}

\begin{theorem}\label{thm:gnslimit}
 Suppose $\{\Phi_t\}_{t\geq 0}$ is bimodule GNS symmetric with respect to $\widehat{\Delta}$.
 Then $\displaystyle \lim_{t\to \infty} \widehat{\Phi}_t$ exists.
If moreover $\{\Phi_t\}_{t\geq 0}$ is relatively ergodic and $\Phi_t$ is bimodule GNS symmetric with respect to $\widehat{\Delta}$ for all $t\geq 0$, then 
 \begin{align}
 \lim_{t\to \infty} \widehat{\Phi}_t= \bE_{\cM_1}\left( \overline{\widehat{\Delta}} \right)^{-1}.
 \end{align}
\end{theorem}
\begin{proof}
By Proposition \ref{prop:keeplap}, we see that $\cL=\cL_a$.
By Lemma \ref{lem:gnssupport}, we have that  $\Ker\Gamma=\Ker\cL_a=\Ker\cL= \nfix(\{\Phi_t\}_{t\geq 0})$.
By Theorem \ref{thm:bilimit}, we obtain that $\displaystyle \lim_{t\to \infty} \widehat{\Phi}_t=\widehat{\bE}_{\Phi}$.
When $\{\Phi_t\}_{t\geq 0}$ is relatively ergodic, we have $\Ker \cL=\cN$. 
Note that $\displaystyle \lim_{t\to \infty} \bE_{\cN} \Phi_t(x)=\lim_{t\to\infty} \Phi_t(x)$ for any $x\in \cM$.
We see that $ \widehat{\bE}_\Phi * \widehat{\bE}_{\cN}=\widehat{\bE}_\Phi$.
This implies that $\widehat{\bE}_\Phi= \bE_{\cM_1}(\widehat{\bE}_\Phi)$, so $\widehat{\mathbb{E}}_{\Phi}\in \cM'\cap \cM_1$.
Since $\Phi_t$ is bimodule GNS symmetry with respect to $\widehat{\Delta}$ for all $t\geq 0$, we have $\overline{\widehat{\bE}_\Phi}= \widehat{\bE}_\Phi\overline{\widehat{\Delta}}$. 
Pictorially, we have
\begin{align*}
    \widehat{\mathbb{E}}_{\Phi}\overline{\widehat{\Delta}} = \vcenter{\hbox{\begin{tikzpicture}[scale=1.2]
        \draw [blue] (0, -0.5)--(0, 1) (0.5, -0.5)--(0.5, 1);
        \draw [fill=white] (-0.2, -0.2) rectangle (0.2, 0.2);
        \node at (0, 0) {\tiny $\widehat{\mathbb{E}}_{\Phi}$};
        \draw [fill=white] (-0.2, 0.3) rectangle (0.7, 0.7);
        \node at (0.25, 0.5) {\tiny $\overline{\widehat{\Delta}}$};
    \end{tikzpicture}}} = \vcenter{\hbox{\begin{tikzpicture}[scale=1.2]
        \draw [blue] (0, -0.5)--(0, 1) (0.5, -0.5)--(0.5, 1);
        \begin{scope}[xshift=0.5cm]
            \draw [fill=white] (-0.2, -0.2) rectangle (0.2, 0.2);
        \node at (0, 0) {\tiny $\overline{\widehat{\mathbb{E}}_{\Phi}}$};
        \end{scope}
    \end{tikzpicture}}} = \overline{\widehat{\mathbb{E}}_{\Phi}}. 
\end{align*}
Hence $\widehat{\bE}_\Phi\bE_{\cM_1}(\overline{\widehat{\Delta}}) = \mathbb{E}_{\cM_1}(\widehat{\mathbb{E}}_{\Phi}\overline{\widehat{\Delta}}) = \mathbb{E}_{\cM_1}(\overline{\widehat{\mathbb{E}_{\Phi}}}) =1$, and $\widehat{\bE}_\Phi= \bE_{\cM_1}\left( \overline{\widehat{\Delta}} \right)^{-1}$.
\end{proof}

\begin{notation}
We denote $F_{\Phi}=\displaystyle \lim_{t\to \infty} \widehat{\Phi}_t$. 
We have $F_{\Phi}* F_{\Phi} = F_{\Phi}$ when it exists. 
If $\{\Phi_t\}_{t\geq 0}$ is relatively ergodic, we have that $F_\Phi\in \cM'\cap \cM_1$.
\end{notation}


\begin{corollary}\label{cor:gnslimitdual}
 Suppose $\{\Phi_t\}_{t\geq 0}$ is bimodule GNS symmetric with respect to $\widehat{\Delta}$  and relatively ergodic.
 Then for any $D\in \cM_+$, we have that 
 \begin{align*}
 \lim_{t\to \infty}  \Phi_t^*(D)= \bE_{\cN}(D)\overline{F_{\Phi}},
 \end{align*} 
 where the contragrredient is taken in $\cN'\cap \cM_1$.   
 If $\Phi_t$ is bimodule GNS symmetry with respect to $\widehat{\Delta}$ for all $t\in \bR$, then for any $D\in \cM_+$, we have that 
 \begin{align*}
 \lim_{t\to \infty}  \Phi_t^*(D)= \bE_{\cN}(D)\gamma_{1, +}(\overline{\bE_{\cM_1}\left( \overline{\widehat{\Delta}} \right)^{-1}}),
 \end{align*}    
 where $\bE_{\cM_1}\left( \overline{\widehat{\Delta}} \right)^{-1}$ is viewed as an element in $\cM'\cap \cM_2$ and the contragredient is taken in the space $\cM'\cap \cM_2$ and $\gamma_{1, +}: \cM_1'\cap \cM_3 \to \cN'\cap \cM_1$ is the pullback.
\end{corollary}
\begin{proof}
By Theorem \ref{thm:gnslimit}, we have that
\begin{align*}
 \lim_{t\to\infty}\Phi_t^* (D)=&\lim_{t\to \infty} \lambda^{-5/2} \bE_{\cM}(e_2e_1 \overline{\widehat{\Phi}_t}De_1e_2)    \\
 =& \lambda^{-5/2} \bE_{\cM}(e_2e_1\overline{\bE_{\cM_1}\left( \overline{\widehat{\Delta}} \right)^{-1}} De_1e_2),
\end{align*}
where $\bE_{\cM_1}\left( \overline{\widehat{\Delta}} \right)^{-1}$ is viewed as an element in $\cM'\cap \cM_2$ and the contragredient is taken in $\cM'\cap \cM_2$.
Continuing the computation, we have that 
 \begin{align*}
 \lim_{t\to\infty}\Phi_t^* (D)
 =& \lambda^{-5/2} \bE_{\cM}(e_2\overline{\bE_{\cM_1}\left( \overline{\widehat{\Delta}} \right)^{-1}} e_1 De_1e_2) \\
 =& \lambda^{-5/2} \bE_{\cM}(e_2\overline{\bE_{\cM_1}\left( \overline{\widehat{\Delta}} \right)^{-1}} \bE_{\cN} (D)e_1e_2) \\
 =& \bE_{\cN} (D)\lambda^{-5/2} \bE_{\cM}(e_2\overline{\bE_{\cM_1}\left( \overline{\widehat{\Delta}} \right)^{-1}} e_1e_2) \\
 =& \bE_{\cN}(D)\gamma_{1, +}(\overline{\bE_{\cM_1}\left( \overline{\widehat{\Delta}} \right)^{-1}}).
\end{align*}
This completes the computation.
\end{proof}

\begin{remark}\label{rem:gnslimit}
In Corollary \ref{cor:gnslimitdual}, the pictorial representation of $\overline{\bE_{\cM_1}\left( \overline{\widehat{\Delta}} \right)^{-1}}$ is the following:
\begin{align*}
\lambda^{-1/2}
\vcenter{\hbox{\begin{tikzpicture}[scale=0.8]
 \draw [blue] (0.2, -0.8) --(0.2, 0.8);
 \draw [fill=white] (-0.5, -0.4) rectangle (0.5, 0.4);
 \node at (0, 0) {\tiny $\widehat{\Delta}$};
\draw [blue] (-0.2, 0.4).. controls +(0, 0.35) and +(0, 0.35)..(-0.8, 0.4)--(-0.8, -0.4);
\draw [blue] (-0.2, -0.4).. controls +(0, -0.35) and +(0, -0.35)..(-0.8, -0.4);
        \end{tikzpicture}}} ^{-1}.
\end{align*}
If $\widehat{\Delta}=  \vcenter{\hbox{\begin{tikzpicture}[scale=1.2]
        \draw [blue] (0, -0.5)--(0, 0.5) (0.5, -0.5)--(0.5, 0.5);
        \draw [fill=white] (-0.2, -0.2) rectangle (0.2, 0.2);
        \node at (0, 0) {\tiny $\overline{\Delta}$};
        \begin{scope}[shift={(0.5, 0)}]
         \draw [fill=white] (-0.2, -0.2) rectangle (0.2, 0.2);
        \node at (0, 0) {\tiny $\Delta^{-1}$};   
        \end{scope}
    \end{tikzpicture}}}$, then $\overline{\bE_{\cM_1}\left( \overline{\widehat{\Delta}} \right)^{-1}}=\Delta$.
This indicates that if $\widehat{\Delta}$ comes from a faithful normal state $\rho$, then $\displaystyle \lim_{t\to \infty}\Phi_t^*(D)=\bE_{\cN}(D) \Delta_\rho$.
If $\cN=\bC$, then $\displaystyle \lim_{t\to \infty}\Phi_t^*(D)=\tau(D) \Delta_\rho$.
\end{remark}

\begin{remark}
Suppose that $\cN\subset \cM$ is irreducible and $\{\Phi_t\}_{t\geq 0}$ is bimodule GNS symmetry with respect to $\widehat{\Delta}$.
Then $\displaystyle \lim_{t\to \infty} \widehat{\Phi}_t$ is a multiple of biprojection.
If $\{\Phi_t\}_{t\geq 0}$ is relatively ergodic, then $\displaystyle \lim_{t\to \infty} \widehat{\Phi}_t=\lambda^{1/2}$.

Suppose that $\{\Phi_t\}_{t\geq 0}$ is GNS symmetry with respect to a normal faithful state $\rho$ such that $e_1\in \Dom(\sigma_{-i})$.
We have $\overline{\widehat{\bE}_\Phi}= \widehat{\bE}_\Phi\overline{\widehat{\Delta}_\rho}$.
This implies that $\widehat{\Delta}_\rho=1$. 
However, by considering the semigroup in Remark \ref{rem:haag1}, we see that the semigroup is not GNS symmetry with respect to any faithful normal state $\rho$ with $e_1\in \Dom(\sigma_{-i})$.
\end{remark}

\subsection{Derivation under Bimodule GNS Symmetry}

In the following, we shall consider $\widehat{\cL}_0\widehat{\Delta}^{-1/2}$ and study its spectral decomposition. 
Let $\mathscr{A}$ be the $C^*$-algebra generated by $\widehat{\mathcal{L}}_0, \widehat{\Delta}\mathcal{R}(\widehat{\cL}_0),\overline{\widehat{\Delta}}\mathcal{R}(\widehat{\cL}_0)$ with identity $\mathcal{R}(\widehat{\cL}_0)$.
By Proposition \ref{prop:bimoduleder}, $\mathscr{A}$ is commutative. 
Therefore we can find a complete set of orthogonal minimal projections $p_j$, $j\in \sI_{0}\cup \sI_1$ and non-negative reals $\omega_j,\mu_j$, that has the following properties:
\begin{enumerate}
    \item $\displaystyle \widehat{\cL}_0\widehat{\Delta}^{-1/2}=\sum_{j\in \sI_{0}\cup \sI_1} \omega_{j} p_{j}$;
    \item $\displaystyle \widehat{\Delta}=\sum_{j\in \sI_{0}\cup \sI_1}  \mu_j p_{j} +1-\cR(\widehat{\cL}_0)$;
    \item for each $j$ there exists $j^*\in \sI_{0}\cup \sI_1$ such that $\overline{p_j} = p_{j^*}$. 
\end{enumerate}
Note that (3) follows from $\overline{\widehat{\cL}_0\widehat{\Delta}^{-1/2}} = \widehat{\cL}_0\widehat{\Delta}^{-1/2}$ and that $\mathscr{A}$ is invariant under the contragradient. 
As a consequence, 
\begin{align*}
    \omega_{j}=\omega_{j^*},\quad j\in \sI_{0}\cup \sI_1.
\end{align*}
\begin{remark}
We have that for any $j$ with $\omega_{j}\neq 0$, $\mu_j\mu_{j^*}=1$ which follows from $\cR(\widehat{\cL})\overline{\widehat{\Delta} }\widehat{\Delta} = \cR(\widehat{\cL})$.
\end{remark}

\begin{example}
    Suppose $\cN=\bC$ and $\cM=M_n(\bC)$.
 We have that
 \begin{align*}
 \widehat{\cL}_0\widehat{\Delta}^{-1/2}= \sum_{j\in \sI} \omega_j \vcenter{\hbox{\begin{tikzpicture}[scale=0.6]
        \draw [blue] (0.2, -0.8) --(0.2, 0.8);
        \draw [blue] (-0.2, -0.8) --(-0.2, 0.8);
        \draw [fill=white] (-0.5, -0.4) rectangle (0.5, 0.4);
    \node at (0, 0) {\tiny $p_j$};
        \end{tikzpicture}}}  
    =\sum_{j\in \sI} \omega_j\vcenter{\hbox{\begin{tikzpicture}[scale=0.65]
    \begin{scope}[shift={(0,1.5)}]
    \draw [blue] (-0.5, 0.8)--(-0.5, 0) .. controls +(0, -0.6) and +(0,-0.6).. (0.5, 0)--(0.5, 0.8);    
\begin{scope}[shift={(0.5, 0.3)}]
\draw [fill=white] (-0.3, -0.3) rectangle (0.3, 0.3);
\node at (0, 0) {\tiny $v_j$};
\end{scope}
    \end{scope}
\draw [blue] (-0.5, -0.8)--(-0.5, 0) .. controls +(0, 0.6) and +(0,0.6).. (0.5, 0)--(0.5, -0.8);
\begin{scope}[shift={(0.5, -0.3)}]
\draw [fill=white] (-0.3, -0.3) rectangle (0.3, 0.3);
\node at (0, 0) {\tiny $v_j^*$};
\end{scope}
\end{tikzpicture}}}, \quad 
\widehat{\Delta}=\sum_{j\in \sI} \mu_j\vcenter{\hbox{\begin{tikzpicture}[scale=0.65]
    \begin{scope}[shift={(0,1.5)}]
    \draw [blue] (-0.5, 0.8)--(-0.5, 0) .. controls +(0, -0.6) and +(0,-0.6).. (0.5, 0)--(0.5, 0.8);    
\begin{scope}[shift={(0.5, 0.3)}]
\draw [fill=white] (-0.3, -0.3) rectangle (0.3, 0.3);
\node at (0, 0) {\tiny $v_j$};
\end{scope}
    \end{scope}
\draw [blue] (-0.5, -0.8)--(-0.5, 0) .. controls +(0, 0.6) and +(0,0.6).. (0.5, 0)--(0.5, -0.8);
\begin{scope}[shift={(0.5, -0.3)}]
\draw [fill=white] (-0.3, -0.3) rectangle (0.3, 0.3);
\node at (0, 0) {\tiny $v_j^*$};
\end{scope}
\end{tikzpicture}}}+1-\cR(\widehat{\cL}_0).
\end{align*}
\end{example}

Now we see that 
\begin{align*}
 \widehat{\cL}_0=\sum_{j\in \sI} \mu_j^{1/2}\omega_{j} p_{j}.
\end{align*}
Moreover, we define the balanced derivation $\partial^\Delta:\cM\to \cM_1$ with respect to $\widehat{\Delta}$ as: 
\begin{align*}
    \partial^\Delta x =[x, \mathfrak{F}^{-1}(\widehat{\cL}_0^{1/2}\widehat{\Delta}^{-1/4})].
\end{align*}
Equivalently, 
\begin{align*}
    (\partial^\Delta x )e_2= \lambda^{-1/2}\widehat{\cL}_0^{1/2}\widehat{\Delta}^{-1/4} [x, e_1] e_2, \quad x\in \cM.
\end{align*}
Now, the bimodule GNS symmetry implies $\overline{\partial^\Delta}=\partial^\Delta$. 
The $j$-th directional derivation $\partial_{j}^{\Delta}$ of $\partial^\Delta$ is defined as follows:
\begin{align*}
    \partial_{j}^{\Delta} x = &  \omega_{j}^{1/2}[x, \mathfrak{F}^{-1}(p_{j})],   \quad j \in  \sI.
\end{align*}
We now obtain that $\overline{\partial_j^\Delta (x)}=-\partial_{j^*}^\Delta(\overline{x})$ for $x\in\cM$ and 
\begin{align*}
    \partial_j=\mu_j^{1/4}\partial_{j}^{\Delta}, \quad  j \in \sI_{0}\cup \sI_{1}.
\end{align*}

Let $\displaystyle \Gamma^\Delta=\frac{\lambda^{-1/2}}{2}(\partial^\Delta )^*\partial^\Delta=\frac{\lambda^{-1/2}}{2}\sum_{j\in \sI} \mu_j^{-1/2}\partial_j^* \partial_j$ be the modified gradient form.
\begin{theorem}
 Suppose $\{\Phi_t\}_{t\geq 0}$ is bimodule GNS symmetric with respect to $\widehat{\Delta}$ and $\CS(\widehat{\cL}_0) =1$.    
Let $\beta$ be the second maximal eigenvalue of $\mathfrak{F}^{-1}(\widehat{\cL}_0 \overline{\widehat{\Delta}}^{1/2})$.
    Then for any $x\in \cM$ with $\bE_{\cN}(x)=0$,
    \begin{align}
        \tau(\Gamma^\Delta(x,x)) \geq (\widehat{\beta}-\beta)\tau(x^*x),
    \end{align}
 where $0\neq   \widehat{\beta}$ is the minimal eigenvalue of $\lambda^{-1/2}\bE_{\cM}\left(\left|\mathfrak{F}^{-1}( \widehat{\cL}_0^{1/2}  \overline{\widehat{\Delta}}^{1/4})\right|^2\right)$.
\end{theorem}
\begin{proof}
We have that 
   \begin{align*}
 \tau((\Gamma^\Delta(x,x))
 =& \frac{\lambda^{-1/2}}{2}\tau_1(( \partial^\Delta x)^*( \partial^\Delta x)) \\
 =& \frac{\lambda^{-3/2}}{2}\tau_2(e_2([x, \mathfrak{F}^{-1}(\widehat{\cL}_0^{1/2}\widehat{\Delta}^{-1/4})])^*([x, \mathfrak{F}^{-1}(\widehat{\cL}_0^{1/2}\widehat{\Delta}^{-1/4})]e_2)) \\
  =& \frac{\lambda^{-5/2}}{2}\tau_2(e_2[x, e_1]^*  \widehat{\cL}_0\widehat{\Delta}^{-1/2} [x, e_1]e_2)) \\
  =& \lambda^{-5/2}\tau_2(x^*x\bE_{\cM}(e_2e_1 \widehat{\cL}_0 \widehat{\Delta}^{-1/2} e_1e_2)) \\
  -& \frac{\lambda^{-1}}{2}\tau_2(xe_1x^* \mathfrak{F}^{-1}(\widehat{\cL}_0\widehat{\Delta}^{-1/2})).
   \end{align*}   
Now by the fact that $\Ker \cL=\Ker \cL_a=\cN$ and the Perron-Frobenius theorem for $\mathfrak{F}$-positive elements \cite[Theorem 3.10]{HJLW23}, there is a unique positive eigenvector for $\mathfrak{F}^{-1}(\widehat{\cL}_0 \overline{\widehat{\Delta}}^{1/2})$ and
\begin{align*}
\mathfrak{F}^{-1}(\widehat{\cL}_0 \overline{\widehat{\Delta}}^{1/2})\leq \|\mathfrak{F}^{-1}(\widehat{\cL}_0)\|e_1 +\beta(1-e_1).     
\end{align*}
This implies that for any $x\in \cM$ with $\bE_{\cN}(x)=0$,
\begin{align*}
\frac{1}{2}\tau_2(xe_1x^* \mathfrak{F}^{-1}(\widehat{\cL}_0 \overline{\widehat{\Delta}}^{1/2}))\leq \beta\tau_2(xe_1 x^*(1-e_1))\leq \lambda \beta \tau(x^*x).
\end{align*}
On the other hand, we have that 
\begin{align*}
& \lambda^{-5/2}\tau_2(x^*x\bE_{\cM}(e_2e_1 \widehat{\cL}_0  \overline{\widehat{\Delta}}^{1/2} e_1e_2))\\
=& \lambda^{-1/2}\tau_2(x^*x\bE_{\cM}\left(\left|\mathfrak{F}^{-1}( \widehat{\cL}_0^{1/2}  \overline{\widehat{\Delta}}^{1/4})\right|^2\right))\geq  \widehat{\beta}\tau_2(x^*x).
\end{align*}
Combining the two inequalities above, we see that the theorem is true.
\end{proof}

In the following, we shall extend the domain of $\partial_j^\Delta$ to $\cM_2$.
\begin{lemma}\label{lem:derivation}
For any $j\in \sI_0\cup \sI_1$ and $x\in \cM$, we have that 
\begin{align*}
\widehat{\Delta}^{1/4}(\partial_j^\Delta x) e_2 = & (\partial_j x )e_2 , \\
\widehat{\Delta}^{1/2}(\partial_j^\Delta x \widehat{\Delta}^{-1}) e_2 = & \omega_j^{1/2}(\mu_j^{-1/2} x\mathfrak{F}^{-1}(p_j)-\mu_j^{1/2}\mathfrak{F}^{-1}(p_j) x) e_2.
\end{align*}
Consequently, $\lambda^{-1}\bE_{\cM_1}(\widehat{\Delta}^{1/4}(\partial_j^\Delta x) e_2)=\partial_j x$, and
\begin{align}\label{eq:derivation2}
\lambda^{-1}\bE_{\cM_1}(\widehat{\Delta}^{1/2}(\partial_j^\Delta x \widehat{\Delta}^{-1}) e_2)
=\omega_j^{1/2}(\mu_j^{-1/2} x\mathfrak{F}^{-1}(p_j)-\mu_j^{1/2}\mathfrak{F}^{-1}(p_j) x).
\end{align}
Furthermore, $\partial x = \lambda^{-1}\bE_{\cM_1}(\widehat{\Delta}^{1/4}(\partial^\Delta x) e_2)$ and 
\begin{align*}
 \sum_{j\in \sI_0\cup \sI_1} \omega_j^{1/2} ( \mu_j^{-1/2} x\mathfrak{F}^{-1}(p_j)-\mu_j^{1/2}\mathfrak{F}^{-1}(p_j) x)
 = \lambda^{-1}\bE_{\cM_1}(\widehat{\Delta}^{1/2}(\partial^\Delta x \widehat{\Delta}^{-1}) e_2).
\end{align*}
\end{lemma}
\begin{proof}
The first equality is true by a routine computation.
For any $x\in \cM$, we have that
\begin{align*}
 \widehat{\Delta}^{1/2}(\partial_j^\Delta x \widehat{\Delta}^{-1}) e_2
=& \lambda^{-1/2}\widehat{\Delta}^{1/2}(\omega_j^{1/2} p_j [x \widehat{\Delta}^{-1}, e_1])e_2\\
=& \lambda^{-1/2}\mu_j^{1/2}(\omega_j^{1/2} p_j (\mu_j^{-1}xe_1-e_1x))e_2\\
=& \lambda^{-1/2}(\omega_j^{1/2} p_j (\mu_j^{-1/2}xe_1- \mu_j^{1/2}e_1x))e_2 \\
=&\omega_j^{1/2}( \mu_j^{-1/2}x\mathfrak{F}^{-1}(p_j) e_2 - \mu_j^{1/2}\mathfrak{F}^{-1}(p_j) x e_2).
\end{align*}
This implies the second equality.
By taking the conditional expectation $\bE_{\cM}$, we have Equation \eqref{eq:derivation2}.
Summing over $j$, we see the last equality is true.
\end{proof}

\begin{remark}
Suppose $\cN=\bC$ and $\cM=M_n(\bC)$. 
We have that 
\begin{align*}
    \widehat{\Delta}^{1/2}(\partial_j^\Delta x \widehat{\Delta}^{-1}) e_2 
    = & \omega_j^{1/2} (\mu_j^{-1/2}xv_j-\mu_j^{1/2} v_j x)\otimes \overline{v_j^*}e_2.
\end{align*}
\end{remark}

There are alternative formulae for the derivations.
\begin{lemma}\label{lem:derivation0}
For any $j\in \sI_0\cup \sI_1$ and $x\in \cM$, we have that 
\begin{align*}
e_2 (\partial_j^\Delta x) \widehat{\Delta}^{-1/4} = & e_2(\partial_j x ) , \\
e_2(\partial_j^\Delta x \widehat{\Delta}^{-1}) \widehat{\Delta}^{1/2} = & \omega_j^{1/2} e_2(\mu_j^{-1/2} x\mathfrak{F}^{-1}(p_j)-\mu_j^{1/2}\mathfrak{F}^{-1}(p_j) x).
\end{align*}
Consequently, $\lambda^{-1}\bE_{\cM_1}(e_2(\partial_j^\Delta x)\widehat{\Delta}^{-1/4} )=\partial_j x$, and
\begin{align}\label{eq:derivation3}
\lambda^{-1}\bE_{\cM_1}(e_2 (\partial_j^\Delta x \widehat{\Delta}^{-1}) \widehat{\Delta}^{1/2})=\omega_j^{1/2} (\mu_j^{-1/2} x\mathfrak{F}^{-1}(\overline{p_j})-\mu_j^{1/2}\mathfrak{F}^{-1}(\overline{p_j}) x).
\end{align}
Furthermore, $\partial x = \lambda^{-1}\bE_{\cM_1}(e_2 (\partial^\Delta x) \widehat{\Delta}^{-1/4})$ and 
\begin{align*}
 \sum_{j\in \sI_0\cup \sI_1} \omega_j^{1/2}\left(  \mu_j^{-1/2} x\mathfrak{F}^{-1}(p_j)-\mu_j^{1/2}\mathfrak{F}^{-1}(p_j) x\right) 
 = \lambda^{-1}\bE_{\cM_1}(e_2 (\partial^\Delta x \widehat{\Delta}^{-1}) \widehat{\Delta}^{1/2}).
\end{align*}
\end{lemma}
\begin{proof}
For any $x\in \cM$, we have that
\begin{align*}
 e_2(\partial_j^\Delta x \widehat{\Delta}^{-1}) \widehat{\Delta}^{1/2}
=& \lambda^{-1/2}\omega_j^{1/2} e_2  [x \widehat{\Delta}^{-1}, e_1] \overline{p_j}\widehat{\Delta}^{1/2} \\
=& \lambda^{-1/2}\mu_j^{-1/2}\omega_j^{1/2} e_2 (xe_1- \mu_j e_1x)\overline{p_j}\\
=&\lambda^{-1/2}\omega_j^{1/2} e_2 (\mu_j^{-1/2}xe_1- \mu_j^{1/2} e_1x)\overline{p_j} \\
=&\omega_j^{1/2}e_2( \mu_j^{-1/2}x\mathfrak{F}^{-1}(\overline{p_j})  - \mu_j^{1/2}\mathfrak{F}^{-1}(\overline{p_j}) x ).
\end{align*}
The rest equalities follows directly.
\end{proof}

\begin{remark}
    We have that $\overline{\partial}x =\lambda^{-1}\bE_{\cM_1}(\widehat{\Delta}^{-1/4}(\partial^\Delta x)e_2) = \lambda^{-1}\bE_{\cM_1}(e_2(\partial^\Delta x)\widehat{\Delta}^{1/4})$.
\end{remark}

\begin{proposition}\label{prop:gnslap1}
For any $x\in \cM$, 
\begin{align}
\cL(x)= & \frac{\lambda^{-2}}{2} \bE_{\cM}(  \widehat{\Delta}^{-1/4}(\partial^\Delta x)e_2 e_1\overline{\widehat{\cL}_0^{1/2} } ) 
-\frac{\lambda^{-2}}{2} \bE_{\cM}( \overline{\widehat{\cL}_0^{1/2}}e_1e_2 (\partial^\Delta x)\widehat{\Delta}^{-1/4}).
\end{align}
\end{proposition}
\begin{proof}
Now by Lemmas \ref{lem:laplacian1} and \ref{lem:derivation0}, we have that 
\begin{align*}
\cL_a (x)= & -\frac{\lambda^{-1/2}}{2} \bE_{\cM}(\mathfrak{F}^{-1}(\overline{\widehat{\cL}_0^{1/2}}) (\partial x )) + \frac{\lambda^{-1/2}}{2} \bE_{\cM}(  (\overline{\partial} x) \mathfrak{F}^{-1}(\widehat{\cL}_0^{1/2} ) ) \\
= & -\frac{\lambda^{-3/2}}{2} \bE_{\cM}(\mathfrak{F}^{-1}(\overline{\widehat{\cL}_0^{1/2}})e_2 (\partial^\Delta x)\widehat{\Delta}^{-1/4}) + \frac{\lambda^{-3/2}}{2} \bE_{\cM}(  \widehat{\Delta}^{-1/4}(\partial^\Delta x)e_2 \mathfrak{F}^{-1}(\widehat{\cL}_0^{1/2} ) ) \\
=& -\frac{\lambda^{-2}}{2} \bE_{\cM}( \overline{\widehat{\cL}_0^{1/2}}e_1e_2 (\partial^\Delta x)\widehat{\Delta}^{-1/4}) + \frac{\lambda^{-2}}{2} \bE_{\cM}(  \widehat{\Delta}^{-1/4}(\partial^\Delta x)e_2 e_1\overline{\widehat{\cL}_0^{1/2} } ).
\end{align*}
This completes the computation.
\end{proof}

\begin{proposition}
    Suppose $\cN=\bC$ and $\cM=M_n(\bC)$.
    We have that 
    \begin{align*}
    \cL(x)=& \frac{1}{2}\sum_{j \in \sI_0}\omega_j \mu_j^{-1/2}[x, v_j]v_j^*-\omega_j \mu_j^{1/2}v_j^* [x, v_j],\\
    =& \frac{1}{2}\sum_{j\in \sI_0} \omega_j\mu_j^{1/2} ([x, v_j^*]v_j-v_j^*[x, v_j]),
\end{align*}
subject to $\tau(v_jv_k^*)=\lambda^{1/2} \delta_{j,k}$, $\tau(v_j)=0$,  and $\omega_j \geq 0$.
\end{proposition}
\begin{proof}
We have that 
\begin{align*}
& \frac{\lambda^{-1}}{2} \bE_{\cM}(  \widehat{\Delta}^{-1/4}(\partial^\Delta x)e_2 e_1\overline{\widehat{\cL}_0^{1/2} } ) \\
=& \frac{\lambda^{-1}}{2}\sum_{j\in  \sI_0} \mu_j^{-1/4} \omega_j^{1/2}\mu_j^{-1/4}\omega_j^{1/2}\bE_{\cM} (p_j[x, \mathfrak{F}^{-1}(p_j)]e_2e_1  p_j)\\
=& \frac{\lambda^{-1/2}}{2}\sum_{j\in  \sI_0} \mu_j^{-1/2}  \omega_j\bE_{\cM} ([x, \mathfrak{F}^{-1}(p_j)]e_2\mathfrak{F}^{-1}(p_j)^*)\\
=& \frac{\lambda}{2}\sum_{j\in \sI_0}\mu_j^{-1/2}\omega_j[x, v_j]v_j^*,
\end{align*}
and 
\begin{align*}
&\frac{\lambda^{-1}}{2} \bE_{\cM}( \overline{\widehat{\cL}_0^{1/2}}e_1e_2 (\partial^\Delta x)\widehat{\Delta}^{-1/4})\\
=& \frac{\lambda^{-1}}{2}\sum_{j\in  \sI_0} \mu_j^{1/4} \omega_j^{1/2}\mu_j^{1/4}\omega_j^{1/2}\bE_{\cM} (\overline{p_j}e_1e_2[x, \mathfrak{F}^{-1}(p_j)]\overline{p_j})\\
=& \frac{\lambda}{2}\sum_{j\in \sI_0}\mu_j^{1/2} \omega_j v_j^*[x, v_j].
\end{align*}
This implies that 
\begin{align*}
    \cL(x)=& \frac{1}{2}\sum_{j \in \sI_0}\omega_j \mu_j^{-1/2}[x, v_j]v_j^*-\omega_j \mu_j^{1/2}v_j^* [x, v_j],\\
    =& \frac{1}{2}\sum_{j\in \sI_0} \omega_j\mu_j^{1/2} ([x, v_j^*]v_j-v_j^*[x, v_j]).
\end{align*}
This completes the computation.
\end{proof}

\begin{proposition}
For any $x\in \cM$,
\begin{align}
    \cL^*(x) = & \frac{\lambda^{-3/2}}{2}\bE_{\cM}\left( \partial^\Delta \bE_{\cM_1}( e_2(\partial^\Delta x\widehat{\Delta}^{-1}) \widehat{\Delta}^{1/2})\right)\nonumber\\
    =&  \frac{\lambda^{-1/2}}{2}\sum_{j\in \sI_0\cup \sI_1} \omega_j^{1/2} \bE_{\cM}\left( \partial^\Delta\left( \mu_j^{-1/2} x\mathfrak{F}^{-1}(p_j)-\mu_j^{1/2}\mathfrak{F}^{-1}(p_j) x \right) \right).
\end{align}
\end{proposition}
\begin{proof}
For any $x, y\in \cM$, we have that 
\begin{align*}
\lambda^2 \tau(\cL^*(y)^*x) = &\lambda^2 \tau(y^* \cL(x)) \\
=& \frac{1}{2} \tau_2\left(y^*\widehat{\Delta}^{-1/4}(\partial^\Delta x)e_2 e_1\overline{\widehat{\cL}_0^{1/2} } \right) -\frac{1}{2} \tau_2\left( y^*\overline{\widehat{\cL}_0^{1/2}}e_1e_2 (\partial^\Delta x)\widehat{\Delta}^{-1/4}\right) \\
=& \frac{1}{2} \tau_2\left((\partial^\Delta x)e_2 e_1 \overline{\widehat{\cL}_0^{1/2} } \widehat{\Delta}^{-1/4} y^* \right) -\frac{1}{2} \tau_2\left( (\partial^\Delta x) y^*\widehat{\Delta}^{-1/4} \overline{\widehat{\cL}_0^{1/2}} e_1 e_2\right) \\
=& \frac{1}{2} \lambda^{1/2}\tau_2\left((\partial^\Delta x)e_2 \mathfrak{F}^{-1}(  \widehat{\cL}_0^{1/2} \overline{\widehat{\Delta}}^{-1/4}) y^* \right) -\frac{1}{2} \tau_2\left( (\partial^\Delta x) y^*\widehat{\Delta}^{-1/4} \overline{\widehat{\cL}_0^{1/2}} e_1 e_2\right) \\
=& \frac{1}{2} \tau_2\left((\partial^\Delta x) \widehat{\cL}_0^{1/2} \overline{\widehat{\Delta}}^{-1/4}e_1e_2 y^* \right) -\frac{1}{2} \tau_2\left( (\partial^\Delta x) y^*\widehat{\Delta}^{-3/4} \widehat{\cL}_0^{1/2} e_1 e_2\right) \\
=& \frac{1}{2} \tau_2\left((\partial^\Delta x) \widehat{\cL}_0^{1/2} \widehat{\Delta}^{1/4}(e_1 y^* -\widehat{\Delta}^{-1}y^*e_1)e_2 \right)\\
=& \frac{1}{2} \tau_2\left((\partial^\Delta x) \widehat{\Delta}^{1/2} \widehat{\cL}_0^{1/2} \widehat{\Delta}^{-1/4}(e_1 y^*\widehat{\Delta}^{-1} -y^*\widehat{\Delta}^{-1}e_1)e_2 \right)\\
=& \frac{1}{2} \lambda^{1/2}\tau_2\left((\partial^\Delta x) \widehat{\Delta}^{1/2} (\partial^\Delta y\widehat{\Delta}^{-1})^* e_2 \right) \\
=& \frac{1}{2}  \lambda^{1/2} \tau_1\left((\partial^\Delta x) \bE_{\cM_1}(\widehat{\Delta}^{1/2} (\partial^\Delta y\widehat{\Delta}^{-1})^* e_2) \right) \\
=& \frac{1}{2} \lambda^{1/2} \tau_1\left( x (\partial^\Delta \bE_{\cM_1}( e_2(\partial^\Delta y\widehat{\Delta}^{-1}) \widehat{\Delta}^{1/2})^* \right). 
\end{align*}
The rest follows from Lemma \ref{lem:derivation0}.
This completes the computation.
\end{proof}

\begin{proposition}
    Suppose that $\cN=\bC$ and $\cM=M_n(\bC)$.
    We have that 
    \begin{align*}
         \cL^*(x)=\frac{1}{2}\sum_{j\in \sI_0}\omega_j[(\mu_j^{-1/2} xv_j -\mu_j^{1/2}v_j x), v_j^*], 
    \end{align*}
    subject to $\tau(v_jv_k^*)=\lambda^{1/2} \delta_{j,k}$, $\tau(v_j)=0$,  and $\omega_j \geq 0$.
\end{proposition}
\begin{proof}
We have that 
\begin{align*}
 \cL^*(x)
 = &   \frac{\lambda^{-1/2}}{2}\sum_{j\in \sI_0 } \omega_j^{1/2} \bE_{\cM}\left( \partial^\Delta\left( \mu_j^{-1/2} x\mathfrak{F}^{-1}(p_j)-\mu_j^{1/2}\mathfrak{F}^{-1}(p_j) x \right) \right) \\
 =& \frac{1}{2}\sum_{j\in \sI_0}\omega_j[(\mu_j^{-1/2} xv_j -\mu_j^{1/2}v_j x), v_j^*].
\end{align*}
This completes the computation.
\end{proof}

\subsection{Divergence} 

Suppose that $D\in \cM$ is strictly positive and $\mu>0$. 
We define the linear map $\K_{D, \mu}:\cM_1\to \cM_1$ as follows: 
\begin{align*}
\K_{D, \mu} (x)=\int_0^1 \mu^{1-2s}D^{s} x D^{1-s} ds, \quad x\in \cM_1.   
\end{align*}
Let $f(s)=\mu^{1-2s}D^{s} v D^{1-s}$ for $v\in \cM_1$.
Then 
\begin{align*}
f'(s)= & (\mu^{-1}D)^s((\log\mu^{-1}D)v-v\log\mu D) (\mu D)^{1-s} \\
=& (\mu^{-1}D)^s( (\log D)v-v\log D) (\mu D)^{1-s} -2(\log \mu)(\mu^{-1}D)^s v (\mu D)^{1-s}.
\end{align*}
This implies that 
\begin{align*}
 \K_{D, \mu}((\log\mu^{-1} D ) v- v\log \mu D ) 
 = \mu^{-1} D v  - \mu v D, \quad v\in \cM_1.
\end{align*}
The inverse of $\K_{D, \mu}$ is known to be 
\begin{align*}
  \K_{D, \mu}^{-1}(x) = \int_0^\infty (s+\mu^{-1} D)^{-1} x  (s+\mu D)^{-1} ds, \quad x\in \cM_1.
\end{align*}
For any $j\in \sI_0\cup \sI_1$, we let $\K_{D, j}=\K_{D, \mu_j^{-1/2}}$ and 
\begin{align*}
\widetilde{\K}_{D, j}(x)=\K_{D, j}(x-\lambda^{-1}\bE_{\cM_1}((\partial_j^\Delta \log \widehat{\Delta} )e_2)), \quad x\in \cM_1.
\end{align*}

\begin{remark}\label{rem:hidden}
Suppose that $ \widehat{\Delta}=  \vcenter{\hbox{\begin{tikzpicture}[scale=1.2]
        \draw [blue] (0, -0.5)--(0, 0.5) (0.5, -0.5)--(0.5, 0.5);
        \draw [fill=white] (-0.2, -0.2) rectangle (0.2, 0.2);
        \node at (0, 0) {\tiny $\overline{\Delta}$};
        \begin{scope}[shift={(0.5, 0)}]
         \draw [fill=white] (-0.2, -0.2) rectangle (0.2, 0.2);
        \node at (0, 0) {\tiny $\Delta^{-1}$};   
        \end{scope}
    \end{tikzpicture}}}$, for some strictly positive $\Delta\in \cN'\cap \cM$.
Then 
\begin{align*}
 \bE_{\cM_1}((\partial_j^\Delta \log \widehat{\Delta} )e_2))
 =\omega_j^{1/2}\bE_{\cM_1}([\log \widehat{\Delta}, \mathfrak{F}^{-1}(p_j)] e_2)
 =\lambda \omega_j^{1/2}[\log \Delta, \mathfrak{F}^{-1}(p_j)].
\end{align*}
This indicates that the item $\bE_{\cM_1}((\partial_j^\Delta \log \widehat{\Delta} )e_2))$ contains more information.
\end{remark}

\begin{theorem}\label{thm:adjointform}
Suppose that $D\in \cM$ is strictly positive.
We have that
\begin{align*}
   \cL^*(D)= \frac{\lambda^{-1/2}}{2}\sum_{j\in \sI_0\cup \sI_1} \bE_{\cM}\left( \partial^\Delta \widetilde{\K}_{D,j^*}\left( \partial_j^\Delta(\log   D ) \right) \right).
\end{align*}
\end{theorem}
\begin{proof}
Note that $(\log \widehat{\Delta}) e_2=0$. 
We have that 
\begin{align*}
& \lambda^{1/2}\cL^*(D)\\
= & \frac{1}{2}\sum_{j\in \sI_0\cup \sI_1} \omega_j^{1/2} \bE_{\cM}\left( \partial^\Delta\left( \mu_j^{-1/2} D\mathfrak{F}^{-1}(p_j)-\mu_j^{1/2}\mathfrak{F}^{-1}(p_j) D \right) \right) \\
=&  \frac{1}{2}\sum_{j\in \sI_0\cup \sI_1} \omega_j^{1/2} \bE_{\cM}\left( \partial^\Delta \K_{D,\mu_j^{1/2}}\left( (\log \mu_j^{-1/2} D)\mathfrak{F}^{-1}(p_j)-\mathfrak{F}^{-1}(p_j) (\log \mu_j^{1/2} D) \right) \right) \\
=& \frac{1}{2}\sum_{j\in \sI_0\cup \sI_1} \omega_j^{1/2} \bE_{\cM}\left( \partial^\Delta \K_{D,\mu_j^{1/2}}\left( (\log   D)\mathfrak{F}^{-1}(p_j)-\mathfrak{F}^{-1}(p_j) (\log   D)- \mathfrak{F}^{-1}(p_j)\log \mu_j \right) \right) \\
=& \frac{\lambda^{-1}}{2}\sum_{j\in \sI_0\cup \sI_1}  \omega_j^{1/2} \bE_{\cM}\left( \partial^\Delta \K_{D,j^*}\left( \bE_{\cM_1}((\log   D \widehat{\Delta}^{-1})\mathfrak{F}^{-1}(p_j)e_2-\mathfrak{F}^{-1}(p_j) (\log   D\widehat{\Delta}^{-1})e_2) \right) \right) \\
=& \frac{1}{2}\sum_{j\in \sI_0\cup \sI_1} \bE_{\cM}\left( \partial^\Delta \widetilde{\K}_{D,j^*}\left(  \partial_j^\Delta(\log   D)  \right) \right).
\end{align*}
This completes the computation.
\end{proof}

\begin{remark}
    Theorem \ref{thm:adjointform} generalizes Theorem 5.10 in \cite{CarMaa17}
\end{remark}

For bimodule GNS symmetric semigroups, we define the modified divergence and gradient as follows. 
The divergence $\Div: \cM_1^{|\sI|} \to \cM$ is defined as
\begin{align*}
    \Div (x_j)_{j\in \sI}=\lambda^{-1/4}\sum_{j\in \sI}\bE_{\cM} (\partial_{j}^{\Delta *} x_j), 
\end{align*}
and the gradient $\nabla: \cM\to \cM_1^{|\sI|}$ is defined as 
\begin{align*}
    \nabla x=\lambda^{-1/4} (\partial_j^\Delta x)_{j\in \sI},
\end{align*}
where $x_j\in \cM_1$, $j\in \sI$ and $x\in \cM$.

\begin{lemma}\label{lem:divrange}
Suppose the bimodule quantum Markov semigroup $\{\Phi_t\}_{t\geq 0}$ is relatively ergodic.
Then the range of $\Div$ is
\begin{align*}
\Ran( \Div)= \{x\in \cM: \bE_{\cN}(x)=0\}.
\end{align*}
\end{lemma}
\begin{proof}
Note that $\Ker \nabla=\cN$.
We have that 
\begin{align*}
\Ran ( \Div)=(\Ker\nabla)^\perp=\cN^\perp=\{x\in \cM: \bE_{\cN}(x)=0\}.
\end{align*}
This completes the proof of the Lemma.
\end{proof}

\begin{remark}
Note that $\cM \subset \Ker \partial_j^{\Delta *}$.
We have that $\cM^{|\sI|}\subset \Ker \Div$.
\end{remark}

We extend $\K_{D, j}$ to the linear map $\K_D: \cM_1^{|\sI|}\to \cM_1^{|\sI|}$ as 
\begin{align*}
    \K_D (x_j)_{j \in \sI} =(\K_{D, j} x_j)_{j\in \sI},
\end{align*}
where $(x_j)_{j \in \sI}\in \cM^{|\sI|}$.
The inverse of $\K_D$ is $\K_D^{-1}=(\K_{D, j}^{-1})_{j \in \sI}$.

\begin{definition}
Suppose  $X=(x_j)_{j\in \sI}$  and $Y=(y_j)_{j\in \sI}$, where $x_j, y_j\in \cM_1$.
We define $\langle X, Y\rangle_{D, \widehat{\Delta}}$ as follows:
\begin{align*}
    \langle X, Y \rangle_{D, \widehat{\Delta}} =\sum_{j \in \sI}\tau_1((\K_{D, j} y_j)^* x_j).
\end{align*}
The sesquilinear form $\langle \cdot, \cdot\rangle_{D, \widehat{\Delta}}$ is an inner product on $\cM_1^{|\sI|}$.
The induced norm is denoted by $\|\cdot\|_{D, \widehat{\Delta}}$.
\end{definition}

Suppose $X, Y\in \cM_1^{|\sI|}$.
We denote by $X\perp_{D, \widehat{\Delta}} Y$ if $\langle X, Y \rangle_{D, \widehat{\Delta}}=0$.
We denote by $\langle \cdot, \cdot \rangle$ the usual inner product, i.e. $\displaystyle \langle X, Y\rangle =\sum_{j\in \sI}\tau_1(y_j^*x_j)$.
Hence we have that 
\begin{align*}
    \langle X, Y\rangle_{D, \widehat{\Delta}} =\langle \K_{D} X, Y\rangle=\langle X, \K_D Y\rangle.
\end{align*}

Suppose that $\Psi:\cM_1\to \cM_1$ is a trace-preserving completely positive map.
Denote by $\Psi^{|\sI|}: \cM_1^{|\sI|} \to \cM_1^{|\sI|}$ the map $(x_j)_{j \in \sI} \mapsto (\Psi(x_j))_{j \in \sI}$.
Then 
\begin{align}\label{eq:concavity}
   \left \langle \Psi^{|\sI|}( X),\K_{\Psi(D)}^{-1}\Psi^{|\sI|}(X)\right \rangle \leq \langle X, \K_D^{-1}(X) \rangle,
\end{align}
which follows from the Lieb's concavity theorem \cite{Lie73,Pet85} for the following map:
\begin{align*}
    (x, D) \mapsto \int_0^\infty \tau_1((s+\mu D)^{-1} x(s+\mu D)^{-1} x^*) ds.
\end{align*}

\begin{theorem}
Suppose the bimodule GNS symmetric quantum Markov semigroup $\{\Phi_t\}_{t\geq 0}$ is relatively ergodic,
$\{D_s\}_{s\in [-\epsilon, \epsilon]}$ is a continuous family of density operators in $\cM$ with $\bE_{\cN}(D_s)$ independent of $s$.
Then there exists $X=(x_j)_{j \in \sI}\in \cM_1^{|\sI|}$ such that
\begin{align}
    \dot{D}_0=\frac{1}{2}\Div \K_{D} (X).
\end{align}
If $X$ is minimal subject to $\|\cdot\|_{D, \widehat{\Delta}}$, then there exists a unique self-adjoint element $x\in \cM$ such that $X=\nabla x$ and $\bE_{\cN}(x)=0$.
\end{theorem}
\begin{proof}
Note that $\bE_{\cN}(\dot{D}_0)=0$.
By the relative ergodicity and Lemma \ref{lem:divrange}, there exists $X\in \cM_1^{|\sI|}$ such that $\displaystyle \dot{D}_0=\frac{1}{2}\Div \K_D (X)$.

Suppose $X\in \cM_1^{|\sI|}$ is minimal subject to the norm $\|\cdot\|_{D, \widehat{\Delta}}$ and $A\in \cM_1^{|\sI|}$ with $\Div A=0$.
Let $Y=X+r \K_D^{-1} A$, $r\in\bC$.
Then $\displaystyle \dot{D}_0=\frac{1}{2}\Div \K_D (Y)$ and $\langle X, X\rangle_{D, \widehat{\Delta}}\leq \langle Y, Y\rangle_{D, \widehat{\Delta}}$.
This implies that 
\begin{align*}
\langle X, \K_D^{-1} A\rangle_{D,\widehat{\Delta}}=0=\langle X, A\rangle .
\end{align*}
Hence $X\perp \Ker \Div$, i.e. $X\in \Ran (\nabla)$.
Therefore, there exists $x_1\in \cM$ such that $\nabla x_1=X$.
Note that the adjoint $x_1^*$ also satisfies $\displaystyle \dot{D}_0=\frac{1}{2}\Div \K_D (\nabla x_1^*)$.
Then $\displaystyle x_2=\frac{1}{2}(x_1+x_1^*)$ is also a solution.
Let $x=x_2-\bE_{\cN}(x_2)$.
We see that $x$ is a solution.

Suppose that $y$ is a solution such that $y=y^*$ and $\bE_{\cN}(y)=0$.
We have that $\nabla (x-y)=0$.
This implies that $x-y\in \cN$, i.e. $x=y$.
We see that $x$ is unique subject to $x=x^*$ and $\bE_{\cN}(x)=0$.
\end{proof}

Let $X_\Delta\in \cM$ such that
\begin{align*}
\left(\lambda^{-1}\bE_{\cM_1}((\partial_j^\Delta \log \widehat{\Delta} )e_2)\right)_{j \in \sI} =\nabla X_\Delta +(x_j)_{j\in \sI},
\end{align*}
where $(x_j)_{j\in \sI} \perp_{D, \widehat{\Delta}} \Ran(\nabla)$.
Note that 
\begin{align*}
\lambda^{-1}\bE_{\cM_1}((\partial_j^\Delta \log \widehat{\Delta} )e_2=(\log\mu_j )\mathfrak{F}^{-1}(p_j).
\end{align*}
By the fact that $(\log\mu_j )\mathfrak{F}^{-1}(p_j)^*=-\log \mu_{j^*}\mathfrak{F}^{-1}(p_{j^*})$.
We have that $X_\Delta$ can be chosen as a self-adjoint element and $\bE_{\cN}(X_{\Delta})=0$. 
Let $D_\Delta=e^{X_\Delta}$.
We shall call the element $D_{\Delta}$ as a \emph{hidden density}. 
Note that this hidden density depends on the strictly positive element $D$.

\begin{remark}
Suppose that $\widehat{\Delta}=  \vcenter{\hbox{\begin{tikzpicture}[scale=1.2]
        \draw [blue] (0, -0.5)--(0, 0.5) (0.5, -0.5)--(0.5, 0.5);
        \draw [fill=white] (-0.2, -0.2) rectangle (0.2, 0.2);
        \node at (0, 0) {\tiny $\overline{\Delta}$};
        \begin{scope}[shift={(0.5, 0)}]
         \draw [fill=white] (-0.2, -0.2) rectangle (0.2, 0.2);
        \node at (0, 0) {\tiny $\Delta^{-1}$};   
        \end{scope}
    \end{tikzpicture}}}$.
Then by Remark \ref{rem:hidden}, we see that $X_\Delta$ can be taken as $\log \Delta$ and $D_\Delta=\Delta$.
By Remark \ref{rem:gnslimit}, we have that $\overline{\bE_{\cM_1}\left( \overline{\widehat{\Delta}} \right)^{-1}}=\gamma_{1, +}^{-1}(\Delta)$.
However, $D_\Delta$ might not be taken as $\overline{F_\Phi}$ in general.
\end{remark}

\begin{proposition}
Suppose that $D$ is a strictly positive element in $\cM$.
We have that 
\begin{align*}
\cL^*(D)=  \frac{1}{2} \Div\K_D \left( \nabla \log D - \nabla \log D_\Delta\right).
\end{align*}
\end{proposition}
\begin{proof}
We have that 
    \begin{align*}
\cL^*(D)
= & \frac{\lambda^{-1/2}}{2}\sum_{j\in \sI_0\cup \sI_1} \bE_{\cM}\left( \partial^\Delta \widetilde{\K}_{D,j}\left( \partial_j^\Delta(\log   D ) \right) \right)\\
=&\frac{\lambda^{-1/4}}{2} \Div\left (\widetilde{\K}_{D,j} ( \partial_j^\Delta(\log   D ))\right)_{j \in \sI}\\
=& \frac{\lambda^{-1/4}}{2} \Div\K_D \left( ( \partial_j^\Delta(\log   D ))_{j\in \sI} - \left(\lambda^{-1}\bE_{\cM_1}((\partial_j^\Delta \log \widehat{\Delta} )e_2)\right)_{j \in \sI} \right)\\
=&  \frac{1}{2} \Div\K_D \left( \nabla \log D - \nabla \log D_\Delta\right).
    \end{align*}
    This completes the computation.
\end{proof}



\subsection{Relative Entropy}

In this section, we study the relative entropy functional with respect to the hidden density. 
For a strictly positive density operator $D\in \cM_+$, we define 
\begin{align*}
    \cM_{+}(D) = \{\rho\in \cM_+\vert \bE_{\cN}(\rho) = \bE_{\cN}(D)\}.
\end{align*}
\begin{definition}[Riemannian Metric]
Let $D$ be a strictly positive density operator in $\cM_+$, and suppose $\{D_s\}_{s\in (a,b)}$ is a continuous path in $\cM_+(D)$ with $D_0 = D$, where $a<0< b$. 
Define the Riemannian metric $g_{\cL}$ on $\cM_+$ with respect to $\{\Phi_t\}_{t\geq 0}$ is defined as 
\begin{align*}
    \|\dot{D}_0\|_{g_{\cL}}=\min\left\{\|X\|_{D, \widehat{\Delta}}: \dot{D}_0=\frac{1}{2}\Div \K_{D} (X)\right\}.
\end{align*}
\end{definition}

Suppose $f:\cM_+\to \bR$ is a differentiable function.
For any self-adjoint $x\in \cM$ with $\tau(x)=0$, there exists $\displaystyle \frac{df}{dD}\in \cM$ such that 
\begin{align*}
    \lim_{s\to 0} \frac{f(D+sx)-f(D)}{s} =\tau\left(\frac{df}{dD} x\right).
\end{align*}

\begin{definition}[Gradient Vector Field]
Suppose $f$ is a differentiable function on $\cM_+$.
The gradient vector field $\grad_{g, D} f\in \{\nabla x: x=x^*\in \cM, \bE_{\cN}(x)=0\}$ is defined as
\begin{align*}
    \left.\frac{d}{ds}f(D_s)\right|_{s=0}= \frac{1}{2}\left\langle \grad_{g, D} f, \nabla x\right\rangle_{D, \widehat{\Delta}},
\end{align*}
where $\displaystyle \dot{D}_0=\frac{1}{2}\Div\K_D(\nabla x)$ and $D_0=D$.
Note that the gradient vector field is unique.
\end{definition}

\begin{lemma}\label{lem:vecterfieldform}
Suppose $f$ is a differentiable function. 
Then 
\begin{align*}
   \grad_{g, D}f= \nabla\frac{df}{dD}.
\end{align*}
and 
\begin{align*}
    \| \grad_{g, D}f\|_{D, \widehat{\Delta}}^2= \left\langle \Div \K_D \nabla \frac{df}{dD}, \frac{df}{dD}\right\rangle.
\end{align*}
\end{lemma}
\begin{proof}
For any differentiable path $\{D_s\}_s$ of density operators, by definition 
\begin{align*}
 \left.\frac{d}{ds}f(D_s)\right|_{s=0}
 =& \left\langle \frac{df}{dD},  \dot{D}_0\right\rangle \\
=& \frac{1}{2}\left\langle \frac{df}{dD},  \Div\K_D (\nabla x )\right\rangle \\
=&\frac{1}{2}\left\langle \nabla\frac{df}{dD}, \K_D (\nabla x )\right\rangle\\
=& \frac{1}{2}\left\langle \nabla\frac{df}{dD},  \nabla x \right\rangle_{D, \widehat{\Delta}}
\end{align*}
This implies that $\displaystyle  \grad_{g, D}f= \nabla\frac{df}{dD}$.
\end{proof}


Let $f(\rho)=H(\rho\| D_{\Delta})$ be the relative entropy functional on $\cM_{+}(D)$. 
Then a differentiable path $\{D_s\}_s$ is a gradient flow (See \cite{Maa11}) for $f$ if $\displaystyle \dot{D}_s=\frac{1}{2}\Div\K_D \left(\nabla\frac{df}{dD_s} \right)$. 
Then we have the following theorem:

\begin{theorem}
Suppose that $\{\Phi_t\}_{t\geq 0}$ is a bimodule quantum Markov semigroup bimodule GNS symmetric with respect to $\widehat{\Delta}$.
Suppose $\{D_s\}_{s\in (-\epsilon, \epsilon)}$ is a differentiable path of density operators in $\cM_{+}(D)$ satisfying
\begin{align*}
    D_0 = D,\quad \dot{D}_s=\cL^*(D_s) \quad \forall s\geq0. 
\end{align*}
Then it is the gradient flow for the relative entropy functional $f(\rho) = H(\rho\| D_{\Delta})$.
\end{theorem}
\begin{proof}
For any $x\in \cM$ with $\bE_{\cN}(x)=0$, we have 
\begin{align*}
    \left.\frac{d}{ds}f(D+tx)\right|_{s=0}=& \tau(x(\log D-\log D_\Delta))+\int_0^\infty \tau\left(\frac{D}{(D+r)^2}x\right)dr\\
    =& \tau(x(\log D-\log D_\Delta)).
\end{align*}
Hence $\displaystyle \frac{df}{dD}=\log D-\log D_{\Delta}$.
On the other hand, we have that 
\begin{align*}
 \cL^*(D_s) =&  \frac{1}{2} \Div\K_D \left( \nabla \log D_s - \nabla \log D_\Delta\right)\\
 =&  \frac{1}{2} \Div\K_D \left( \nabla \frac{df}{dD_s}\right).
\end{align*}
This completes the proof of the theorem.
\end{proof}

\begin{corollary}\label{cor:diffgrad}
Suppose that $\{D_s\}_{s\in (-\epsilon, \epsilon)}$ is a differentiable path of density operators such that $\bE_{\cN}(D_s)$ is independent of $s$.
Then 
\begin{align*}
   \left \|\frac{d}{ds} D_s\right\|_{g_{\cL}} = \|\grad_{g, D} f\|_{D, \widehat{\Delta}}^2,
\end{align*}
where $f(D)=H(D\| D_{\Delta})$.
\end{corollary}
\begin{proof}
We have that 
\begin{align*}
  \left \|\frac{d}{ds} D_s\right\|_{g_{\cL}}
  = & \min\left\{\|X\|_{D_s, \widehat{\Delta}}: \frac{d}{ds}D_s=\frac{1}{2}\Div \K_{D_s} (X).\right\}\\
  =& \left\| \nabla \frac{df}{dD_s}\right\|_{D, \widehat{\Delta}}\\
  =&\|\grad_{g, D} f\|_{D, \widehat{\Delta}}.
\end{align*}
This completes the computation.
\end{proof}

Now we shall show the bimodule version of the modified logarithm Sobolev inequality and Talagrand inequality under the intertwining property introduced by Carlen and Maas\cite{CarMaa17}. 
\begin{definition}[Intertwining Property]
Suppose that $\{\Phi_t\}_{t\geq 0}$ is a bimodule quantum Markov semigroup bimodule GNS symmetric with respect to $\widehat{\Delta}$. 
We say that $\{\Phi_t\}_{t\geq 0}$ satisfies the intertwining property if there exist $\beta>0$ and a trace-preserving completely positive map $\widetilde{\Phi_t}^*: \cM_1\to \cM_1$ such that $\widetilde{\Phi_t}^*|_{\cM}=\Phi^*_t$ and
\begin{align*}
   \Phi_t^* \Div = e^{-\beta t}\Div \widetilde{\Phi_t}^*, \quad  t\geq 0.
\end{align*}   
\end{definition}

\begin{theorem}[Bimodule Logarithmic Sobolev Inequality]\label{thm:entropydecay}
Suppose that $\{\Phi_t\}_{t\geq 0}$ is a bimodule quantum Markov semigroup bimodule GNS symmetric with respect to $\widehat{\Delta}$ and satisfies the intertwining property.    
Then for any strictly positive density operator $D\in \cM$, we have that 
\begin{equation}\label{eq:entropy1}
\begin{aligned}
 & H(\Phi_t^*(D)\|D_\Delta) -H(\bE_{\cN}(D)\overline{F_\Phi} \|D_{\Delta}) \\
 \leq &  e^{-2\beta t}\left( H(D\| D_{\Delta})  - H(\bE_{\cN}(D)\overline{F_\Phi} \|D_{\Delta})\right). 
\end{aligned}
\end{equation}
Furthermore, we have that 
\begin{align}\label{eq:logsob}
H(D\|D_{\Delta})- H(\bE_{\cN}(D)\overline{F_\Phi} \|D_{\Delta}) \leq \frac{1}{2\beta} \tau(\cL^*(D)( \log D-\log D_{\Delta})).    
\end{align}
This is called the generalized logarithmic Sobolev inequality.
\end{theorem}

\begin{proof}
Suppose that $X(s)$ is the solution of $\displaystyle \frac{d}{ds} D_s=\frac{1}{2}\Div \K_D X(s)$ minimizing $\|X(s)\|_{D, \widehat{\Delta}}$, where $\{D_s\}_s$ is a differentiable path in $\cM_+(D)$ such that $D_0=D$. 
We have that 
\begin{align*}
    \frac{d}{ds} \Phi_t^*(D_s) 
    = &\frac{1}{2} \Phi_t^*\Div \K_D X(s)
    = \frac{1}{2}e^{-\beta t}\Div \widetilde{\Phi_t}^*(\K_D X(s) )\\
    = & \frac{1}{2}e^{-\beta t} \Div \K_{\Phi_t^*(D)} \left(\K_{\Phi_t^*(D)}^{-1} \widetilde{\Phi_t}^*(\K_D X(s))\right).
\end{align*}
Now we have that 
\begin{align*}
 \left\|\frac{d}{ds}\Phi_t^*(D_s)\right\|_{g_\cL} ^2
\leq & e^{-2\beta t} \|\K_{\Phi_t^*(D)}^{-1}\widetilde{\Phi_t}^* (\K_D X(s))\|_{\Phi_t^*(D), \widehat{\Delta}}^2 \\
\leq & e^{-2\beta t} \left\langle \K_{\Phi_t^*(D)}^{-1}\widetilde{\Phi_t}^* (\K_D X(s)), \widetilde{\Phi_t}^* \K_D X(s) \right\rangle \\
\leq & e^{-2\beta t}\left\langle \K_D^{-1}  (\K_D  X(s) , ( \K_D X(s) \right\rangle   \\
=& e^{-2\beta t}\left\langle  \K_D  X(s) ,  X(s) \right\rangle   \\
= & e^{-2\beta t} \left\|\frac{d}{ds} D_s\right\|_{g_\cL}^2.
\end{align*}
This implies that $\displaystyle \lim_{t\to \infty}  \left\|\frac{d}{ds}\Phi_t^*(D_s)\right\|_{g_\cL}=0$ and 
\begin{align*}
   \frac{d}{dt}  \left\|\frac{d}{ds}\Phi_t^*(D_s)\right\|_{g_{\cL}} ^2  
 \leq -2\beta \left\|\frac{d}{ds}\Phi_t^*( D_s)\right\|_{g_{\cL}}^2.
\end{align*}
By Corollary \ref{cor:diffgrad}, we have that 
\begin{align*}
\|\grad_{g, \Phi_t^*(D_s)} f\|_{\Phi_t^*(D_s), \widehat{\Delta}}^2   \leq e^{-2\beta t} \|\grad_{g, D_s} f\|_{D_s, \widehat{\Delta}}^2. 
\end{align*}
Hence 
\begin{align*}
  \left. \frac{d}{dt}  \|\grad_{g, \Phi_t^*(D_s)} f\|_{D, \widehat{\Delta}}^2\right|_{t=0^+} \leq -2\beta \|\grad_{g, D_s} f\|_{D, \widehat{\Delta}}^2.
\end{align*}
Note that 
\begin{align*}
    \left.\frac{d}{dt} f(\Phi_t^*(D))\right|_{t=0}
    =& -\tau\left(\frac{df}{d D}\cL^*(D)\right)\\
    =& - \frac{1}{2}\tau\left(\frac{df}{d D}\Div\K_D \left( \nabla \frac{df}{dD}\right)\right)\\
    =& -\frac{1}{2}\left\langle \K_D \left( \nabla \frac{df}{dD}\right), \nabla\frac{df}{d D}\right\rangle \\
    =& -\frac{1}{2}\|\grad_{g, D} f\|_{D, \widehat{\Delta}}^2.
\end{align*}
We see that 
\begin{align}\label{eq:energe}
2 \frac{d}{dt} H(\Phi_t^*(D_s)\|D_{\Delta}) =-\left\|\grad_{g, \Phi_t^*(D_s)} f\right\|_{\Phi_t^*(D_s), \widehat{\Delta}}^2.
\end{align}
Hence
\begin{align*}
    2\frac{d}{dt} H(\Phi_t^*(D_s)\|D_{\Delta})= & -\left\|\frac{d}{ds}\Phi_t^*(D_s)\right\|_{g_\cL} ^2\\
    =& \int_t^\infty \frac{d}{dr}\left\|\frac{d}{ds}\Phi_r^*(D_s)\right\|_{g_\cL} ^2 dr\\
    \leq & -2\beta \int_t^\infty \|\grad_{g, \Phi^*_r(D_s)} f\|_{\Phi_r^*(D_s), \widehat{\Delta}}^2 dr\\
    = & 4\beta H(\bE_\Phi^*(D_s)\|D_{\Delta}) -4\beta  H(\Phi_t^*(D_s)\|D_{\Delta}),
\end{align*}
i.e.
\begin{align}\label{eq:entropy0}
    \frac{d}{dt} H(\Phi_t^*(D_s)\|D_{\Delta}) \leq 2\beta H(\bE_\Phi^*(D_s)\|D_{\Delta}) -2\beta H(\Phi_t^*(D_s)\| D_{\Delta}).
\end{align}
By integrating Equation \eqref{eq:entropy0} with respect to $t$, we obtain that
\begin{align*}
 H(\Phi_t^*(D_s)\|D_\Delta)-H(\bE_\Phi^*(D_s)\|D_{\Delta}) \leq e^{-2\beta t} (H(D_s\| D_{\Delta}) - H(\bE_\Phi^*(D_s)\|D_{\Delta})).  
\end{align*}
By Equation \eqref{eq:energe}, we obtain that 
\begin{align*}
 H(D_s\| D_{\Delta}) - H(\bE_\Phi^*(D_s)\|D_{\Delta})\leq \frac{1}{4\beta} \|\grad_{g, D} f\|_{D, \widehat{\Delta}}^2. 
\end{align*}
Note that 
\begin{align*}
   H(\bE_\Phi^*(D_s)\|D_{\Delta}) =H(\bE_{\cN}(D_s)\overline{F_\Phi} \|D_{\Delta}).
\end{align*}
We see the first equality of the theorem is true.
By Lemma \ref{lem:vecterfieldform}, we have that 
\begin{align*}
    \|\grad_{g, D} f\|_{D, \widehat{\Delta}}^2=2\tau(\cL^*(D)(\log D-\log D_{\Delta})).
\end{align*}
Therefore, the second inequality of the theorem is true.

\end{proof}

\begin{remark}
If $\widehat{\Delta}=  \vcenter{\hbox{\begin{tikzpicture}[scale=1.2]
        \draw [blue] (0, -0.5)--(0, 0.5) (0.5, -0.5)--(0.5, 0.5);
        \draw [fill=white] (-0.2, -0.2) rectangle (0.2, 0.2);
        \node at (0, 0) {\tiny $\overline{\Delta}$};
        \begin{scope}[shift={(0.5, 0)}]
         \draw [fill=white] (-0.2, -0.2) rectangle (0.2, 0.2);
        \node at (0, 0) {\tiny $\Delta^{-1}$};   
        \end{scope}
    \end{tikzpicture}}}$, then by Remarks \ref{rem:gnslimit} and \ref{rem:hidden}, we have that 
\begin{align*}
    H(\bE_{\cN}(D)\overline{F_\Phi} \|D_{\Delta})
    =&  H(\bE_{\cN}(D)\Delta\|\Delta)\\
    =& \tau\left(\bE_{\cN}(D)\Delta \log \bE_{\cN}(D)\Delta - \bE_{\cN}(D)\Delta \log \Delta\right)\\
    =& \tau (\bE_{\cN}(D)\Delta \log \bE_{\cN}(D)).
\end{align*}
If $\cN=\bC$, we see that $ H(\bE_{\cN}(D)\overline{F_\Phi} \|D_{\Delta})=0$ for density operator $D\in \cM$.
\end{remark}

The celebrated Otto-Villani theorem \cite{OttoVillani2000} establishes the connection between log-Sobolev inequality and Talagrand inequality in the context of Riemannian geometry. 
In \cite{CarMaa17}, Talagrand inequality \cite{Tal96} was obtained from modified log-Sobolev inequality of quantum Markov semigroups. 
In the following, we obtain a bimodule version of Talagrand inequality.
\begin{theorem}[Bimodule Talagrand Inequality]\label{thm:talagrand}
Suppose that $\{\Phi_t\}_{t\geq 0}$ is a bimodule quantum Markov semigroup bimodule GNS symmetry with respect to $\widehat{\Delta}$ and relatively ergodic. 
Suppose that Equation \eqref{eq:logsob} holds. 
We have that 
\begin{align*}
    d(D, D_\Delta) \leq 2\sqrt{ \frac{ H(D\| D_\Delta)- H(\bE_{\cN}(D)\overline{F_\Phi} \|D_{\Delta})}{\beta}}.
\end{align*}
\end{theorem}
\begin{proof}
Suppose that $D\in \cM$ is a density operator.
Let $\Delta_D=\bE_{\cN}(D)\overline{F_\Phi}$.
Then the distance between $D$ and $\Delta_D$ is described by the following equation:
\begin{align*}
d(D, \Delta_D)= & \int_0^\infty \left\|\frac{d}{dt} \Phi_t^*(D)\right\|_{g_\cL} dt .
\end{align*}
Note that $\displaystyle \left\|\frac{d}{dt} \Phi_t^*(D)\right\|_{g_\cL}^2 =-2\frac{d}{dt} H(\Phi_t^*(D)\| D_{\Delta})$.
By the Cauchy-Schwarz inequality, for any $0< t_1<  t_2< \infty$, we have that 
\begin{align*}
    \int_{t_1}^{t_2} \left\|\frac{d}{dt} \Phi_t^*(D)\right\|_{g_\cL} dt
    \leq \sqrt{2} \sqrt{t_2-t_1} \sqrt{H(\Phi_{t_1}^*(D)\| D_{\Delta})-H(\Phi_{t_2}^*(D)\|D_\Delta)}
\end{align*}
Fix $\epsilon>0$, for each $k\in \bN$, we take $t_k\in \bR$ such that 
\begin{align*}
   &  H(\Phi_{t_k}^*(D)\| D_{\Delta})-  H(\bE_{\cN}(D)\overline{F_\Phi} \|D_{\Delta})\\
    =&  e^{-k \epsilon} (H(D\| D_{\Delta})-  H(\bE_{\cN}(D)\overline{F_\Phi} \|D_{\Delta})).
\end{align*}
By Equation \eqref{eq:entropy1}, we have that $\displaystyle t_k-t_{k-1}\leq \frac{\epsilon}{2\beta}$.
By Equation \eqref{eq:entropy1} again, we have that
\begin{equation}\label{eq:taga}
\begin{aligned}
   &  \int_{t_{k-1}}^{t_k} \left\|\frac{d}{dt} \Phi_t^*(D)\right\|_{g_{\cL}}dt \\
     \leq & \sqrt{2} \sqrt{t_2-t_1} \sqrt{H(\Phi_{t_1}^*(D)\| D_{\Delta})-H(\Phi_{t_2}^*(D)\|D_\Delta)} \\
    \leq & \sqrt{2} \sqrt{\frac{\epsilon}{2\beta} (e^{-(k-1)\epsilon}-e^{-k\epsilon}) (H(D\| D_{\Delta})-  H(\bE_{\cN}(D)\overline{F_\Phi} \|D_{\Delta}))}\\
    =&\sqrt{2} e^{-k\epsilon /2}\sqrt{\epsilon (e^\epsilon-1)} \sqrt{\frac{H(D\| D_{\Delta})-  H(\bE_{\cN}(D)\overline{F_\Phi} \|D_{\Delta})}{2\beta}} .
\end{aligned}
\end{equation}
Note that 
\begin{align*}
 \lim_{\epsilon\to 0} \sum_{k=1}^\infty e^{-k\epsilon /2}\sqrt{\epsilon (e^\epsilon-1)} 
=& \lim_{\epsilon\to 0} \frac{e^{-\epsilon/2}\sqrt{\epsilon (e^{\epsilon}-1)}}{1-e^{-\epsilon/2}}\\
=& \lim_{\epsilon\to 0} \frac{ \sqrt{\epsilon (e^{\epsilon}-1)}}{e^{\epsilon/2}-1}\\
=&\lim_{\epsilon\to 0}  \sqrt{\frac{\epsilon (e^{\epsilon/2}+1)}{e^{\epsilon/2}-1}}\\
=& \sqrt{2}\lim_{\epsilon\to 0}  \sqrt{\frac{\epsilon}{e^{\epsilon/2}-1}}=2.
\end{align*}
By taking summation of Equation \eqref{eq:taga} with respect to $k$, we see that the theorem is true.

\end{proof}

\subsection{Fermion Algebras}

In this section, we establish a pictorial proof of the intertwining property for the semigroup described in \cite{CarMaa17}, working entirely within this framework.

We begin by describing the extension of a bimodule map.
Suppose that $\Psi: \cM_1\to\cM_1$ is a bimodule completely positive map for the inclusion $\cN \subset \cM_1$ which is a $\lambda^2$-extension.
Recall that the associated Jones projections $\widetilde{e}_1=\lambda^{-1} e_2e_1e_3e_2$ and $\widetilde{e}_2=\lambda^{-1} e_4e_3e_5e_4$.
We have that 
\begin{align}\label{eq:jones1}
    e_1e_2 \widetilde{e}_2\widetilde{e}_1=\lambda^2 e_1e_4e_3e_2.
\end{align}
Then the Fourier multiplier $\widehat{\Psi}$ is in $\cM_1'\cap \cM_5$ and 
\begin{align*}
    \Psi(x)=& \lambda^{-5}\bE_{\cM_1}(\widetilde{e}_2\widetilde{e}_1 \widehat{\Psi} x \widetilde{e}_1 \widetilde{e}_2)\\
    =& \lambda^{-9} \bE_{\cM_1}(e_4e_5e_3e_4e_2e_3e_1e_2 \widehat{\Psi}x e_2e_1e_3e_2e_4e_3e_5e_4).
\end{align*}
and 
\begin{align*}
      \Psi(x)\widetilde{e}_2=\lambda^{-3} \widetilde{e}_2\widetilde{e}_1 \widehat{\Psi} x \widetilde{e}_1 \widetilde{e}_2
\end{align*}

\begin{proposition}
Suppose that $\Phi: \cM\to \cM$ and $\Psi: \cM_1 \to \cM_1$ are bimodule maps.
Then $\Psi|_{\cM}=\Phi$ if and only if 
\begin{align}\label{eq:extension1}
   \lambda^{-7/2} \bE_{\cM_4}(e_5e_4e_3e_2\widehat{\Psi}e_2e_3e_4e_5)=\widehat{\Phi}.
\end{align}
Pictorially, we have that 
\begin{align*}
    \vcenter{\hbox{
    \begin{tikzpicture}
       \draw [blue] (-0.45, -0.3).. controls+(0, -0.3) and +(0, -0.3).. (-0.75, -0.3)--(-0.75, 0.3) .. controls +(0, 0.3) and +(0, 0.3)..(-0.45, 0.3) (-0.15, -0.6)--(-0.15, 0.6) (0.15, -0.6)--(0.15, 0.6) (0.45, -0.6)--(0.45, 0.6);
       \draw[fill=white] (-0.6, -0.3) rectangle (0.6, 0.3);
       \node at (0, 0) {\tiny $\widehat{\Psi}$};
    \end{tikzpicture}
    }}
    =   \lambda^{-1} \vcenter{\hbox{
    \begin{tikzpicture}
       \draw [blue]   (-0.15, -0.6)--(-0.15, 0.6) (0.15, -0.6)--(0.15, 0.6) (0.45, -0.6)--(0.45, 0.6);
       \draw[fill=white] (-0.35, -0.3) rectangle (0.35, 0.3);
       \node at (0, 0) {\tiny $\widehat{\Phi}$};
    \end{tikzpicture}
    }}.
\end{align*}
\end{proposition}
\begin{proof}
Suppose that $\Psi|_{\cM}=\Phi$.
For any $x\in \cM$, we have that 
\begin{align*}
   \lambda^{-3} \widetilde{e}_2\widetilde{e}_1 \widehat{\Psi} x \widetilde{e}_1 \widetilde{e}_2
   = \Phi(x)\widetilde{e}_2 
   = \lambda^{-5/2} \bE_{\cM}(e_2e_1 \widehat{\Phi} x e_1e_2) \widetilde{e}_2. 
\end{align*}
Multiplying $e_1e_2$ from the left hand side, we obtain that 
\begin{align*}
   \lambda^{-3} e_1e_2\widetilde{e}_2\widetilde{e}_1 \widehat{\Psi} x \widetilde{e}_1 \widetilde{e}_2
    =& \lambda^{-3/2} e_1e_2e_1 \widehat{\Phi} x e_1e_2 \widetilde{e}_2
\end{align*}
By Equation \eqref{eq:jones1}, we have that 
\begin{align*}
 \lambda^{-1} e_1 x e_4e_3e_2  \widehat{\Psi} \widetilde{e}_1 \widetilde{e}_2
    =& \lambda^{-1/2} e_1 x \widehat{\Phi} e_1e_2 \widetilde{e}_2
\end{align*}
Applying the Pimsner-Popa basis, we obtain that 
\begin{align*}
  e_4e_3e_2  \widehat{\Psi} \widetilde{e}_1 \widetilde{e}_2
    =& \lambda^{1/2}  \widehat{\Phi} e_1e_2 \widetilde{e}_2
\end{align*}
Multiplying $\widetilde{e}_1$ from the right hand side, we see that 
\begin{align*}
      e_4e_3e_2  \widehat{\Psi} \widetilde{e}_1
    =& \lambda^{1/2}  \widehat{\Phi} e_1e_4e_3e_2.
\end{align*}
Expanding $\widetilde{e}_1$, we have that 
\begin{align*}
    e_4e_3e_2  \widehat{\Psi} e_2e_1e_3e_2
    =& \lambda^{3/2}  \widehat{\Phi} e_1e_4e_3e_2
\end{align*}
Multiplying $e_3$ from the right hand side, we see that 
\begin{align*}
      e_4e_3e_2  \widehat{\Psi} e_2e_1e_3 
    =& \lambda^{3/2}  \widehat{\Phi} e_1e_4e_3.
\end{align*}
By taking the conditional expectation $\bE_{\cM'}$, we have that 
\begin{align*}
          e_4e_3e_2  \widehat{\Psi} e_2 e_3
    =& \lambda^{1/2} \widehat{\Phi} e_4e_3.
\end{align*}
Multiplying $e_4$ from the right hand side, we have that 
\begin{align}\label{eq:extension2}
e_4e_3e_2  \widehat{\Psi} e_2 e_3e_4
    =\lambda^{3/2} \widehat{\Phi} e_4.
\end{align}
Multiplying $e_5$ from the both sides, we have that 
\begin{align*}
  e_5 e_4e_3e_2  \widehat{\Psi} e_2 e_3e_4 e_5
    =\lambda^{5/2} \widehat{\Phi} e_5 . 
\end{align*}
By taking the conditional expectation $\bE_{\cM_4}$, we have that Equation \eqref{eq:extension1} is true.

Suppose that Equation \eqref{eq:extension1} is true.
We see that Equation \eqref{eq:extension2} is true.
By the previous computation, we see that $\Psi|_{\cM}=\Phi$.
\end{proof}

Suppose that $\{\Psi_t\}_{t\geq 0}$ is a bimodule quantum Markov semigroup such that $\Psi_t|_{\cM}=\Phi_t$ and 
\begin{align*}
   \Phi^*_t \Div = e^{-\widetilde{\beta} t}\Div \Psi_t^*.
\end{align*}
Let $\cJ$ be the generator of the semigroup $\{\Psi_t\}_{t \geq 0}$ and $\widehat{\cJ}|_{\cM}=\widehat{\cL}$ .
The intertwining property is
\begin{align*}
    \partial_k \cL - \cJ \partial_k =\widetilde{\beta} \partial_k. 
\end{align*}
Pictorially, we have that 
\begin{align}\label{eq:inter}
 \vcenter{\hbox{
\begin{tikzpicture}
\draw [blue] (-0.15, -0.8)--(-0.15, 0.8);
 \draw [blue] (0.15, -0.8)--(0.15, 0.3) .. controls +(0, 0.35) and +(0, 0.35).. (0.6, 0.3);
 \draw [fill=white] (-0.3, -0.3) rectangle (0.3, 0.3);
 \node at (0, 0) {\tiny $\widehat{\cL}$};
 \begin{scope}[shift={(0.75, 0)}]
  \draw [fill=white] (-0.3, -0.3) rectangle (0.3, 0.3);
 \node at (0, 0) {\tiny $E_k$};
 \draw [blue] (0.15, 0.3)--(0.15, 0.8) (-0.15, -0.3)--(-0.15, -0.8);
  \draw [blue] (0.15, -0.3) .. controls +(0, -0.35) and +(0, -0.35).. (0.6, -0.3)--(0.6, 0.8);
 \end{scope}
\end{tikzpicture}
 }}
 -
\vcenter{\hbox{
    \begin{tikzpicture}
    \draw [blue] (-0.15, -0.8)--(-0.15, 0.8);
 \draw [blue] (0.15, 0.8)--(0.15, -0.3) .. controls +(0, -0.35) and +(0, -0.35).. (0.6, -0.3);
 \draw [fill=white] (-0.3, -0.3) rectangle (0.3, 0.3);
 \node at (0, 0) {\tiny $\widehat{\cL}$};
 \begin{scope}[shift={(0.75, 0)}]
  \draw [fill=white] (-0.3, -0.3) rectangle (0.3, 0.3);
 \node at (0, 0) {\tiny $\overline{E_k}$};
 \draw [blue] (0.15, -0.3)--(0.15, -0.8) (-0.15, 0.3)--(-0.15, 0.8);
  \draw [blue] (0.15, 0.3) .. controls +(0, 0.35) and +(0, 0.35).. (0.6, 0.3)--(0.6, -0.8);
 \end{scope}
    \end{tikzpicture}
    }}
    -
       \lambda\vcenter{\hbox{
    \begin{tikzpicture}
       \draw [blue] (-0.75, -0.3)--(-0.75, 0.6) (-0.45, -0.3)--(-0.45, 0.6) (-0.15, -0.8)--(-0.15, 0.6) (0.15, -0.8)--(0.15, 1.6) (0.45, -0.8)--(0.45, 1.6);
       \draw[fill=white] (-0.6, -0.3) rectangle (0.6, 0.3);
       \node at (0, 0) {\tiny $\widehat{\cJ}$};
       \begin{scope}[shift={(-0.6, 0.8)}]
          \draw[fill=white] (-0.3, -0.3) rectangle (0.3, 0.3); 
           \node at (0, 0) {\tiny $E_k$};
        \draw[blue] (0.15, 0.3).. controls+(0, 0.3) and +(0, 0.3).. (0.45, 0.3)--(0.45, -0.3);
        \draw[blue] (-0.15, 0.3)--(-0.15, 0.8);
        \draw [blue] (-0.15, -1.1) .. controls +(0, -0.25) and +(0, -0.25).. (0.15, -1.1);
       \end{scope}
    \end{tikzpicture}
    }}
    +
    \lambda\vcenter{\hbox{
    \begin{tikzpicture}
       \draw [blue]   (-0.15, 0.8)--(-0.15, 0.3) (0.15, -0.8)--(0.15, 0.8) (0.45, -0.8)--(0.45, 0.8);
        \draw[fill=white] (-0.6, -0.3) rectangle (0.6, 0.3);
       \node at (0, 0) {\tiny $\widehat{\cJ}$};
       \begin{scope}[shift={(-1.2, 0)}]
       \draw [blue] (0.15, 0.3) .. controls +(0, 0.6) and +(0, 0.6) .. (-0.75, 0.3)--(-0.75, -0.8);
          \draw[fill=white] (-0.3, -0.3) rectangle (0.3, 0.3); 
           \node at (0, 0) {\tiny $E_k$};
        \draw[blue] (0.15, -0.3).. controls+(0, -0.25) and +(0, -0.25).. (0.45, -0.3)--(0.45, 0.3).. controls +(0, 0.3) and +(0, 0.3) .. (0.75, 0.3);
        \draw[blue] (-0.15, 0.3).. controls+(0, 0.3) and +(0, 0.3).. (-0.45, 0.3)--(-0.45, -0.3).. controls +(0, -0.6) and +(0, -0.6) .. (1.05, -0.3);
        \draw [blue] (-0.15, -0.3).. controls +(0, -0.4) and +(0, -0.4) .. (0.75, -0.3);
       \end{scope}
    \end{tikzpicture}
    }}
 =
\widetilde{\beta} \vcenter{\hbox{
    \begin{tikzpicture}
 \draw [blue] (0.15, 0.8)--(0.15, -0.3) .. controls +(0, -0.35) and +(0, -0.35).. (0.6, -0.3)--(0.6, 0.8);
  \draw [blue] (-0.15, 0.8)--(-0.15, -0.8);
 \draw [fill=white] (-0.3, -0.3) rectangle (0.3, 0.3);
 \node at (0, 0) {\tiny $E_k$};
  \draw [blue] (-0.7, -0.8) .. controls +(0, 0.25) and +(0, 0.25).. (-0.4, -0.8);
    \end{tikzpicture}
    }}  
-
\widetilde{\beta} \vcenter{\hbox{
    \begin{tikzpicture}
 \draw [blue] (0.15, -0.8)--(0.15, 0.3) .. controls +(0, 0.35) and +(0, 0.35).. (0.6, 0.3)--(0.6, -0.8);
  \draw [blue] (-0.15, 0.8)--(-0.15, -0.8);
 \draw [fill=white] (-0.3, -0.3) rectangle (0.3, 0.3);
 \node at (0, 0) {\tiny $\overline{E_k}$};
  \draw [blue] (-0.7, 0.8) .. controls +(0, -0.25) and +(0, -0.25).. (-0.4, 0.8);
    \end{tikzpicture}
    }} . 
\end{align}

Suppose that $Q_1, \ldots, Q_m, P_1, \ldots, P_m$ be the generators with
\begin{align*}
Q_jQ_k+Q_kQ_j=P_jP_k+P_kP_j=2\delta_{j,k}, \quad Q_jP_k+P_kQ_j=0, \quad 1\leq j, k\leq m.
\end{align*}
Let $\displaystyle w=i^m\prod_{j=1}^m Q_jP_j$ and $\displaystyle v_j=\frac{1}{\sqrt{2}}w(Q_j+iP_j)$.
Then $wv_jw=-v_j$.

Suppose that 
\begin{align*}
    \widehat{\cL}_0=\frac{1}{2}\sum_{j=1}^m e^{\beta a_j/2} \vcenter{\hbox{\begin{tikzpicture}[scale=0.65]
    \begin{scope}[shift={(0,1.5)}]
    \draw [blue] (-0.5, 0.8)--(-0.5, 0) .. controls +(0, -0.6) and +(0,-0.6).. (0.5, 0)--(0.5, 0.8);    
\begin{scope}[shift={(0.5, 0.3)}]
\draw [fill=white] (-0.3, -0.3) rectangle (0.3, 0.3);
\node at (0, 0) {\tiny $v_j$};
\end{scope}
    \end{scope}
\draw [blue] (-0.5, -0.8)--(-0.5, 0) .. controls +(0, 0.6) and +(0,0.6).. (0.5, 0)--(0.5, -0.8);
\begin{scope}[shift={(0.5, -0.3)}]
\draw [fill=white] (-0.3, -0.3) rectangle (0.3, 0.3);
\node at (0, 0) {\tiny $v_j^*$};
\end{scope}
\end{tikzpicture}}}
+ e^{-\beta a_j/2} \vcenter{\hbox{\begin{tikzpicture}[scale=0.65]
    \begin{scope}[shift={(0,1.5)}]
    \draw [blue] (-0.5, 0.8)--(-0.5, 0) .. controls +(0, -0.6) and +(0,-0.6).. (0.5, 0)--(0.5, 0.8);    
\begin{scope}[shift={(0.5, 0.3)}]
\draw [fill=white] (-0.3, -0.3) rectangle (0.3, 0.3);
\node at (0, 0) {\tiny $v_j^*$};
\end{scope}
    \end{scope}
\draw [blue] (-0.5, -0.8)--(-0.5, 0) .. controls +(0, 0.6) and +(0,0.6).. (0.5, 0)--(0.5, -0.8);
\begin{scope}[shift={(0.5, -0.3)}]
\draw [fill=white] (-0.3, -0.3) rectangle (0.3, 0.3);
\node at (0, 0) {\tiny $v_j$};
\end{scope}
\end{tikzpicture}}},
\end{align*}
where $a_j\in \bR$ and $\beta>0$.
Let
\begin{align*}
    \mathbf{y}=1*\widehat{\cL}_0 =\frac{1}{2}\sum_{j=1}^m e^{\beta a_j/2} v_j^*v_j +e^{-\beta a_j/2} v_jv_j^*.
\end{align*}
The directional derivation is given by $
    E_k=\vcenter{\hbox{\begin{tikzpicture}[scale=0.65]
    \begin{scope}[shift={(0,1.5)}]
    \draw [blue] (-0.5, 0.8)--(-0.5, 0) .. controls +(0, -0.6) and +(0,-0.6).. (0.5, 0)--(0.5, 0.8);    
\begin{scope}[shift={(0.5, 0.3)}]
\draw [fill=white] (-0.4, -0.3) rectangle (0.4, 0.3);
\node at (0, 0) {\tiny $v_k$};
\end{scope}
    \end{scope}
\draw [blue] (-0.5, -0.8)--(-0.5, 0) .. controls +(0, 0.6) and +(0,0.6).. (0.5, 0)--(0.5, -0.8);
\begin{scope}[shift={(0.5, -0.3)}]
\draw [fill=white] (-0.4, -0.3) rectangle (0.4, 0.3);
\node at (0, 0) {\tiny $v_k^*$};
\end{scope}
\end{tikzpicture}}}$ and skew derivation is given by $wE_k$.

The extension $\cJ$ of $\cL$ is given by
\begin{align*}
\widehat{\cJ}_0=\frac{\lambda^{-1}}{2}\sum_{j=1}^m e^{\beta a_j/2} \vcenter{\hbox{\begin{tikzpicture}[scale=0.65]
    \begin{scope}[shift={(0,1.5)}]
    \draw [blue] (-0.5, 0.8)--(-0.5, 0) .. controls +(0, -0.6) and +(0,-0.6).. (0.5, 0)--(0.5, 0.8);  
    \draw [blue] (-1, 0.8)--(-1, 0) .. controls +(0, -0.8) and +(0,-0.8).. (1, 0)--(1, 0.8); 
\begin{scope}[shift={(0.5, 0.3)}]
\draw [fill=white] (-0.3, -0.3) rectangle (0.3, 0.3);
\node at (0, 0) {\tiny $v_j$};
\end{scope}
    \end{scope}
\draw [blue] (-0.5, -0.8)--(-0.5, 0) .. controls +(0, 0.6) and +(0,0.6).. (0.5, 0)--(0.5, -0.8);
\draw [blue] (-1, -0.8)--(-1, 0) .. controls +(0, 0.8) and +(0,0.8).. (1, 0)--(1, -0.8); 
\begin{scope}[shift={(0.5, -0.3)}]
\draw [fill=white] (-0.3, -0.3) rectangle (0.3, 0.3);
\node at (0, 0) {\tiny $v_j^*$};
\end{scope}
\end{tikzpicture}}}
+ e^{-\beta a_j/2} \vcenter{\hbox{\begin{tikzpicture}[scale=0.65]
    \begin{scope}[shift={(0,1.5)}]
    \draw [blue] (-0.5, 0.8)--(-0.5, 0) .. controls +(0, -0.6) and +(0,-0.6).. (0.5, 0)--(0.5, 0.8);   \draw [blue] (-1, 0.8)--(-1, 0) .. controls +(0, -0.8) and +(0,-0.8).. (1, 0)--(1, 0.8);  
\begin{scope}[shift={(0.5, 0.3)}]
\draw [fill=white] (-0.3, -0.3) rectangle (0.3, 0.3);
\node at (0, 0) {\tiny $v_j^*$};
\end{scope}
    \end{scope}
\draw [blue] (-0.5, -0.8)--(-0.5, 0) .. controls +(0, 0.6) and +(0,0.6).. (0.5, 0)--(0.5, -0.8);
\draw [blue] (-1, -0.8)--(-1, 0) .. controls +(0, 0.8) and +(0,0.8).. (1, 0)--(1, -0.8); 
\begin{scope}[shift={(0.5, -0.3)}]
\draw [fill=white] (-0.3, -0.3) rectangle (0.3, 0.3);
\node at (0, 0) {\tiny $v_j$};
\end{scope}
\end{tikzpicture}}}.
\end{align*}
The four items in the intertwining property is the following:
\begin{align*}
\frac{1}{2}\vcenter{\hbox{\begin{tikzpicture}[scale=0.65]
     \draw [blue] (1.4, 2.3)--(1.4, -0.8);
    \begin{scope}[shift={(0,1.5)}]
    \draw [blue] (-0.5, 0.8)--(-0.5, 0) .. controls +(0, -0.6) and +(0,-0.6).. (0.5, 0)--(0.5, 0.8); 
\begin{scope}[shift={(0.5, 0.3)}]
\draw [fill=white] (-0.5, -0.3) rectangle (0.5, 0.3);
\node at (0, 0) {\tiny $\mathbf{y} v_k$};
\end{scope}
    \end{scope}
\draw [blue] (-0.5, -0.8)--(-0.5, 0) .. controls +(0, 0.6) and +(0,0.6).. (0.5, 0)--(0.5, -0.8);
\begin{scope}[shift={(0.5, -0.3)}]
\draw [fill=white] (-0.3, -0.3) rectangle (0.3, 0.3);
\node at (0, 0) {\tiny $w$};
\end{scope}
\begin{scope}[shift={(1.4, 0.8)}]
\draw [fill=white] (-0.3, -0.3) rectangle (0.3, 0.3);
\node at (0, 0) {\tiny $\overline{v_k^*}$};  
\end{scope}
\end{tikzpicture}}}
+
\frac{1}{2}\vcenter{\hbox{\begin{tikzpicture}[scale=0.65]
     \draw [blue] (1.4, 2.3)--(1.4, -0.8);
    \begin{scope}[shift={(0,1.5)}]
    \draw [blue] (-0.5, 0.8)--(-0.5, 0) .. controls +(0, -0.6) and +(0,-0.6).. (0.5, 0)--(0.5, 0.8); 
\begin{scope}[shift={(0.5, 0.3)}]
\draw [fill=white] (-0.4, -0.3) rectangle (0.4, 0.3);
\node at (0, 0) {\tiny $v_k$};
\end{scope}
    \end{scope}
\draw [blue] (-0.5, -0.8)--(-0.5, 0) .. controls +(0, 0.6) and +(0,0.6).. (0.5, 0)--(0.5, -0.8);
\begin{scope}[shift={(0.5, -0.3)}]
\draw [fill=white] (-0.3, -0.3) rectangle (0.3, 0.3);
\node at (0, 0) {\tiny $w\mathbf{y}$};
\end{scope}
\begin{scope}[shift={(1.4, 0.8)}]
\draw [fill=white] (-0.3, -0.3) rectangle (0.3, 0.3);
\node at (0, 0) {\tiny $\overline{v_k^*}$};  
\end{scope}
\end{tikzpicture}}}
-\frac{1}{2}\sum_{j=1}^m e^{\beta a_j/2}\vcenter{\hbox{\begin{tikzpicture}[scale=0.65]
     \draw [blue] (1.4, 2.3)--(1.4, -0.8);
    \begin{scope}[shift={(0,1.5)}]
    \draw [blue] (-0.5, 0.8)--(-0.5, 0) .. controls +(0, -0.6) and +(0,-0.6).. (0.5, 0)--(0.5, 0.8); 
\begin{scope}[shift={(0.5, 0.3)}]
\draw [fill=white] (-0.55, -0.3) rectangle (0.55, 0.3);
\node at (0, 0) {\tiny $v_jv_k$};
\end{scope}
    \end{scope}
\draw [blue] (-0.5, -0.8)--(-0.5, 0) .. controls +(0, 0.6) and +(0,0.6).. (0.5, 0)--(0.5, -0.8);
\begin{scope}[shift={(0.5, -0.3)}]
\draw [fill=white] (-0.3, -0.3) rectangle (0.3, 0.3);
\node at (0, 0) {\tiny $wv_j^*$};
\end{scope}
\begin{scope}[shift={(1.4, 0.8)}]
\draw [fill=white] (-0.3, -0.3) rectangle (0.3, 0.3);
\node at (0, 0) {\tiny $\overline{v_k^*}$};  
\end{scope}
\end{tikzpicture}}}
-\frac{1}{2}\sum_{j=1}^m e^{-\beta a_j/2}\vcenter{\hbox{\begin{tikzpicture}[scale=0.65]
     \draw [blue] (1.4, 2.3)--(1.4, -0.8);
    \begin{scope}[shift={(0,1.5)}]
    \draw [blue] (-0.5, 0.8)--(-0.5, 0) .. controls +(0, -0.6) and +(0,-0.6).. (0.5, 0)--(0.5, 0.8); 
\begin{scope}[shift={(0.5, 0.3)}]
\draw [fill=white] (-0.55, -0.3) rectangle (0.55, 0.3);
\node at (0, 0) {\tiny $v_j^*v_k$};
\end{scope}
    \end{scope}
\draw [blue] (-0.5, -0.8)--(-0.5, 0) .. controls +(0, 0.6) and +(0,0.6).. (0.5, 0)--(0.5, -0.8);
\begin{scope}[shift={(0.5, -0.3)}]
\draw [fill=white] (-0.3, -0.3) rectangle (0.3, 0.3);
\node at (0, 0) {\tiny $wv_j$};
\end{scope}
\begin{scope}[shift={(1.4, 0.8)}]
\draw [fill=white] (-0.3, -0.3) rectangle (0.3, 0.3);
\node at (0, 0) {\tiny $\overline{v_k^*}$};  
\end{scope}
\end{tikzpicture}}}
\end{align*}

\begin{align*}
\frac{1}{2}\vcenter{\hbox{\begin{tikzpicture}[scale=0.65]
     \draw [blue] (1.4, 2.3)--(1.4, -0.8);
    \begin{scope}[shift={(0,1.5)}]
    \draw [blue] (-0.5, 0.8)--(-0.5, 0) .. controls +(0, -0.6) and +(0,-0.6).. (0.5, 0)--(0.5, 0.8); 
\begin{scope}[shift={(0.5, 0.3)}]
\draw [fill=white] (-0.5, -0.3) rectangle (0.5, 0.3);
\node at (0, 0) {\tiny $\mathbf{y}$};
\end{scope}
    \end{scope}
\draw [blue] (-0.5, -0.8)--(-0.5, 0) .. controls +(0, 0.6) and +(0,0.6).. (0.5, 0)--(0.5, -0.8);
\begin{scope}[shift={(0.5, -0.3)}]
\draw [fill=white] (-0.4, -0.3) rectangle (0.4, 0.3);
\node at (0, 0) {\tiny $w v_k $};
\end{scope}
\begin{scope}[shift={(1.4, 0.8)}]
\draw [fill=white] (-0.3, -0.3) rectangle (0.3, 0.3);
\node at (0, 0) {\tiny $\overline{v_k^*}$};  
\end{scope}
\end{tikzpicture}}}
+
\frac{1}{2}\vcenter{\hbox{\begin{tikzpicture}[scale=0.65]
     \draw [blue] (1.4, 2.3)--(1.4, -0.8);
    \begin{scope}[shift={(0,1.5)}]
    \draw [blue] (-0.5, 0.8)--(-0.5, 0) .. controls +(0, -0.6) and +(0,-0.6).. (0.5, 0)--(0.5, 0.8); 
    \end{scope}
\draw [blue] (-0.5, -0.8)--(-0.5, 0) .. controls +(0, 0.6) and +(0,0.6).. (0.5, 0)--(0.5, -0.8);
\begin{scope}[shift={(0.5, -0.3)}]
\draw [fill=white] (-0.55, -0.3) rectangle (0.55, 0.3);
\node at (0, 0) {\tiny $wv_k\mathbf{y}$};
\end{scope}
\begin{scope}[shift={(1.4, 0.8)}]
\draw [fill=white] (-0.3, -0.3) rectangle (0.3, 0.3);
\node at (0, 0) {\tiny $\overline{v_k^*}$};  
\end{scope}
\end{tikzpicture}}}
-\frac{1}{2}\sum_{j=1}^m e^{\beta a_j/2}\vcenter{\hbox{\begin{tikzpicture}[scale=0.65]
     \draw [blue] (1.4, 2.3)--(1.4, -0.8);
    \begin{scope}[shift={(0,1.5)}]
    \draw [blue] (-0.5, 0.8)--(-0.5, 0) .. controls +(0, -0.6) and +(0,-0.6).. (0.5, 0)--(0.5, 0.8); 
\begin{scope}[shift={(0.5, 0.3)}]
\draw [fill=white] (-0.3, -0.3) rectangle (0.3, 0.3);
\node at (0, 0) {\tiny $v_j$};
\end{scope}
    \end{scope}
\draw [blue] (-0.5, -0.8)--(-0.5, 0) .. controls +(0, 0.6) and +(0,0.6).. (0.5, 0)--(0.5, -0.8);
\begin{scope}[shift={(0.5, -0.3)}]
\draw [fill=white] (-0.55, -0.3) rectangle (0.55, 0.3);
\node at (0, 0) {\tiny $wv_kv_j^*$};
\end{scope}
\begin{scope}[shift={(1.4, 0.8)}]
\draw [fill=white] (-0.3, -0.3) rectangle (0.3, 0.3);
\node at (0, 0) {\tiny $\overline{v_k^*}$};  
\end{scope}
\end{tikzpicture}}}
-\frac{1}{2}\sum_{j=1}^m e^{-\beta a_j/2}\vcenter{\hbox{\begin{tikzpicture}[scale=0.65]
     \draw [blue] (1.4, 2.3)--(1.4, -0.8);
    \begin{scope}[shift={(0,1.5)}]
    \draw [blue] (-0.5, 0.8)--(-0.5, 0) .. controls +(0, -0.6) and +(0,-0.6).. (0.5, 0)--(0.5, 0.8); 
\begin{scope}[shift={(0.5, 0.3)}]
\draw [fill=white] (-0.3, -0.3) rectangle (0.3, 0.3);
\node at (0, 0) {\tiny $v_j^*$};
\end{scope}
    \end{scope}
\draw [blue] (-0.5, -0.8)--(-0.5, 0) .. controls +(0, 0.6) and +(0,0.6).. (0.5, 0)--(0.5, -0.8);
\begin{scope}[shift={(0.5, -0.3)}]
\draw [fill=white] (-0.55, -0.3) rectangle (0.55, 0.3);
\node at (0, 0) {\tiny $wv_k v_j$};
\end{scope}
\begin{scope}[shift={(1.4, 0.8)}]
\draw [fill=white] (-0.3, -0.3) rectangle (0.3, 0.3);
\node at (0, 0) {\tiny $\overline{v_k^*}$};  
\end{scope}
\end{tikzpicture}}}
\end{align*}

\begin{align*}
 \frac{1}{2}\vcenter{\hbox{\begin{tikzpicture}[scale=0.65]
     \draw [blue] (1.4, 2.3)--(1.4, -0.8);
    \begin{scope}[shift={(0,1.5)}]
    \draw [blue] (-0.5, 0.8)--(-0.5, 0) .. controls +(0, -0.6) and +(0,-0.6).. (0.5, 0)--(0.5, 0.8); 
\begin{scope}[shift={(0.5, 0.3)}]
\draw [fill=white] (-0.5, -0.3) rectangle (0.5, 0.3);
\node at (0, 0) {\tiny $v_k \mathbf{y} $};
\end{scope}
    \end{scope}
\draw [blue] (-0.5, -0.8)--(-0.5, 0) .. controls +(0, 0.6) and +(0,0.6).. (0.5, 0)--(0.5, -0.8);
\begin{scope}[shift={(0.5, -0.3)}]
\draw [fill=white] (-0.3, -0.3) rectangle (0.3, 0.3);
\node at (0, 0) {\tiny $w$};
\end{scope}
\begin{scope}[shift={(1.4, 0.8)}]
\draw [fill=white] (-0.3, -0.3) rectangle (0.3, 0.3);
\node at (0, 0) {\tiny $\overline{v_k^*}$};  
\end{scope}
\end{tikzpicture}}}
+
\frac{1}{2}\vcenter{\hbox{\begin{tikzpicture}[scale=0.65]
     \draw [blue] (1.4, 2.3)--(1.4, -0.8);
    \begin{scope}[shift={(0,1.5)}]
    \draw [blue] (-0.5, 0.8)--(-0.5, 0) .. controls +(0, -0.6) and +(0,-0.6).. (0.5, 0)--(0.5, 0.8); 
\begin{scope}[shift={(0.5, 0.3)}]
\draw [fill=white] (-0.4, -0.3) rectangle (0.4, 0.3);
\node at (0, 0) {\tiny $v_k$};
\end{scope}
    \end{scope}
\draw [blue] (-0.5, -0.8)--(-0.5, 0) .. controls +(0, 0.6) and +(0,0.6).. (0.5, 0)--(0.5, -0.8);
\begin{scope}[shift={(0.5, -0.3)}]
\draw [fill=white] (-0.3, -0.3) rectangle (0.3, 0.3);
\node at (0, 0) {\tiny $\mathbf{y}w$};
\end{scope}
\begin{scope}[shift={(1.4, 0.8)}]
\draw [fill=white] (-0.3, -0.3) rectangle (0.3, 0.3);
\node at (0, 0) {\tiny $\overline{v_k^*}$};  
\end{scope}
\end{tikzpicture}}}
-\frac{1}{2}\sum_{j=1}^m e^{\beta a_j/2}\vcenter{\hbox{\begin{tikzpicture}[scale=0.65]
     \draw [blue] (1.4, 2.3)--(1.4, -0.8);
    \begin{scope}[shift={(0,1.5)}]
    \draw [blue] (-0.5, 0.8)--(-0.5, 0) .. controls +(0, -0.6) and +(0,-0.6).. (0.5, 0)--(0.5, 0.8); 
\begin{scope}[shift={(0.5, 0.3)}]
\draw [fill=white] (-0.55, -0.3) rectangle (0.55, 0.3);
\node at (0, 0) {\tiny $v_kv_j$};
\end{scope}
    \end{scope}
\draw [blue] (-0.5, -0.8)--(-0.5, 0) .. controls +(0, 0.6) and +(0,0.6).. (0.5, 0)--(0.5, -0.8);
\begin{scope}[shift={(0.5, -0.3)}]
\draw [fill=white] (-0.3, -0.3) rectangle (0.3, 0.3);
\node at (0, 0) {\tiny $v_j^*w$};
\end{scope}
\begin{scope}[shift={(1.4, 0.8)}]
\draw [fill=white] (-0.3, -0.3) rectangle (0.3, 0.3);
\node at (0, 0) {\tiny $\overline{v_k^*}$};  
\end{scope}
\end{tikzpicture}}}
-\frac{1}{2}\sum_{j=1}^m e^{-\beta a_j/2}\vcenter{\hbox{\begin{tikzpicture}[scale=0.65]
     \draw [blue] (1.4, 2.3)--(1.4, -0.8);
    \begin{scope}[shift={(0,1.5)}]
    \draw [blue] (-0.5, 0.8)--(-0.5, 0) .. controls +(0, -0.6) and +(0,-0.6).. (0.5, 0)--(0.5, 0.8); 
\begin{scope}[shift={(0.5, 0.3)}]
\draw [fill=white] (-0.55, -0.3) rectangle (0.55, 0.3);
\node at (0, 0) {\tiny $v_kv_j^*$};
\end{scope}
    \end{scope}
\draw [blue] (-0.5, -0.8)--(-0.5, 0) .. controls +(0, 0.6) and +(0,0.6).. (0.5, 0)--(0.5, -0.8);
\begin{scope}[shift={(0.5, -0.3)}]
\draw [fill=white] (-0.3, -0.3) rectangle (0.3, 0.3);
\node at (0, 0) {\tiny $v_jw$};
\end{scope}
\begin{scope}[shift={(1.4, 0.8)}]
\draw [fill=white] (-0.3, -0.3) rectangle (0.3, 0.3);
\node at (0, 0) {\tiny $\overline{v_k^*}$};  
\end{scope}
\end{tikzpicture}}}   
\end{align*}

\begin{align*}
\frac{1}{2}\vcenter{\hbox{\begin{tikzpicture}[scale=0.65]
     \draw [blue] (1.4, 2.3)--(1.4, -0.8);
    \begin{scope}[shift={(0,1.5)}]
    \draw [blue] (-0.5, 0.8)--(-0.5, 0) .. controls +(0, -0.6) and +(0,-0.6).. (0.5, 0)--(0.5, 0.8); 
\begin{scope}[shift={(0.5, 0.3)}]
\draw [fill=white] (-0.5, -0.3) rectangle (0.5, 0.3);
\node at (0, 0) {\tiny $\mathbf{y}$};
\end{scope}
    \end{scope}
\draw [blue] (-0.5, -0.8)--(-0.5, 0) .. controls +(0, 0.6) and +(0,0.6).. (0.5, 0)--(0.5, -0.8);
\begin{scope}[shift={(0.5, -0.3)}]
\draw [fill=white] (-0.4, -0.3) rectangle (0.4, 0.3);
\node at (0, 0) {\tiny $wv_k $};
\end{scope}
\begin{scope}[shift={(1.4, 0.8)}]
\draw [fill=white] (-0.3, -0.3) rectangle (0.3, 0.3);
\node at (0, 0) {\tiny $\overline{v_k^*}$};  
\end{scope}
\end{tikzpicture}}}
+
\frac{1}{2}\vcenter{\hbox{\begin{tikzpicture}[scale=0.65]
     \draw [blue] (1.4, 2.3)--(1.4, -0.8);
    \begin{scope}[shift={(0,1.5)}]
    \draw [blue] (-0.5, 0.8)--(-0.5, 0) .. controls +(0, -0.6) and +(0,-0.6).. (0.5, 0)--(0.5, 0.8); 
    \end{scope}
\draw [blue] (-0.5, -0.8)--(-0.5, 0) .. controls +(0, 0.6) and +(0,0.6).. (0.5, 0)--(0.5, -0.8);
\begin{scope}[shift={(0.5, -0.3)}]
\draw [fill=white] (-0.55, -0.3) rectangle (0.55, 0.3);
\node at (0, 0) {\tiny $\mathbf{y}wv_k$};
\end{scope}
\begin{scope}[shift={(1.4, 0.8)}]
\draw [fill=white] (-0.3, -0.3) rectangle (0.3, 0.3);
\node at (0, 0) {\tiny $\overline{v_k^*}$};  
\end{scope}
\end{tikzpicture}}}
-\frac{1}{2}\sum_{j=1}^m e^{\beta a_j/2}\vcenter{\hbox{\begin{tikzpicture}[scale=0.65]
     \draw [blue] (1.4, 2.3)--(1.4, -0.8);
    \begin{scope}[shift={(0,1.5)}]
    \draw [blue] (-0.5, 0.8)--(-0.5, 0) .. controls +(0, -0.6) and +(0,-0.6).. (0.5, 0)--(0.5, 0.8); 
\begin{scope}[shift={(0.5, 0.3)}]
\draw [fill=white] (-0.3, -0.3) rectangle (0.3, 0.3);
\node at (0, 0) {\tiny $v_j$};
\end{scope}
    \end{scope}
\draw [blue] (-0.5, -0.8)--(-0.5, 0) .. controls +(0, 0.6) and +(0,0.6).. (0.5, 0)--(0.5, -0.8);
\begin{scope}[shift={(0.5, -0.3)}]
\draw [fill=white] (-0.55, -0.3) rectangle (0.55, 0.3);
\node at (0, 0) {\tiny $v_j^*wv_k$};
\end{scope}
\begin{scope}[shift={(1.4, 0.8)}]
\draw [fill=white] (-0.3, -0.3) rectangle (0.3, 0.3);
\node at (0, 0) {\tiny $\overline{v_k^*}$};  
\end{scope}
\end{tikzpicture}}}
-\frac{1}{2}\sum_{j=1}^m e^{-\beta a_j/2}\vcenter{\hbox{\begin{tikzpicture}[scale=0.65]
     \draw [blue] (1.4, 2.3)--(1.4, -0.8);
    \begin{scope}[shift={(0,1.5)}]
    \draw [blue] (-0.5, 0.8)--(-0.5, 0) .. controls +(0, -0.6) and +(0,-0.6).. (0.5, 0)--(0.5, 0.8); 
\begin{scope}[shift={(0.5, 0.3)}]
\draw [fill=white] (-0.3, -0.3) rectangle (0.3, 0.3);
\node at (0, 0) {\tiny $v_j^*$};
\end{scope}
    \end{scope}
\draw [blue] (-0.5, -0.8)--(-0.5, 0) .. controls +(0, 0.6) and +(0,0.6).. (0.5, 0)--(0.5, -0.8);
\begin{scope}[shift={(0.5, -0.3)}]
\draw [fill=white] (-0.55, -0.3) rectangle (0.55, 0.3);
\node at (0, 0) {\tiny $ v_jwv_k$};
\end{scope}
\begin{scope}[shift={(1.4, 0.8)}]
\draw [fill=white] (-0.3, -0.3) rectangle (0.3, 0.3);
\node at (0, 0) {\tiny $\overline{v_k^*}$};  
\end{scope}
\end{tikzpicture}}}
\end{align*}

By the fact that
\begin{align*}
\mathbf{y}v_k-v_k \mathbf{y} =  e^{-\beta a_k/2}v_k  -   e^{\beta a_k/2} v_k,
\end{align*}
we see that the left hand side of Equation \eqref{eq:inter} is 
\begin{align*}
& \frac{1}{2} (e^{-\beta a_k/2} -  e^{\beta a_k/2})
 \vcenter{\hbox{\begin{tikzpicture}[scale=0.65]
     \draw [blue] (1.4, 2.3)--(1.4, -0.8);
    \begin{scope}[shift={(0,1.5)}]
    \draw [blue] (-0.5, 0.8)--(-0.5, 0) .. controls +(0, -0.6) and +(0,-0.6).. (0.5, 0)--(0.5, 0.8); 
\begin{scope}[shift={(0.5, 0.3)}]
\draw [fill=white] (-0.3, -0.3) rectangle (0.3, 0.3);
\node at (0, 0) {\tiny $v_k$};
\end{scope}
    \end{scope}
\draw [blue] (-0.5, -0.8)--(-0.5, 0) .. controls +(0, 0.6) and +(0,0.6).. (0.5, 0)--(0.5, -0.8);
\begin{scope}[shift={(0.5, -0.3)}]
\draw [fill=white] (-0.3, -0.3) rectangle (0.3, 0.3);
\node at (0, 0) {\tiny $ w$};
\end{scope}
\begin{scope}[shift={(1.4, 0.8)}]
\draw [fill=white] (-0.3, -0.3) rectangle (0.3, 0.3);
\node at (0, 0) {\tiny $\overline{v_k^*}$};  
\end{scope}
\end{tikzpicture}}}  
- e^{-\beta a_k/2} 
\vcenter{\hbox{\begin{tikzpicture}[scale=0.65]
     \draw [blue] (1.4, 2.3)--(1.4, -0.8);
    \begin{scope}[shift={(0,1.5)}]
    \draw [blue] (-0.5, 0.8)--(-0.5, 0) .. controls +(0, -0.6) and +(0,-0.6).. (0.5, 0)--(0.5, 0.8); 
    \end{scope}
\draw [blue] (-0.5, -0.8)--(-0.5, 0) .. controls +(0, 0.6) and +(0,0.6).. (0.5, 0)--(0.5, -0.8);
\begin{scope}[shift={(0.5, -0.3)}]
\draw [fill=white] (-0.4, -0.3) rectangle (0.4, 0.3);
\node at (0, 0) {\tiny $ wv_j$};
\end{scope}
\begin{scope}[shift={(1.4, 0.8)}]
\draw [fill=white] (-0.3, -0.3) rectangle (0.3, 0.3);
\node at (0, 0) {\tiny $\overline{v_k^*}$};  
\end{scope}
\end{tikzpicture}}} \\
& + \frac{1}{2} (e^{-\beta a_k/2} -  e^{\beta a_k/2})
 \vcenter{\hbox{\begin{tikzpicture}[scale=0.65]
     \draw [blue] (1.4, 2.3)--(1.4, -0.8);
    \begin{scope}[shift={(0,1.5)}]
    \draw [blue] (-0.5, 0.8)--(-0.5, 0) .. controls +(0, -0.6) and +(0,-0.6).. (0.5, 0)--(0.5, 0.8); 
    \end{scope}
\draw [blue] (-0.5, -0.8)--(-0.5, 0) .. controls +(0, 0.6) and +(0,0.6).. (0.5, 0)--(0.5, -0.8);
\begin{scope}[shift={(0.5, -0.3)}]
\draw [fill=white] (-0.3, -0.3) rectangle (0.3, 0.3);
\node at (0, 0) {\tiny $ wv_k$};
\end{scope}
\begin{scope}[shift={(1.4, 0.8)}]
\draw [fill=white] (-0.3, -0.3) rectangle (0.3, 0.3);
\node at (0, 0) {\tiny $\overline{v_k^*}$};  
\end{scope}
\end{tikzpicture}}}  +   e^{\beta a_k/2}\vcenter{\hbox{\begin{tikzpicture}[scale=0.65]
     \draw [blue] (1.4, 2.3)--(1.4, -0.8);
    \begin{scope}[shift={(0,1.5)}]
    \draw [blue] (-0.5, 0.8)--(-0.5, 0) .. controls +(0, -0.6) and +(0,-0.6).. (0.5, 0)--(0.5, 0.8); 
\begin{scope}[shift={(0.5, 0.3)}]
\draw [fill=white] (-0.3, -0.3) rectangle (0.3, 0.3);
\node at (0, 0) {\tiny $v_k$};
\end{scope}
    \end{scope}
\draw [blue] (-0.5, -0.8)--(-0.5, 0) .. controls +(0, 0.6) and +(0,0.6).. (0.5, 0)--(0.5, -0.8);
\begin{scope}[shift={(0.5, -0.3)}]
\draw [fill=white] (-0.55, -0.3) rectangle (0.55, 0.3);
\node at (0, 0) {\tiny $w$};
\end{scope}
\begin{scope}[shift={(1.4, 0.8)}]
\draw [fill=white] (-0.3, -0.3) rectangle (0.3, 0.3);
\node at (0, 0) {\tiny $\overline{v_k^*}$};  
\end{scope}
\end{tikzpicture}}} \\
=&
-\frac{1}{2} (e^{-\beta a_k/2} + e^{\beta a_k/2})
 \vcenter{\hbox{\begin{tikzpicture}[scale=0.65]
     \draw [blue] (1.4, 2.3)--(1.4, -0.8);
    \begin{scope}[shift={(0,1.5)}]
    \draw [blue] (-0.5, 0.8)--(-0.5, 0) .. controls +(0, -0.6) and +(0,-0.6).. (0.5, 0)--(0.5, 0.8); 
    \end{scope}
\draw [blue] (-0.5, -0.8)--(-0.5, 0) .. controls +(0, 0.6) and +(0,0.6).. (0.5, 0)--(0.5, -0.8);
\begin{scope}[shift={(0.5, -0.3)}]
\draw [fill=white] (-0.3, -0.3) rectangle (0.3, 0.3);
\node at (0, 0) {\tiny $ wv_k$};
\end{scope}
\begin{scope}[shift={(1.4, 0.8)}]
\draw [fill=white] (-0.3, -0.3) rectangle (0.3, 0.3);
\node at (0, 0) {\tiny $\overline{v_k^*}$};  
\end{scope}
\end{tikzpicture}}}  + \frac{1}{2} (e^{-\beta a_k/2} + e^{\beta a_k/2}) \vcenter{\hbox{\begin{tikzpicture}[scale=0.65]
     \draw [blue] (1.4, 2.3)--(1.4, -0.8);
    \begin{scope}[shift={(0,1.5)}]
    \draw [blue] (-0.5, 0.8)--(-0.5, 0) .. controls +(0, -0.6) and +(0,-0.6).. (0.5, 0)--(0.5, 0.8); 
\begin{scope}[shift={(0.5, 0.3)}]
\draw [fill=white] (-0.3, -0.3) rectangle (0.3, 0.3);
\node at (0, 0) {\tiny $v_k$};
\end{scope}
    \end{scope}
\draw [blue] (-0.5, -0.8)--(-0.5, 0) .. controls +(0, 0.6) and +(0,0.6).. (0.5, 0)--(0.5, -0.8);
\begin{scope}[shift={(0.5, -0.3)}]
\draw [fill=white] (-0.3, -0.3) rectangle (0.3, 0.3);
\node at (0, 0) {\tiny $w$};
\end{scope}
\begin{scope}[shift={(1.4, 0.8)}]
\draw [fill=white] (-0.3, -0.3) rectangle (0.3, 0.3);
\node at (0, 0) {\tiny $\overline{v_k^*}$};  
\end{scope}
\end{tikzpicture}}}
\end{align*}
This show that the semigroup given by $\widehat{\cL}_0$ has the intertwining property.
It is worth noting that the extension is not unique.
The analysis of the generalized logarithmic Sobolev inequality will be pursued in a forthcoming paper.
\bibliographystyle{abbrv}
\bibliography{Markov}
\end{document}